\documentclass{CSML}

\def\dOi{11(4:13)2015}
\lmcsheading%
{\dOi}
{1--23}
{}
{}
{Sep.~16, 2014}
{Dec.~22, 2015}
{}

\ACMCCS{[{\bf Theory of computation}]: Models of computation;
  Semantics and reasoning---Program constructs; [{\bf Software and its
    engineering}]: Software notations and tools---General programming
  languages---Language features} 

\usepackage{amsmath,amssymb,amsthm}
\usepackage{stmaryrd}
\usepackage{url,hyperref}
\usepackage{xcolor}
\usepackage{soul}
\usepackage{etoolbox}
\usepackage{mathpartir}
\usepackage{upgreek}
\usepackage{prooftree}
\usepackage{xspace}
\usepackage{tikz}
\usetikzlibrary{calc}

\newif\ifdraft
\drafttrue
\draftfalse

\newcommand{\ie}{\emph{i.e.}}
\newcommand{\eg}{\emph{e.g.}}

\newcommand{\emphdef}[1]{\emph{#1}}

\newcommand{\set}[1]{\{#1\}}
\newcommand{\subst}[2]{\{#1/#2\}}

\newcommand{\treeref}[1]{\text{\upshape(\textsc{D}#1)}}

\newcommand{\mkkeyword}[1]{\mathtt{\KeywordStyle#1}}

\newcommand{\proofcase}[1]{\vskip1ex\noindent\fbox{#1}}
\newcommand{\rproofcase}[1]{\proofcase{\refrule{#1}}}
\newcommand{\eoe}{\hfill{$\blacksquare$}}

\ifdraft
\newcommand{\Luca}[1]{
  \marginpar{\CommentStyle\tiny\textbf{Luca}:~#1}
}
\else
\newcommand{\Luca}[1]{}
\fi

\newcommand{\ttparens}[1]{\ttlparen#1\ttrparen}

\newcommand{\ttbar}{\texttt{\upshape|}}
\newcommand{\ttunderscore}{\texttt{\upshape\char`_}}
\newcommand{\ttlparen}{\texttt{\upshape(}}
\newcommand{\ttrparen}{\texttt{\upshape)}}

\newcommand{\ttqmark}{\texttt{\upshape?}}
\newcommand{\ttemark}{\texttt{\upshape!}}
\newcommand{\ttstar}{\texttt{\upshape*}}
\newcommand{\ttdot}{\texttt{\upshape.}}
\newcommand{\ttcomma}{\texttt{\upshape,}}

\newcommand{\ttlbrack}{\texttt{\upshape[}}
\newcommand{\ttrbrack}{\texttt{\upshape]}}
\newcommand{\ttplus}{\texttt{\upshape+}}

\newenvironment{lines}[1][t]{
  \begin{array}[#1]{@{}l@{}}
}{
  \end{array}
}

\newcommand{\finalrevision}[1]{#1}
\newcommand{\revised}[1]{#1}

\newcommand{\pt}[1]{\[#1\]}

\newcommand{\rulename}[1]{%
  \text{\small[\textsc{#1}]}%
}

\newcommand{\defrule}[1]{%
  \hypertarget{rule:#1}{%
    \rulename{#1}%
  }%
}

\newcommand{\refrule}[1]{%
  \hyperlink{rule:#1}{%
    \rulename{#1}%
  }%
}

\definecolor{mymagenta}{rgb}{0.5,0,0.5}
\definecolor{mygreen}{rgb}{0,0.4,0}
\definecolor{myblue}{rgb}{0,0,0.6}
\definecolor{myred}{rgb}{0.4,0,0}
\definecolor{hlcolor}{rgb}{1,0.95,0}

\newcommand{\var}{\varX}
\newcommand{\varX}{x}
\newcommand{\varY}{y}
\newcommand{\varZ}{z}

\newcommand{\tvar}{\tvarA}
\newcommand{\tvarA}{\alpha}
\newcommand{\tvarB}{\beta}
\newcommand{\tvarC}{\gamma}
\newcommand{\tvarD}{\delta}

\newcommand{\Process}{\ProcessP}
\newcommand{\ProcessP}{P}
\newcommand{\ProcessQ}{Q}
\newcommand{\ProcessR}{R}

\newcommand{\Channel}{\ChannelA}
\newcommand{\ChannelA}{a}
\newcommand{\ChannelB}{b}
\newcommand{\ChannelC}{c}
\newcommand{\ChannelD}{d}

\newcommand{\Name}{\NameU}
\newcommand{\NameU}{u}
\newcommand{\NameV}{v}

\newcommand{\Term}{\TypeExpr}

\newcommand{\Type}{\TypeT}
\newcommand{\TypeT}{t}
\newcommand{\TypeS}{s}

\newcommand{\Label}{\ell}

\newcommand{\assignment}{substitution\xspace}
\newcommand{\assignments}{substitutions\xspace}
\newcommand{\Solution}{\sigma}

\newcommand{\SessionType}{\SessionTypeT}
\newcommand{\SessionTypeT}{T}
\newcommand{\SessionTypeS}{S}

\newcommand{\stin}[2]{{\ttqmark}{#1}\ttdot#2}
\newcommand{\stout}[2]{{\ttemark}{#1}\ttdot#2}

\newcommand{\Expression}{\ExpressionE}
\newcommand{\ExpressionE}{\mathsf{e}}
\newcommand{\ExpressionF}{\mathsf{f}}

\newcommand{\Value}{\ValueV}
\newcommand{\ValueV}{\mathsf{v}}
\newcommand{\ValueW}{\mathsf{w}}

\newcommand{\Nil}{\text{\KeywordStyle\texttt{Nil}}}
\newcommand{\Cons}{\mathbin{\KeywordStyle\texttt{Cons}}}

\newcommand{\PairX}[2]{#1\ttcomma#2}
\newcommand{\Pair}[2]{\ttparens{\PairX{#1}{#2}}}
\newcommand{\Project}[2]{#1\ifblank{#2}{}{\ttparens{#2}}}
\newcommand{\FirstX}[1]{\Project{\mkkeyword{fst}}{#1}}
\newcommand{\SecondX}[1]{\Project{\mkkeyword{snd}}{#1}}
\newcommand{\First}[1][]{\FirstX{#1}}
\newcommand{\Second}[1][]{\SecondX{#1}}

\newcommand{\Inject}[2]{#1\ifblank{#2}{}{\ttparens{#2}}}
\newcommand{\Left}[1][]{\Inject{\mkkeyword{inl}}{#1}}
\newcommand{\Right}[1][]{\Inject{\mkkeyword{inr}}{#1}}

\newcommand{\EmptyEnv}{\emptyset}

\newcommand{\idle}{\mkkeyword{idle}}
\newcommand{\bang}{{\ttstar}}
\newcommand{\new}[1]{\mkkeyword{new}~#1~\mkkeyword{in}~}
\newcommand{\send}[2]{#1\ttemark#2}
\newcommand{\receive}[3][\ttdot]{
  #2\ttqmark\ttlparen#3\ttrparen#1
}
\newcommand{\parop}{\mathbin\ttbar}

\newcommand{\Split}[4]{\mkkeyword{split}~#1~\mkkeyword{as}~\PairX{#2}{#3}~\mkkeyword{in}~#4}

\newcommand{\CaseV}[2]{
  \begin{array}[t]{@{}l@{~\Rightarrow~}l@{}}
    \multicolumn{2}{@{}l@{}}{
      \mkkeyword{case}~#1~\mkkeyword{of}
    }
    \\
    #2
  \end{array}
}

\newcommand{\CaseShort}[4]{
  \mkkeyword{case}{~#1~}
  \set{
    \Inject{#2}{#3_{#2}}\Rightarrow#4_{#2}
  }_{#2=\Left,\Right}
}
\newcommand{\IfX}[3]{
  \mkkeyword{if}#1\mkkeyword{then}#2\mkkeyword{else}#3
}
\newcommand{\If}[3]{\IfX{~#1~}{~#2~}{~#3}}

\newcommand{\uvar}{\varrho}

\newcommand{\Use}{\kappa}

\newcommand{\tint}{\mkkeyword{int}}
\newcommand{\tbool}{\mkkeyword{bool}}
\newcommand{\tsum}{\oplus}

\newcommand{\tand}{+}

\newcommand{\toper}{\odot}

\newcommand{\tchan}[3]{\ttlbrack#1\ttrbrack^{#2,#3}}

\newcommand{\TypeExpr}{\TypeExprT}
\newcommand{\TypeExprT}{\mathsf{T}}
\newcommand{\TypeExprS}{\mathsf{S}}

\newcommand{\UseExpr}{\UseExprU}
\newcommand{\UseExprU}{\mathsf{U}}
\newcommand{\UseExprV}{\mathsf{V}}

\newcommand{\Env}{\Upgamma}
\newcommand{\EnvX}{\Updelta}

\newcommand{\defined}[2]{\mathsf{def}_{#2}(#1)}
\newcommand{\undefined}[2]{\mathsf{undef}_{#2}(#1)}

\newcommand{\dom}{\mathsf{dom}}
\newcommand{\expr}{\mathsf{expr}}

\newcommand{\fn}{\mathsf{fn}}
\newcommand{\bn}{\mathsf{bn}}
\newcommand{\co}[1]{\overline{#1}}
\newcommand{\enc}[1]{\llbracket#1\rrbracket}

\newcommand{\mktvar}{\mathsf{t}}

\newcommand{\combineop}{\sqcup}
\newcommand{\mergeop}{\sqcap}

\newcommand{\combinenv}[4]{#1 \combineop #2 \leadsto #3; #4}
\newcommand{\mergenv}[4]{#1 \mergeop #2 \leadsto #3; #4}

\newcommand{\combine}{+}

\newcommand{\instance}[2]{
  \mathsf{instance}
  \ifblank{#1#2}{}{(#1,#2)}
}

\newcommand{\crep}[3]{\mathsf{crep}_{#3}(#1,#2)}

\newcommand{\Constraint}{\varphi}
\newcommand{\Constraints}{\mathcal{C}}

\newcommand{\con}[1]{\mathrel{\hat{#1}}}
\newcommand{\can}[2][\Constraints]{#2_{#1}}

\newcommand{\ceq}{\mathrel{\hat=}}

\newcommand{\eqdef}{\stackrel{\text{\tiny\sffamily\upshape def}}{=}}

\newcommand{\lred}[1][]{\stackrel{#1}{\longrightarrow}}

\newcommand{\atype}[1]{#1_{\mathit{type}}}

\newcommand{\rrel}{\mathrel{\mathcal{R}}}

\newcommand{\crel}{\atype{\mathcal{C}}}
\newcommand{\orel}{\sqsubseteq}
\newcommand{\prel}{\leq}

\newcommand{\eval}{\mathrel{\downarrow}}

\newcommand{\sle}{\preccurlyeq}
\newcommand{\seq}{\equiv}

\newcommand{\wte}[4][]{#2 \vdash\ifblank{#1}{}{_{#1}} #3 : #4}
\newcommand{\wtp}[3][]{#2 \vdash\ifblank{#1}{}{_{#1}} #3}

\newcommand{\rte}[4]{#1 : #2 \blacktriangleright #3; #4}
\newcommand{\rtp}[3]{#1 \blacktriangleright #2; #3}

\newcommand{\ded}[2]{#1 \Vdash #2}

\newcommand{\unlimited}[1]{\mathsf{un}(#1)}

\newcommand{\meq}{=}

\newcommand{\mall}{\sim}

\makeatletter
\newtheorem*{rep@theorem}{\rep@title}
\newcommand{\newreptheorem}[2]{%
\newenvironment{rep#1}[1]{%
 \def\rep@title{#2 \ref{##1}}%
 \begin{rep@theorem}}%
 {\end{rep@theorem}}}
\makeatother

\theoremstyle{definition}
\newtheorem{definition}{Definition}[section]
\newtheorem{example}[definition]{Example}

\theoremstyle{plain}
\newtheorem{proposition}[definition]{Proposition}
\newtheorem{lemma}[definition]{Lemma}
\newtheorem{theorem}[definition]{Theorem}
\newreptheorem{theorem}{Theorem}

\theoremstyle{remark}
\newtheorem{remark}[definition]{Remark}

\newcommand{\KeywordStyle}{\color{myblue}}
\newcommand{\CommentStyle}{\color{magenta}}

\begin{document}

\title[Type Reconstruction for the Linear $\pi$-Calculus]{
  Type Reconstruction for the Linear $\pi$-Calculus
  \\
  with Composite Regular Types\rsuper*
}

\author[L.~Padovani]{Luca Padovani}
\address{Dipartimento di Informatica, Universit\`a di Torino, Italy}
\email{luca.padovani@di.unito.it}
\thanks{This work has been supported by ICT COST Action IC1201 BETTY, MIUR
project CINA, Ateneo/CSP project SALT, and the bilateral project
RS13MO12 DART}

\keywords{linear pi-calculus, composite \revised{regular} types, shared access to
  data structures with linear values, type reconstruction}

\titlecomment{{\lsuper*}A preliminary version of this paper~\cite{Padovani14A}
  appears in the proceedings of the 17th International Conference on
  Foundations of Software Science and Computation Structures
  (FoSSaCS'14).}

\begin{abstract}
  We extend the linear $\pi$-calculus with composite regular types
  in such a way that data containing linear values can be shared among several
  processes, if there is no overlapping access to such values.
  We describe a type reconstruction algorithm for the extended type system and
  discuss some practical aspects of its implementation.
\end{abstract}

\maketitle

\section{Introduction}
\label{sec:introduction}

The linear $\pi$-calculus~\cite{KobayashiPierceTurner99} is a formal model of
communicating processes that distinguishes between \emph{unlimited} and
\emph{linear} channels. Unlimited channels can be used without restrictions,
whereas linear channels can be used for one communication only. Despite this
seemingly severe restriction, there is evidence that a significant portion of
communications in actual systems take place on linear
channels~\cite{KobayashiPierceTurner99}. It has also been shown that
structured communications can be encoded using linear channels and a
continuation-passing style~\cite{Kobayashi02b,DardhaGiachinoSangiorgi12}.  The
interest in linear channels has solid motivations: linear channels are
efficient to implement, they enable important
optimizations~\cite{IgarashiKobayashi97,Igarashi97,KobayashiPierceTurner99},
and communications on linear channels enjoy important properties such as
interference freedom and partial
confluence~\cite{NestmannSteffen97,KobayashiPierceTurner99}.  It follows that
understanding whether a channel is used linearly or not has a primary impact
in the analysis of systems of communicating processes.

Type reconstruction is the problem of \emph{inferring} the type of entities
used in an unannotated (\ie, untyped) program. In the case of the linear
$\pi$-calculus, the problem translates into understanding whether a channel is
linear or unlimited, and determining the type of messages sent over the
channel.
This problem has been addressed and solved in~\cite{IgarashiKobayashi00}.
The goal of our work is the definition of a type reconstruction
algorithm for the linear $\pi$-calculus extended with pairs, disjoint
sums, and \revised{possibly infinite} types.  These
\revised{features}, albeit standard, gain relevance and combine in
non-trivial ways with the features of the linear $\pi$-calculus. We
explain why this is the case in the rest of this section.

\newcommand{\succc}{\mathtt{succ}}
\newcommand{\printc}{\mathtt{print}}

The term below
\begin{equation}
\label{eq:succ}
  \bang\receive\succc{\PairX\varX\varY}
  \send\varY{\ttparens{\varX+1}}
  \parop
  \new\Channel{
    \ttparens{
      \send\succc{\Pair{39}\Channel}
      \parop
      \receive\Channel\varZ
      \send\printc\varZ
    }
  }
\end{equation}
models a program made of a persistent service (the $\bang$-prefixed process
waiting for messages on channel $\succc$) that computes the successor of a
number and a client (the $\mkkeyword{new}$-scoped process) that invokes the
service and prints the result of the invocation. Each message sent to the
service is a pair made of the number $\varX$ and a continuation channel
$\varY$ on which the service sends the result of the computation back to the
client.
There are three channels in this program, $\succc$ for invoking the service,
$\printc$ for printing numbers, and a private channel $\Channel$ which is used
by the client for receiving the result of the invocation.
In the linear $\pi$-calculus, \revised{types keep track of how each
  occurrence of a channel is being used. For example, the above
  program is well typed in the environment
\[
  \printc : \tchan\tint01, \succc : \tchan{\tint \times \tchan\tint01}\omega1
\]
where the type of $\printc$ indicates not only the type of messages
sent over the channel ($\tint$ in this case), but also that $\printc$
is never used for input operations (the $0$ annotation) and is used
once for one output operation (the $1$ annotation).}

The type of $\succc$ \revised{indicates that} messages sent over
$\succc$ are pairs of type $\tint \times \tchan\tint01$ -- the service
performs exactly one output operation on the channel $\varY$ which is
the second component of the pair -- \revised{and that $\succc$ is used
  for an unspecified number of input operations (the $\omega$
  annotation) and exactly one output operation (the $1$ annotation).}
\revised{Interestingly, the overall type of $\succc$ can be expressed
  as the combination of two slightly different types describing how
  each occurrence of $\succc$ is being used by the program:}
the leftmost occurrence of $\succc$ is used according to the type
$\tchan{\tint \times \tchan\tint01}\omega0$ (arbitrary inputs, no
outputs), while the rightmost occurrence of $\succc$ is used according
to the type $\tchan{\tint \times \tchan\tint01}01$ (no inputs, one
output). \revised{Following~\cite{KobayashiPierceTurner99}, we capture
  the overall use of a channel by means of a \emphdef{combination
    operator} $+$ on types such that, for example,}
\[
  \tchan{\tint \times \tchan\tint01}\omega0
  +
  \tchan{\tint \times \tchan\tint01}01
  =
  \tchan{\tint \times \tchan\tint01}\omega1  
\]
\revised{Concerning the restricted channel $\Channel$,} its rightmost
occurrence \revised{is used according to the} type $\tchan\tint10$,
since there $\Channel$ is used for one input of an integer number; the
occurrence of $\Channel$ in $\Pair{39}\Channel$ is in a message sent
on $\succc$, and we have already argued that the service uses this
channel according to the type $\tchan\tint01$; the type of the
leftmost, binding occurrence of $\Channel$ is the combination of these
two types, namely:
\[
  \tchan\tint01 + \tchan\tint10 = \tchan\tint11
\]
The type of $\Channel$ indicates that the program performs exactly one
input and exactly one output on $\Channel$, hence $\Channel$ is a
linear channel. Since $\Channel$ is restricted in the program, even if
the program is extended with more processes, it is not possible to
perform operations on $\Channel$ other than the ones we have tracked
in its type.

The key ingredient in the discussion above is the notion of type
combination~\cite{KobayashiPierceTurner99,IgarashiKobayashi00,SangiorgiWalker01},
which allows us to gather the overall number of input/output
operations performed on a channel. We now discuss how type combination
extends to composite \revised{and possibly infinite} types, which is
the main novelty of the present work.

So far we have taken for granted the ability to perform \emph{pattern
  matching} on the message received by the service on $\succc$ and to assign
distinct names, $\varX$ and $\varY$, to the components of the pair being
analyzed. Pattern matching is usually compiled using more basic
operations. For example, in the case of pairs these operations are the
$\First$ and $\Second$ projections that respectively extract the first and the
second component of the pair. So, a low-level modeling of the successor
service that uses $\First$ and $\Second$ could look like this:
\begin{equation}
  \label{eq:succ_fst_snd}
  \bang\receive\succc{p}
  \send{\Second[p]}{
    \ttparens{
      \First[p] + 1
    }
  }
\end{equation}

This version of the service is operationally equivalent to the
previous one, but from the viewpoint of typing there is an interesting
difference: in \eqref{eq:succ} the two components of the pair are
given distinct names $\varX$ and $\varY$ and each name is used
\emph{once} in the body of the service; in \eqref{eq:succ_fst_snd}
there is only one name $p$ for the whole pair which is projected
\emph{twice} in the body of the service. Given that each projection
accesses only one of the two components of the pair and ignores the
other, we can argue that the occurrence of $p$ in $\Second[p]$
\revised{is used according to the type} $\tint \times \tchan\tint01$
(the $1$ \revised{annotation} reflects the fact that the second
component of $p$ is a channel used for an output operation) whereas
the occurrence of $p$ in $\First[p]$ \revised{is used according to the
  type} $\tint \times \tchan\tint00$ (the second component of $p$ is
not used).  The key idea, then, is that we can
extend the type combination operator $+$ \revised{component-wise} to
product types to express the overall type of $p$ as the combination of
these two types:
\[
  (\tint \times \tchan\tint01)
  +
  (\tint \times \tchan\tint00)
  =
  (\tint + \tint) \times (\tchan\tint01 + \tchan\tint00)
  =
  \tint \times \tchan\tint01
\]
\newcommand{\Producer}{\mathit{P}}
\newcommand{\Consumer}{\mathit{C}}
\newcommand{\Print}{\mathit{Print}}
\newcommand{\TypeList}{\Type_{\mathit{list}}}%
\newcommand{\TypeOdd}{\Type_{\mathit{odd}}}%
\newcommand{\TypeEven}{\Type_{\mathit{even}}}%
\newcommand{\DoP}{\mathit{Odd}}%
\newcommand{\SkipP}{\mathit{Even}}%
\newcommand{\evenc}{\mathtt{even}}
\newcommand{\oddc}{\mathtt{odd}}
\newcommand{\sumc}{\mathtt{sum}}
\newcommand{\acc}{\mathit{acc}}
According to \revised{the result of such combination}, the second
component of $p$ is effectively used only once despite the multiple
syntactic occurrences of $p$.

The extension of type combination to products carries over to
\revised{disjoint sums and also to infinite types} as well. To
illustrate, consider the type $\TypeList$ satisfying the
\revised{equality}
\[
  \TypeList = \Nil \oplus \Cons\ttparens{\tchan\tint{1}{0} \times \TypeList}
\]
which is the disjoint sum between $\Nil$, the type of empty lists, and
$\Cons\ttparens{\tchan\tint{1}{0} \times \TypeList}$, the type of
non-empty lists with head of type $\tchan\tint{1}{0}$ and tail of type
$\TypeList$ \revised{(we will see shortly that there is a unique type
  $\TypeList$ satisfying the above equality relation)}.
Now, $\TypeList$ can be expressed as the combination $\TypeOdd +
\TypeEven$, where $\TypeOdd$ and $\TypeEven$ are the types that
satisfy the \revised{equalities}
\begin{equation}
\label{eq:oddeven}
  \TypeOdd = \Nil \oplus \Cons\ttparens{\tchan\tint{1}{0} \times \TypeEven}
  \text{\quad and\quad}
  \TypeEven = \Nil \oplus \Cons\ttparens{\tchan\tint{0}{0} \times \TypeOdd}
\end{equation}
\revised{(again, there are unique $\TypeOdd$ and $\TypeEven$ that
  satisfy these equalities, see Section~\ref{sec:types}).}

In words, $\TypeOdd$ is the type of lists of channels in which each
channel in an odd-indexed position is used for one input, while
$\TypeEven$ is the type of lists of channel in which each channel in
an even-indexed position is used for one input.
The reason why this particular decomposition of $\TypeList$ could be
interesting is that it enables the sharing of a list containing linear
channels among two processes, if we know that one process uses the
list according to the type $\TypeOdd$ and the other process uses the
same list according to the type $\TypeEven$.
\revised{For example, the process $\ProcessR$ defined below
\[
\begin{array}{@{}rcr@{}l@{}}
  \ProcessP & \eqdef &
  \bang\receive\oddc{l\ttcomma\acc\ttcomma r}&
  \CaseV{~l~}{
    \quad \Nil & \send{r}\acc
    \\
    \quad \Cons\Pair\varX{l'} & \receive\varX\varY \send\evenc{\ttparens{l'\ttcomma\ttparens{\acc+\varY}\ttcomma r}}
  }
  \\\\
  \ProcessQ & \eqdef &
  \bang\receive\evenc{l\ttcomma\acc\ttcomma r}&
  \CaseV{~l~}{
    \quad \Nil & \send{r}\acc
    \\
    \quad \Cons\Pair\varX{l'} & \send\oddc{\ttparens{l'\ttcomma\acc\ttcomma r}}
  }
  \\\\
  \ProcessR & \eqdef & 
  \multicolumn{2}{l}{
  \ProcessP \parop \ProcessQ \parop
  \new{\ChannelA\ttcomma\ChannelB}
  \ttparens{
    \send\oddc{\ttparens{l\ttcomma 0\ttcomma\ChannelA}}
    \parop
    \send\evenc{\ttparens{l\ttcomma 0\ttcomma\ChannelB}}
    \parop
    \receive\ChannelA\varX
    \receive\ChannelB\varY
    \send{r}{\ttparens{\varX + \varY}}
  }
  }
\end{array}
\]
uses each channel in a list $l$ for receiving a number, sums all such
numbers together, and sends the result on another channel
$r$. However, instead of scanning the list $l$ sequentially in a
single thread, $\ProcessR$ spawns two parallel threads (defined by
$\ProcessP$ and $\ProcessQ$) that share the \emph{very same list} $l$:
the first thread uses only the odd-indexed channels in $l$, whereas
the second thread uses only the even-indexed channels in $l$; the
(partial) results obtained by these two threads are collected by
$\ProcessR$ on two locally created channels $\ChannelA$ and
$\ChannelB$; the overall result is eventually sent on $r$.}
We are then able to deduce that $\ProcessR$ makes full use of the
channels in $l$, namely that $l$ has type $\TypeList$, even though the
list as a whole is simultaneously accessed by two parallel threads. In
general, we can see that the extension of type combination to
composite, \revised{potentially infinite} types is an effective tool
that fosters the parallelization of programs and allows composite data
structures containing linear values to be safely shared by a pool of
multiple processes, if there is enough information to conclude that
each linear value is accessed by exactly one of the processes in the
pool.

Such detailed reasoning on the behavior of programs comes at the price
of a more sophisticated definition of type combination. This brings us
back to the problem of type reconstruction. The reconstruction
algorithm described in this article is able to infer the types
$\TypeOdd$ and $\TypeEven$ of the messages accepted \revised{by
  $\ProcessP$ and $\ProcessQ$} by looking at the structure of these
two processes and of understanding that the overall type of $l$ in
$\ProcessR$ is $\TypeList$, namely that every channel in $l$ is used
exactly once.

\subsection*{Related work}
Linear type systems with composite types have been discussed
in~\cite{Igarashi97,IgarashiKobayashi97} for the linear $\pi$-calculus and
in~\cite{TurnerWadlerMossin95} for a functional language. In these works,
however, every structure that contains linear values becomes linear itself
(there are a few exceptions for specific types~\cite{Kobayashi06} or relaxed
notions of linearity~\cite{Kobayashi99}).

The original type reconstruction algorithm for the linear
$\pi$-calculus is described in~\cite{IgarashiKobayashi00}. Our work
extends~\cite{IgarashiKobayashi00} to composite and \revised{infinite}
types. Unlike~\cite{IgarashiKobayashi00}, however, we do not deal with
structural subtyping, whose integration into our type reconstruction
algorithm is left for future work. The type reconstruction algorithm
in~\cite{IgarashiKobayashi00} and the one we present share a common
structure in that they both comprise constraint generation and
constraint resolution phases. The main difference concerns the fact
that we have to deal with constraints expressing the combination of
yet-to-be-determined types, whereas in~\cite{IgarashiKobayashi00}
non-trivial type combinations only apply to channel types. This
allows~\cite{IgarashiKobayashi00} to use an efficient constraint
resolution algorithm based on unification. In our setting, the
\revised{presence of infinite} types hinders the use of unification,
and in some cases the resolution algorithm may conservatively
approximate the outcome in order to ensure proper termination.

Session types~\cite{Honda93,HondaVasconcelosKubo98} describe
\emph{linearized} channels, namely channels that can be used for
multiple communications, but only in a sequential way. There is a
tight connection between linear and linearized channels: as shown
in~\cite{Kobayashi02b,DemangeonHonda11,DardhaGiachinoSangiorgi12,Dardha14},
linearized channels can be encoded in the linear $\pi$-calculus. A
consequence of this encoding is that the type reconstruction algorithm
we present in this article can be used for inferring \revised{possibly
  infinite} session types (we will see an example of this feature in
Section~\ref{sec:examples}). The task of reconstructing session types
directly has been explored in~\cite{Mezzina08}, but for finite types
only.

\subsection*{Structure of the paper}
We present the calculus in Section~\ref{sec:language} and the type system in
Section~\ref{sec:types}. The type reconstruction algorithm consists of a
constraint generation phase (Section~\ref{sec:generator}) and a constraint
resolution phase (Section~\ref{sec:solver}).
We discuss some important issues related to the implementation of the
algorithm in Section~\ref{sec:implementation} and a few more elaborate
examples in Section~\ref{sec:examples}.
Section~\ref{sec:conclusion} concludes and hints at some ongoing and future
work.
Proofs of the results in Sections~\ref{sec:types}
and~\ref{sec:generator} are in Appendixes~\ref{sec:extra_types}
and~\ref{sec:extra_generator}, respectively.
\revised{Appendix~\ref{sec:extra_solver} illustrates a few typing
  derivations of examples discussed in Section~\ref{sec:solver}.}
A proof-of-concept implementation of the algorithm is available on the
author's home page.

\section{The $\pi$-calculus with data types}
\label{sec:language}

In this section we define the syntax and operational semantics of the
formal language we work with, which is an extension of the
$\pi$-calculus featuring \finalrevision{base} and composite data types
and a pattern matching construct.

\subsection{Syntax}
Let us introduce some notation first.
We use integer numbers $m$, $n$, $\dots$, a countable set of
\emphdef{channels} $\ChannelA$, $\ChannelB$, $\dots$, and a countable
set of \emphdef{variables} $\varX$, $\varY$, $\dots$ which is disjoint
from the set of channels;
\emphdef{names} $\NameU$, $\NameV$, $\dots$ are either channels or
variables.

\begin{table}
\framebox[\textwidth]{
\begin{math}
\displaystyle
\begin{array}{@{}c@{~}c@{}}
\begin{array}[t]{@{}rcl@{\quad}l@{}}
  \ProcessP, \ProcessQ & ~::=~ & & \textbf{Process} \\
  &   & \idle & \text{(idle process)} \\
  & | & \receive\Expression\var\Process & \text{(input)} \\
  & | & \send\ExpressionE\ExpressionF & \text{(output)} \\
  & | & \ProcessP \parop \ProcessQ & \text{(parallel composition)} \\
  & | & \bang\Process & \text{(process replication)} \\
  & | & \new\ChannelA \Process & \text{(channel restriction)} \\
  & | & \CaseShort\Expression{i}\varX\ProcessP & \text{(pattern matching)} \\
  \\
  \ExpressionE, \ExpressionF & ::= & & \textbf{Expression} \\
  &   & n & \text{(integer constant)} \\
  & | & \Name & \text{(name)} \\
  & | & \Pair\ExpressionE\ExpressionF & \text{(pair)} \\
  & | & \First[\Expression] & \text{(first projection)} \\
  & | & \Second[\Expression] & \text{(second projection)} \\
  & | & \Left[\Expression] & \text{(left injection)} \\
  & | & \Right[\Expression] & \text{(right injection)} \\
\end{array}
\end{array}
\end{math}
}
\caption{\label{tab:syntax} Syntax of processes and expressions.}
\end{table}

The syntax of expressions and processes is given in
Table~\ref{tab:syntax}.
\emphdef{Expressions} $\ExpressionE$, $\ExpressionF$, $\dots$ are
either integers, names, pairs $\Pair\ExpressionE\ExpressionF$ of
expressions, the $i$-th projection of an expression
$\Project{i}\Expression$ where $i\in\set{\First,\Second}$, or the
injection $\Inject{i}\Expression$ of an expression $\Expression$ using
the constructor $i\in\set{\Left,\Right}$.
\revised{Using projections $\First$ and $\Second$ instead of a pair
  splitting construct, as found for instance
  in~\cite{SangiorgiWalker01,Pierce04}, is somewhat unconventional,
  but helps us highlighting some features of our type system. We will
  discuss some practical aspects of this choice in
  Section~\ref{sec:splitting}.}

%
\emphdef{Values} $\ValueV$, $\ValueW$, $\dots$ are expressions without
variables and occurrences of the projections $\First$ and $\Second$.

\emphdef{Processes} $\ProcessP$, $\ProcessQ$, $\dots$ comprise and
extend the standard constructs of the asynchronous $\pi$-calculus.
The $\idle$ process performs no action;
the input process $\receive\Expression\var\Process$ waits for a
message $\Value$ from the channel denoted by $\Expression$
\finalrevision{and continues as $\Process$ where $\var$ has been
  replaced by $\Value$};
the output process $\send\ExpressionE\ExpressionF$ sends the value
resulting from the evaluation of $\ExpressionF$ on the channel
resulting from the evaluation of $\ExpressionE$;
the composition $\ProcessP \parop \ProcessQ$ executes $\ProcessP$ and
$\ProcessQ$ in parallel;
the replication $\bang\Process$ denotes infinitely many copies of
$\Process$ executing in parallel;
the restriction $\new\Channel\Process$ creates a new channel
$\Channel$ with scope $\Process$.
In addition to these, we include a pattern matching construct
$\CaseShort\Expression{i}\var\Process$ which evaluates $\Expression$
to a value of the form $\Inject{i}\Value$ for some
$i\in\set{\Left,\Right}$, binds $\Value$ to $\var_i$ and continues as
$\Process_i$.
The notions of \emphdef{free names} $\fn(\Process)$ and \emphdef{bound
  names} $\bn(\Process)$ of $\Process$ are as expected, recalling that
$\CaseShort\Expression{i}\var\Process$ binds $\var_i$ in $\Process_i$.
We identify processes modulo renaming of bound names and we write
$\Expression\subst\Value\var$ and $\Process\subst\Value\var$ for the
capture-avoiding substitutions of $\Value$ for the free occurrences of $\var$
in $\Expression$ and $\Process$, respectively.
\revised{Occasionally, we omit $\idle$ when it is guarded by a prefix.}

\begin{table}
\framebox[\columnwidth]{
\begin{math}
\displaystyle
\begin{array}[t]{@{}c@{}}
  \inferrule[\defrule{s-par 1}]{}{
    \idle \parop \Process \seq \Process
  }
  \qquad
  \inferrule[\defrule{s-par 2}]{}{
    \ProcessP \parop \ProcessQ \seq \ProcessQ \parop \ProcessP
  }
  \qquad
  \inferrule[\defrule{s-par 3}]{}{
    \ProcessP \parop (\ProcessQ \parop \ProcessR)
    \seq
    (\ProcessP \parop \ProcessQ) \parop \ProcessR
  }
  \qquad
  \inferrule[\defrule{s-rep}]{}{
    *\Process \sle *\Process \parop \Process
  }
  \\\\
  \inferrule[\defrule{s-res 1}]{}{
    \new\ChannelA
    \new\ChannelB
    \Process
    \equiv
    \new\ChannelB
    \new\ChannelA
    \Process
  }
  \qquad
  \inferrule[\defrule{s-res 2}]{
    \ChannelA \not\in \fn(\ProcessQ)
  }{
    (\new\ChannelA \ProcessP) \parop \ProcessQ
    \equiv
    \new\ChannelA (\ProcessP \parop \ProcessQ)
  }
\end{array}
\end{math}
}
  \caption{\label{tab:congruence} Structural pre-congruence for processes.}
\end{table}

\begin{table}[t]
\framebox[\columnwidth]{
\begin{math}
\displaystyle
\begin{array}[t]{@{}c@{}}
  \inferrule[\defrule{e-int}]{}{
    n \eval n
  }
  \qquad
  \inferrule[\defrule{e-chan}]{}{
    \Channel \eval \Channel
  }
  \\\\
  \inferrule[\defrule{e-pair}]{
    \Expression_i \eval \Value_i~{}^{(i=1,2)}
  }{
    \Pair{\Expression_1}{\Expression_2}
    \eval
    \Pair{\Value_1}{\Value_2}
  }
  \qquad
  \inferrule[\defrule{e-fst}]{
    \Expression \eval \Pair\ValueV\ValueW
  }{
    \First[\Expression] \eval \ValueV
  }
  \qquad
  \inferrule[\defrule{e-snd}]{
    \Expression \eval \Pair\ValueV\ValueW
  }{
    \Second[\Expression] \eval \ValueW
  }
  \qquad
  \inferrule[\defrule{e-inr}, \defrule{e-inl}]{
    \Expression \eval \Value
    \\
    k\in\set{\Left,\Right}
  }{
    \Inject{k}{\Expression} \eval \Inject{k}{\Value}
  }
  \\\\
  \hline
  \\
  \inferrule[\defrule{r-comm}]{
    \ExpressionE_i \eval \Channel~{}^{(i=1,2)}
    \\
    \ExpressionF \eval \Value
  }{
    \send{\ExpressionE_1}{\ExpressionF}
    \parop
    \receive{\ExpressionE_2}{\var}\ProcessQ
    \lred[\ChannelA]
    \ProcessQ\subst{\Value}{\var}
  }
  \qquad
  \inferrule[\defrule{r-case}]{
    \Expression \eval \Inject{k}\Value
    \\
    k\in\set{\Left,\Right}
  }{
    \CaseShort\Expression{i}\var\Process
    \lred[\tau]
    \ProcessP_k\subst\Value{\varX_k}
  }
  \\\\
  \inferrule[\defrule{r-par}]{
    \ProcessP \lred[\Label] \ProcessP'
  }{
    \ProcessP \parop \ProcessQ \lred[\Label] \ProcessP' \parop
    \ProcessQ
  }
  \qquad
  \inferrule[\defrule{r-new 1}]{
    \ProcessP
    \lred[\Channel]
    \ProcessQ
  }{
    \new\Channel\ProcessP
    \lred[\tau]
    \new\Channel\ProcessQ
  }
  \qquad
  \inferrule[\defrule{r-new 2}]{
    \ProcessP
    \lred[\Label]
    \ProcessQ
    \\
    \Label\ne\Channel
  }{
    \new\Channel\ProcessP
    \lred[\Label]
    \new\Channel\ProcessQ
  }
  \\\\
  \inferrule[\defrule{r-struct}]{
    \ProcessP \sle \ProcessP'
    \\
    \ProcessP' \lred[\Label] \ProcessQ'
    \\
    \ProcessQ' \sle \ProcessQ
  }{
    \ProcessP \lred[\Label] \ProcessQ
  }
\end{array}
\end{math}
}
\caption{\label{tab:reduction} Evaluation of expressions and reduction of processes.}
\end{table}

\subsection{Operational semantics}
The operational semantics of the language is defined in terms of a
structural pre-congruence relation for processes, an evaluation
relation for expressions, and a reduction relation for processes.
\emphdef{Structural pre-congruence} $\sle$ is meant to rearrange process terms
which should not be distinguished. The relation is defined in
Table~\ref{tab:congruence}, where we write $\ProcessP \seq \ProcessQ$ in place
of the two inequalities $\ProcessP \sle \ProcessQ$ and $\ProcessQ \sle
\ProcessP$. Overall $\seq$ coincides with the conventional structural
congruence of the $\pi$-calculus, except that, as in~\cite{Kobayashi02}, we
omit the relation $\bang\Process \parop \Process \sle \bang\Process$ (the
reason will be explained in Remark~\ref{rem:sle}).

\emphdef{Evaluation} $\Expression \eval \Value$ and \emphdef{reduction}
$\ProcessP \lred[\Label] \ProcessQ$ are defined in Table~\ref{tab:reduction}.
Both relation are fairly standard. As in~\cite{KobayashiPierceTurner99},
reduction is decorated with a \emphdef{label} $\Label$ that is either a
channel or the special symbol $\tau$: in \refrule{r-comm} the label is the
channel $\Channel$ on which a message is exchanged; in \refrule{r-case} it is
$\tau$ since pattern matching is an internal computation not involving
communications.
Note that, as we allow expressions in input and output processes for both the
subject and the object of a communication, rule \refrule{r-comm} provides
suitable premises to evaluate them.
Rules~\refrule{r-par}, \refrule{r-new 1}, and \refrule{r-new 2}
propagate labels through parallel compositions and restrictions. In
\refrule{r-new 1}, the label $\Channel$ becomes $\tau$ when it escapes
the scope of $\Channel$. Rule \refrule{r-struct} closes reduction
under structural congruence.







\begin{example}[list sharing]
\label{ex:evenodd}
Below are the desugared representations of $\ProcessP$ and $\ProcessQ$
discussed in Section~\ref{sec:introduction}:
\[
\begin{array}{@{}rcl@{}}
  \ProcessP' & \eqdef &
  \bang\begin{lines}
  \receive\oddc\varZ\\
  \CaseV{~\First[\varZ]~}{
    \quad\Left[\ttunderscore]
    &
    \send{\SecondX{\SecondX\varZ}}{\FirstX{\SecondX\varZ}}
    \\
    \quad\Right[\varX]
    &
    \receive{\First[\varX]}\varY
    \send\evenc{
      \Pair{
        \Second[\varX]
      }{
        \Pair{
          \FirstX{\SecondX\varZ} + \varY
        }{
          \SecondX{\SecondX\varZ}
        }
      }
    }
  }
\end{lines}
\\
  \ProcessQ' & \eqdef &
  \bang\begin{lines}
  \receive\evenc\varZ\\
  \CaseV{~\First[\varZ]~}{
    \quad\Left[\ttunderscore]
    &
    \send{\SecondX{\SecondX\varZ}}{\FirstX{\SecondX\varZ}}
    \\
    \quad\Right[\varX]
    &
    \send\oddc{
      \Pair{
        \Second[\varX]
      }{
        \Pair{
          \FirstX{\SecondX\varZ}
        }{
          \SecondX{\SecondX\varZ}
        }
      }
    }
  }
\end{lines}
\end{array}
\]
where the constructors $\Left$ and $\Right$ respectively replace
$\Nil$ and $\Cons$, $\Left$ has an (unused) argument denoted by the
anonymous variable $\ttunderscore$, \revised{and tuple components are
  accessed using (possibly repeated) applications of $\First$ and
  $\Second$.}
\eoe
\end{example}

\section{Type system}
\label{sec:types}

In this section we define a type system for the language presented in
Section~\ref{sec:language}. The type system extends the one for the
linear $\pi$-calculus~\cite{KobayashiPierceTurner99} with composite
and \revised{possibly infinite, regular} types.
The key feature of the linear $\pi$-calculus is that channel types are
enriched with information about the number of times the channels they
denote are used for input/output operations. Such number is abstracted
into a \emphdef{use} $\Use$, $\dots$, which is an element of the set
$\set{0,1,\omega}$ where 0 and 1 obviously stand for no use and one
use only, while $\omega$ stands for any number of uses.


\revised{
\begin{definition}[types]
\label{def:types}
\emphdef{Types}, ranged over by $\TypeT$, $\TypeS$, $\dots$, are the
\emph{possibly infinite regular trees} built using the nullary
constructor $\tint$, the unary constructors
$\tchan\cdot{\Use_1}{\Use_2}$ for every combination of $\Use_1$ and
$\Use_2$, the binary constructors $\cdot \times \cdot$ (product) and
$\cdot \tsum \cdot$ (disjoint sum).
\end{definition}
}

The type $\tchan\Type{\Use_1}{\Use_2}$ denotes channels for exchanging
messages of type $\Type$. The uses $\Use_1$ and $\Use_2$ respectively denote
how many input and output operations are allowed on the channel. For example:
a channel with type $\tchan\Type{0}{1}$ cannot be used for input and must be
used once for sending a message of type $\Type$; a channel with type
$\tchan\Type{0}{0}$ cannot be used at all; a channel with type
$\tchan\Type\omega\omega$ can be used any number of times for sending and/or
receiving messages of type $\Type$.
A product $\Type_1 \times \Type_2$ describes pairs
$\Pair{\ValueV_1}{\ValueV_2}$ where $\ValueV_i$ has type $\TypeT_i$
for $i=1,2$.
A disjoint sum $\Type_1 \tsum \Type_2$ describes values of the form
$\Left[\Value]$ where $\Value$ has type $\Type_1$ or of the form
$\Right[\Value]$ where $\Value$ has type $\Type_2$.
Throughout the paper we let $\toper$ stand for either $\times$ or
$\tsum$.

\revised{ We do not provide a concrete, finite syntax for denoting
  infinite types and work directly with regular trees instead.
  Recall that a regular tree \finalrevision{is a partial function from
    paths to type constructors (see \eg~\cite[Chapter 21]{Pierce02}),
    it} consists of finitely many distinct subtrees, and admits finite
  representations using either the well-known $\mu$ notation or finite
  systems of equations~\cite{Courcelle83} \finalrevision{(our
    implementation internally uses both)}.
  Working directly with regular trees gives us the coarsest possible
  notion of type equality ($\TypeT = \TypeS$ means that
  \finalrevision{$\TypeT$ and $\TypeS$ are the same partial function})
  and it allows us to reuse some key results on regular trees that
  will be essential in the following. In particular, throughout the
  paper we will implicitly use the next result to define types as
  solutions of particular systems of equations:

\begin{theorem}
\label{thm:courcelle}
Let $\set{\tvar_i = \Term_i \mid 1\leq i\leq n}$ be a finite system of
equations where each $\Term_i$ is a finite term built using the
constructors in Definition~\ref{def:types} and the pairwise distinct
unknowns $\set{\tvar_1,\dots,\tvar_n}$.
If none of the $\Term_i$ is an unknown, then there exists a unique
substitution $\Solution = \set{ \tvar_i \mapsto \Type_i \mid 1 \leq
  i\leq n }$ such that $\Type_i = \Solution\Term_i$ and $\Type_i$ is a
regular tree for each $1\leq i\leq n$.
\end{theorem}
\begin{proof}
  All the right hand sides of the equations are finite -- hence
  regular -- and different from an unknown, therefore this result is
  just a particular case of~\cite[Theorem~4.3.1]{Courcelle83}.
\end{proof}

\begin{example}[integer stream]
  The type of \emph{integer streams} $\tint \times \ttparens{\tint
    \times \ttparens{\tint \times \cdots}}$ is the unique regular tree
  $\Type$ such that $\Type = \tint \times \Type$.
  To make sense out of this statement we have to be sure that such
  $\Type$ does exist and is indeed unique. Consider the equation
  $\tvar = \tint \times \tvar$ obtained from the above equality by
  turning each occurrence of the metavariable $\Type$ into the unknown
  $\tvar$ and observe that the right hand side of such equation is not
  an unknown. By Theorem~\ref{thm:courcelle}, there exists a unique
  regular tree $\Type$ such that $\Type = \tint \times \Type$.
  \finalrevision{Note that $\Type$ consists of two distinct subtrees,
    $\tint$ and $\Type$ itself.}
  \eoe
\end{example}

\begin{example}[lists]
  To verify the existence of the types $\TypeOdd$ and $\TypeEven$
  informally introduced in Section~\ref{sec:introduction}, consider
  the system of equations
\[
  \set{
  \tvar_1 = \tint \oplus \ttparens{\tchan\tint{1}{0} \times \tvar_2},
  \tvar_2 = \tint \oplus \ttparens{\tchan\tint{0}{0} \times \tvar_1}
  }
\]
obtained by turning the metavariables~$\TypeOdd$ and $\TypeEven$
in~\eqref{eq:oddeven} respectively into the unknowns $\tvar_1$ and
$\tvar_2$ and by using basic types and disjoint sums in place of the
list constructors $\Nil$ and $\Cons$. Theorem~\ref{thm:courcelle} says
that there exist two unique regular trees $\TypeOdd$ and $\TypeEven$
such that $\TypeOdd = \tint \oplus \ttparens{\tchan\tint10 \times
  \TypeEven}$ and $\TypeEven = \tint \oplus \ttparens{\tchan\tint00
  \times \TypeOdd}$.
Similarly, $\TypeList$ is the unique type such that $\TypeList = \tint
\oplus \ttparens{\tchan\tint10 \times \TypeList}$.
\eoe
\end{example}
}


We now define some key notions on uses and types. To begin with, we
define a binary operation $+$ on uses that allows us to express the
\emph{combined} use $\Use_1 + \Use_2$ of a channel that is used both
as denoted by $\Use_1$ \emph{and} as denoted by $\Use_2$. Formally:
\begin{equation}
\label{eq:comb}
\Use_1 + \Use_2
\eqdef
\begin{cases}
  \Use_1 & \text{if $\Use_2 = 0$} \\
  \Use_2 & \text{if $\Use_1 = 0$} \\
  \omega & \text{otherwise}
\end{cases}
\end{equation}

Note that $0$ is neutral and $\omega$ is absorbing for $+$ and that $1 + 1 =
\omega$, since $\omega$ is the only use allowing us to express the fact that a
channel is used twice.
In a few places we will write $2\Use$ as an abbreviation for $\Use + \Use$.






We now lift the notion of combination from uses to types.  Since types may be
infinite, we resort to a coinductive definition.

\begin{definition}[type combination]
\label{def:tand}
Let $\crel$ be the largest relation between pairs of types and types
such that $((\TypeT_1, \TypeT_2), \TypeS) \in \crel$ implies either:
\begin{itemize}
\item $\TypeT_1 = \TypeT_2 = \TypeS = \tint$, or

\item $\TypeT_1 = \tchan\Type{\Use_1}{\Use_2}$ and $\TypeT_2 =
  \tchan\Type{\Use_3}{\Use_4}$ and $\TypeS = \tchan\Type{\Use_1 +
    \Use_3}{\Use_2 + \Use_4}$, or

\item $\TypeT_1 = \TypeT_{11} \odot \TypeT_{12}$ and $\TypeT_2 =
  \TypeT_{21} \odot \TypeT_{22}$ and $\TypeS = \TypeS_1 \odot
  \TypeS_2$ and $((\TypeT_{1i}, \TypeT_{2i}), \TypeS_i) \in \crel$ for
  $i=1,2$.
\end{itemize}

\revised{Observe that $\crel$ is a partial binary function on types,
  that is $((\Type_1,\Type_2),\TypeS_1) \in \crel$ and
  $((\Type_1,\Type_2),\TypeS_2) \in \crel$ implies $\TypeS_1 =
  \TypeS_2$.
  When $(\TypeT, \TypeS) \in \dom(\crel)$, we write $\TypeT + \TypeS$
  for $\crel(\TypeT, \TypeS)$, that is the \emphdef{combination} of
  $\TypeT$ and $\TypeS$.
  Occasionally we also write $2\Type$ in place of $\Type + \Type$.}
\end{definition}

\revised{Intuitively, basic types combine with themselves and the
  combination of channel types with equal message types is obtained by
  combining corresponding uses.} For example, we have $\tchan\tint01 +
\tchan\tint10 = \tchan\tint11$ and $\tchan{\tchan\tint10}01 +
\tchan{\tchan\tint10}11 = \tchan{\tchan\tint10}1\omega$. In the latter
example, note that the uses of channel types within the top-most ones
are \emph{not} combined together.
Type combination propagates \revised{component-wise} on composite
types. \revised{For instance, we have $(\tchan\tint01 \times
  \tchan\tint00) + (\tchan\tint00 \times \tchan\tint10) =
  (\tchan\tint01 + \tchan\tint00) \times (\tchan\tint00 +
  \tchan\tint10) = \tchan\tint01 \times \tchan\tint10$.}
Unlike use combination, type combination is a partial operation: it is
undefined to combine two types having different structures, or to
combine two channel types carrying messages of different types. For
example, $\tint + \tchan\tint00$ is undefined and so is
$\tchan{\tchan\tint00}{0}{1} + \tchan{\tchan\tint01}{0}{1}$, because
$\tchan\tint00$ and $\tchan\tint01$ differ.

Types that can be combined together play a central role, so we name a relation
that characterizes them:

\begin{definition}[coherent types]
\label{def:mall}
We say that $\TypeT$ and $\TypeS$ are \emphdef{structurally coherent}
or simply \emphdef{coherent}, notation $\TypeT \mall \TypeS$, if
$\TypeT + \TypeS$ is defined, \finalrevision{namely there exists
  $\Type'$ such that $((\TypeT, \TypeS), \Type') \in \crel$.}
\end{definition}

Observe that $\mall$ is an equivalence relation, implying that a type
can always be combined with itself (\ie, $2\Type$ is always defined).
Type combination is also handy for characterizing a fundamental
partitioning of types:

\begin{definition}[unlimited and linear types]
\label{def:un}
We say that $\Type$ is \emph{unlimited}, notation $\unlimited\Type$,
if $2\Type = \Type$. We say that it is \emph{linear} otherwise.
\end{definition}

Channel types are either linear or unlimited depending on their
uses.
\revised{For example, $\tchan\Type00$ is unlimited because
  $\tchan\Type00 + \tchan\Type00 = \tchan\Type00$, whereas
  $\tchan\Type10$ is linear because $\tchan\Type10 + \tchan\Type10 =
  \tchan\Type\omega0 \ne \tchan\Type10$.
  Similarly, $\tchan\Type\omega\omega$ is unlimited while
  $\tchan\Type{0}{1}$ and $\tchan\Type11$ are linear.}
Other types are linear or unlimited depending on the channel types
occurring in them. For instance, $\tchan\Type{0}{0} \times
\tchan\Type{1}{0}$ is linear while $\tchan\Type{0}{0} \times
\tchan\Type{\omega}{0}$ is unlimited. Note that only the topmost
channel types of a type matter. For example, $\tchan{\tchan\Type11}00$
is unlimited despite of the fact that it contains the subterm
$\tchan\Type11$ which is itself linear, because such subterm is found
within an unlimited channel type.

We use \emphdef{type environments} to track the type of free names occurring
in expressions and processes.  Type environments $\Env$, $\dots$ are finite
maps from names to types that we write as $\Name_1 : \Type_1, \dots, \Name_n :
\Type_n$.
We identify type environments modulo the order of their associations,
write $\EmptyEnv$ for the \emphdef{empty environment}, $\dom(\Env)$
for the \emphdef{domain} of $\Env$, namely the set of names for which
there is an association in $\Env$, and $\Env_1, \Env_2$ for the
\emphdef{union} of $\Env_1$ and $\Env_2$ when $\dom(\Env_1) \cap
\dom(\Env_2) = \emptyset$.
We also extend the partial \revised{combination} operation $\tand$ on
types to a partial \revised{combination} operation on type
environments, thus:
\begin{equation}
\label{eq:combine}
\Env_1 \combine \Env_2
\eqdef
\begin{cases}
  \Env_1, \Env_2
  & \text{if $\dom(\Env_1) \cap \dom(\Env_2) = \emptyset$}
  \\
  (\Env_1' + \Env_2'), \Name : \Type_1 + \Type_2
  & \text{if $\Env_i = \Env_i', \Name : \Type_i$ for $i=1,2$}
\end{cases}
\end{equation}
The operation $\combine$ extends type combination
in~\cite{KobayashiPierceTurner99} and the $\uplus$ operator
in~\cite{SangiorgiWalker01}.
Note that $\Env_1 \combine \Env_2$ is undefined if there is
$\Name\in\dom(\Env_1) \cap \dom(\Env_2)$ such that $\Env_1(\Name)
\tand \Env_2(\Name)$ is undefined. Note also that $\dom(\Env_1
\combine \Env_2) = \dom(\Env_1) \cup \dom(\Env_2)$.
Thinking of type environments as of specifications of the resources
used by expressions/processes, $\Env_1 \combine \Env_2$ expresses the
combined use of the resources specified in $\Env_1$ and $\Env_2$. Any
resource occurring in only one of these environments occurs in $\Env_1
\combine \Env_2$; any resource occurring in both $\Env_1$ and $\Env_2$
is used according to the combination of its types in $\Env_1 \combine
\Env_2$.
For example, if a process sends an integer over a channel $\Channel$,
it will be typed in an environment that contains the association
$\Channel : \tchan\tint{0}{1}$; if another process uses the same
channel $\Channel$ for receiving an integer, it will be typed in an
environment that contains the association $\Channel :
\tchan\tint{1}{0}$. Overall, the parallel composition of the two
processes uses channel $\Channel$ according to the type
$\tchan\tint{0}{1} \tand \tchan\tint{1}{0} = \tchan\tint{1}{1}$ and
therefore it will be typed in an environment that contains the
association $\Channel : \tchan\tint11$.

The last notion we need before presenting the type rules is that of an unlimited
type environment. This is a plain generalization of the notion of unlimited
type, extended to the range of a type environment. We say that $\Env$ is
unlimited, notation $\unlimited\Env$, if $\unlimited{\Env(\Name)}$ for every
$\Name\in\dom(\Env)$. 
A process typed in an unlimited type environment \revised{need not}
use any of the resources described therein.

\begin{table}[t]
\framebox[\columnwidth]{
\begin{math}
\displaystyle
\begin{array}[t]{@{}c@{}}
  \multicolumn{1}{@{}l@{}}{\textbf{Expressions}} \\\\
  \inferrule[\defrule{t-int}]{
    \unlimited\Env
  }{
    \wte{\Env}{n}{\tint}
  }
  \qquad
  \inferrule[\defrule{t-name}]{
    \unlimited\Env
  }{
    \wte{\Env,\Name:\Type}{\Name}{\Type}
  }
  \qquad
  \inferrule[\defrule{t-inl}]{
    \wte{\Env}{\Expression}{\TypeT}
  }{
    \wte{\Env}{\Left[\Expression]}{\TypeT \tsum \TypeS}
  }
  \qquad
  \inferrule[\defrule{t-inr}]{
    \wte{\Env}{\Expression}{\TypeS}
  }{
    \wte{\Env}{\Right[\Expression]}{\TypeT \tsum \TypeS}
  }
  \\\\
  \inferrule[\defrule{t-pair}]{
    \wte{\Env_i}{\Expression_i}{\TypeT_i}~{}^{(i=1,2)}
  }{
    \wte{\Env_1 \tand \Env_2}{\Pair{\Expression_1}{\Expression_2}}{\TypeT_1 \times \TypeT_2}
  }
  \qquad
  \inferrule[\defrule{t-fst}]{
    \wte\Env\Expression{\TypeT \times \TypeS}
    \\
    \unlimited\TypeS
  }{
    \wte\Env{\First[\Expression]}\TypeT
  }
  \qquad
  \inferrule[\defrule{t-snd}]{
    \wte\Env\Expression{\TypeT \times \TypeS}
    \\
    \unlimited\TypeT
  }{
    \wte\Env{\Second[\Expression]}\TypeS
  }
  \\\\
  \multicolumn{1}{@{}l@{}}{\textbf{Processes}} \\\\
  \inferrule[\defrule{t-idle}]{
      \unlimited\Env
  }{
    \wtp{\Env}{\idle}
  }
  \qquad
  \inferrule[\defrule{t-in}]{
    \wte{\Env_1}\Expression{\tchan\Type{1+\Use_1}{2\Use_2}}
    \\
    \wtp{\Env_2, \var : \Type}\Process
  }{
    \wtp{
      \Env_1 \tand \Env_2
    }{
      \receive\Expression\var\Process
    }
  }
  \qquad
  \inferrule[\defrule{t-out}]{
    \wte{\Env_1}{\ExpressionE}{\tchan\Type{2\Use_1}{1+\Use_2}}
    \\
    \wte{\Env_2}{\ExpressionF}{\Type}
  }{
    \wtp{
      \Env_1 \tand \Env_2
    }{
      \send\ExpressionE\ExpressionF
    }
  }
  \\\\
  \inferrule[\defrule{t-rep}]{
    \wtp{\Env}{\Process}
    \\
    \unlimited\Env
  }{
    \wtp{
      \Env
    }{
      \bang\Process
    }
  }
  \qquad
  \inferrule[\defrule{t-par}]{
    \wtp{\Env_i}{\ProcessP_i}~{}^{(i=1,2)}
  }{
    \wtp{\Env_1 \tand \Env_2}{\ProcessP_1 \parop \ProcessP_2}
  }
  \qquad
  \inferrule[\defrule{t-new}]{
    \wtp{
      \Env, \Channel : \tchan\Type{\Use}{\Use}
    }{
      \Process
    }
  }{
    \wtp{\Env}{
      \new{\Channel}\Process
    }
  }
  \\\\
  \inferrule[\defrule{t-case}]{
    \wte{\Env_1}\Expression{\TypeT\tsum\TypeS}
    \\
    \wtp{\Env_2, \varX_i : \TypeT}{\ProcessP_i}
    ~{}^{(i=\Left,\Right)}
  }{
    \wtp{\Env_1 \tand \Env_2}{
      \CaseShort\Expression{i}{\var}\Process
    }
  }
\end{array}
\end{math}
}
\caption{\label{tab:type_system} Type rules for expressions and processes.}
\end{table}

Type rules for expressions and processes are presented in
Table~\ref{tab:type_system}.  These rules are basically the same as
those found in the
literature~\cite{KobayashiPierceTurner99,IgarashiKobayashi00}. \revised{The
  possibility of sharing data structures among several processes,
  which we have exemplified in Section~\ref{sec:introduction},} is a
consequence of our notion of type combination \revised{extended to
  composite regular types.}

Type rules for expressions are unremarkable. Just observe that unused
type environments must be unlimited. Also, the projections $\First$
and $\Second$ discard one component of a pair, so the discarded
component must have an unlimited type.

Let us move on to the type rules for processes.
The idle process does nothing, so it is well typed only in an
unlimited environment.
Rule~\refrule{t-in} types an input process
$\receive\Expression\var\Process$. The subject $\Expression$ must
evaluate to a channel whose input use is either $1$ or $\omega$ and
whose output use is either $0$ or $\omega$. We capture the first
condition saying that the input use of the channel has the form $1 +
\Use_1$ for some $\Use_1$, and the second condition saying that the
output use of the channel has the form $2\Use_2$ for some $\Use_2$.
The continuation $\Process$ is typed in an environment enriched with
the association for the received message $\var$. \revised{Note the
  combination $\Env_1 + \Env_2$ in the conclusion of
  rule~\refrule{t-in}.} In particular, if $\Expression$ evaluates to a
linear channel, its input capability is consumed by the operation and
such channel can no longer be used for inputs in the continuation.
Rule~\refrule{t-out} types an output process
$\send\ExpressionE\ExpressionF$. The rule is dual to \refrule{t-in} in
that it requires the channel to which $\ExpressionE$ evaluates to have
a positive output use.
Rule~\refrule{t-rep} states that a replicated process $\bang\Process$ is well
typed in the environment $\Env$ provided that $\Process$ is well typed in an
unlimited $\Env$. The rationale is that $\bang\Process$ stands for an
unbounded number of copies of $\Process$ composed in parallel, hence
$\Process$ cannot contain (free) linear channels.
The rules \refrule{t-par} and \refrule{t-case} are conventional, with
the by now familiar use of environment combination for properly
distributing linear resources to the various subterms of a process.
The rule \refrule{t-new} is also conventional. We require the restricted
channel to have the same input and output uses. While this is not necessary
for the soundness of the type system, in practice it is a reasonable
requirement. We also argue that this condition is important for the modular
application of the type reconstruction algorithm; we will discuss this aspect
more in detail in Section~\ref{sec:implementation}.

\revised{As in many behavioral type systems,} the type environment in
which the reducing process is typed may change as a consequence of the
reduction. More specifically, reductions involving a communication on
channels \emph{consume} 1 unit from both the input and output uses of
the channel's type.
In order to properly state subject reduction, we define a reduction relation
over type environments. In particular, we write $\lred[\Label]$ for the least
relation between type environments such that
\[
\Env \lred[\tau] \Env
\text{\qquad\qquad}
\Env + \Channel : \tchan\Type{1}{1}
\lred[\Channel]
\Env
\]
In words, $\lred[\tau]$ denotes an internal computation (pattern matching) or
a communication on some restricted channel which does not consume any resource
from the type environment, while $\lred[\Channel]$ denotes a communication on
channel $\Channel$ which consumes 1 use from both the input and output slots
in $\Channel$'s type. For example, we have
\[
  \Channel : \tchan\tint11 \lred[\Channel] \Channel : \tchan\tint00
\]
\revised{by taking $\Env \eqdef \Channel : \tchan\tint00$ in the
  definition of $\lred[\Channel]$ above, since $\Env + \Channel :
  \tchan\tint11 = \Channel : \tchan\tint11$.} The residual environment
denotes the fact that the (linear) channel $\Channel$ can no longer be
used for communication.

Now we have:

\begin{theorem}
\label{thm:sr}
Let $\wtp\Env\ProcessP$ and $\ProcessP \lred[\Label] \ProcessQ$. Then
$\wtp{\Env'}\ProcessQ$ for some $\Env'$ such that $\Env \lred[\Label]
\Env'$.
\end{theorem}

Theorem~\ref{thm:sr} establishes not only a subject reduction result,
but also a soundness result because it implies that a channel is used
no more than its type allows.
It is possible to establish more properties of the linear $\pi$-calculus, such
as the fact that communications involving linear channels enjoy partial
confluence. In this work we focus on the issue of type reconstruction. The
interested reader may refer to~\cite{KobayashiPierceTurner99} for further
results.

\newcommand{\OddC}{\ChannelA}
\newcommand{\EvenC}{\ChannelB}

\newcommand{\TypeZero}{\Type_{\mathit{zero}}}

\newcommand{\SZero}{\TypeS_{\mathit{zero}}}
\newcommand{\SEven}{\TypeS_{\mathit{even}}}
\newcommand{\SOdd}{\TypeS_{\mathit{odd}}}

\begin{example}
\label{ex:do_typing}
We consider again the processes \finalrevision{$\ProcessP'$ and
  $\ProcessQ'$} in Example~\ref{ex:evenodd} and sketch a few key
derivation steps to argue that they are well typed. To this aim,
consider the types $\TypeOdd$, $\TypeEven$, and $\TypeZero$ that
satisfy the equalities below
\[
\begin{array}{rcl}
  \TypeOdd  & = & \tint \oplus \ttparens{\tchan\tint01 \times \TypeEven} \\
  \TypeEven & = & \tint \oplus \ttparens{\tchan\tint00 \times \TypeOdd} \\
  \TypeZero & = & \tint \oplus \ttparens{\tchan\tint00 \times \TypeZero} \\
\end{array}
\]
and also consider the types of the messages respectively carried by
$\oddc$ and $\evenc$:
\[
\begin{array}{rcl}
  \SOdd & \eqdef & \TypeOdd \times \ttparens{\tint \times \tchan\tint01} \\
  \SEven & \eqdef & \TypeEven \times \ttparens{\tint \times \tchan\tint01} \\
\end{array}
\]
Now, in the $\Left$ branch of \finalrevision{$\ProcessP'$} we derive
\treeref1
\[
\begin{prooftree}
  \pt{
    \pt{
      \justifies
      \wte{
        \varZ : \TypeZero \times \ttparens{\tint \times \tchan\tint01}
      }{
        \varZ
      }{
        \TypeZero \times \ttparens{\tint \times \tchan\tint01}
      }
      \using\refrule{t-name}
    }
    \justifies
    \wte{
      \varZ : \TypeZero \times \ttparens{\tint \times \tchan\tint01}
    }{
      \SecondX\varZ
    }{
      \tint \times \tchan\tint01
    }
    \using\refrule{t-snd}
  }
  \justifies
  \wte{
    \varZ : \TypeZero \times \ttparens{\tint \times \tchan\tint01}
  }{
    \SecondX{\SecondX\varZ}
  }{
    \tchan\tint01
  }
  \using\refrule{t-snd}
\end{prooftree}
\]
using the fact that $\unlimited\TypeZero$ and $\unlimited\tint$. We
also derive \treeref2
\[
\begin{prooftree}
  \pt{
    \pt{
      \justifies
      \wte{
        \varZ : \TypeZero \times \ttparens{\tint \times \tchan\tint00}
      }{
        \varZ
      }{
        \TypeZero \times \ttparens{\tint \times \tchan\tint00}
      }
      \using\refrule{t-name}
    }
    \justifies
    \wte{
      \varZ : \TypeZero \times \ttparens{\tint \times \tchan\tint00}
    }{
      \SecondX\varZ
    }{
      \tint \times \tchan\tint00
    }
    \using\refrule{t-snd}
  }
  \justifies
  \wte{
    \varZ : \TypeZero \times \ttparens{\tint \times \tchan\tint00}
  }{
    \FirstX{\SecondX\varZ}
  }{
    \tint
  }
  \using\refrule{t-fst}
\end{prooftree}
\]
using the fact that $\unlimited\TypeZero$ and
$\unlimited{\tchan\tint00}$, therefore we derive \treeref3
\[
\begin{prooftree}
  \treeref1
  \qquad
  \qquad
  \treeref2
  \justifies
  \wtp{
    \varZ : \TypeZero \times \ttparens{\tint \times \tchan\tint01},
    \ttunderscore : \tint
  }{
    \send{\SecondX{\SecondX\varZ}}{\FirstX{\SecondX\varZ}}
  }
  \using\refrule{t-out}
\end{prooftree}
\]
using the combination
\[
  (\TypeZero \times \ttparens{\tint \times \tchan\tint01})
  +
  (\TypeZero \times \ttparens{\tint \times \tchan\tint00})
  =
  \TypeZero \times \ttparens{\tint \times \tchan\tint01}
\]
Already in this sub-derivation we appreciate that although the pair
$\varZ$ is accessed twice, its type \finalrevision{in the conclusion
  of \treeref3} correctly tracks the fact that the channel contained
\finalrevision{in $\varZ$} is only used once, for an output.

For the $\Right$ branch \finalrevision{in $\ProcessP'$ there exists
  another derivation \treeref4 concluding}
\[
\begin{prooftree}
  \vdots
  \qquad\qquad
  \vdots
  \justifies
  \wtp{
    \evenc : \tchan\SEven0\omega,
    \varX : \tchan\tint10 \times \TypeEven,
    \varZ : \TypeZero \times \tint \times \tchan\tint10
  }{
    \receive{\First[\varX]}\varY
    \cdots
  }
  \using\refrule{t-in}
\end{prooftree}
\]
Now we conclude
\[
\begin{prooftree}
  \[
    \[
      \[
        \justifies
        \wte{
          \varZ : \SOdd
        }{
          \varZ
        }{
          \SOdd
        }
        \using\refrule{t-name}
      \]
      \qquad
      \treeref3
      \qquad
      \treeref4
      \qquad
      \justifies
      \wtp{
        \evenc : \tchan\SEven{0}{\omega},
        \varZ : \SOdd
      }{
        \mkkeyword{case}~\varZ~\mkkeyword{of} \cdots
      }
      \using\refrule{t-case}
    \]
    \justifies
    \wtp{
      \oddc : \tchan\SOdd{\omega}{0},
      \evenc : \tchan\SEven{0}{\omega}
    }{
      \receive\oddc\varZ
      \mkkeyword{case}~\varZ~\mkkeyword{of} \cdots
    }
    \using\refrule{t-in}
  \]
  \justifies
  \wtp{
    \oddc : \tchan\SOdd{\omega}{0},
    \evenc : \tchan\SEven{0}{\omega}
  }{
    \ProcessP'
  }
  \using\refrule{t-rep}
\end{prooftree}
\]
Note that $\oddc$ and $\evenc$ must be unlimited channels because they
occur free in a replicated process, for which rule~\refrule{t-rep}
requires an unlimited environment. A similar derivation shows that
\finalrevision{$\ProcessQ'$} is well typed in an environment where the
types of $\oddc$ and $\evenc$ have swapped uses
\[
\begin{prooftree}
  \vdots
  \justifies
  \wtp{
    \oddc : \tchan\SOdd{0}{\omega},
    \evenc : \tchan\SEven{\omega}{0}
  }{
    \ProcessQ'
  }
  \using\refrule{t-rep}
\end{prooftree}
\]
so the combined types of $\oddc$ and $\evenc$ are
$\tchan\SOdd\omega\omega$ and $\tchan\SEven\omega\omega$,
respectively.  Using these, we find a typing derivation for the
process $\ProcessR$ \finalrevision{in
  Section~\ref{sec:introduction}}. Proceeding bottom-up we have
\[
\begin{prooftree}
  \vdots
  \justifies
  \wtp{
    \oddc : \tchan\SOdd{\omega}{\omega},
    l : \TypeOdd,
    \ChannelA : \tchan\tint01
  }{
    \send\oddc{\ttparens{l\ttcomma 0\ttcomma\ChannelA}}
  }
  \using\refrule{t-out}
\end{prooftree}
\]
and
\[
\begin{prooftree}
  \vdots
  \justifies
  \wtp{
    \evenc : \tchan\SEven{\omega}{\omega},
    l : \TypeEven,
    \ChannelB : \tchan\tint01
  }{
    \send\evenc{\ttparens{l\ttcomma 0\ttcomma\ChannelB}}
  }
  \using\refrule{t-out}
\end{prooftree}
\]
as well as
\[
\begin{prooftree}
  \vdots
  \justifies
  \wtp{
    \ChannelA : \tchan\tint10,
    \ChannelB : \tchan\tint10,
    r : \tchan\tint01
  }{
    \receive\ChannelA\varX
    \receive\ChannelB\varY
    \send{r}{\ttparens{\varX + \varY}}
  }
  \using\refrule{t-in}
\end{prooftree}
\]
from which we conclude
\[
\begin{prooftree}
  \vdots
  \Justifies
  \wtp{
    \oddc : \tchan\SOdd{\omega}{\omega},
    \evenc : \tchan\SEven{\omega}{\omega},
    l : \TypeList,
    r : \tchan\tint01
  }{
    \new{\ChannelA\ttcomma\ChannelB}
    \cdots
  }
  \using\text{\refrule{t-new} (twice)}
\end{prooftree}
\]
using the property $\TypeOdd \tand \TypeEven = \TypeList$.
\eoe
\end{example}

\finalrevision{We conclude this section with a technical remark to
  justify the use of a structural precongruence relation in place of a
  more familiar symmetric one.}

\begin{remark}
\label{rem:sle}
\newcommand{\EnvP}[1]{\Env_{#1}}
Let us show why the relation $\bang\Process \parop \Process \sle
\bang\Process$ would \revised{invalidate Theorem~\ref{thm:sr} (more
  specifically, Lemma~\ref{lem:struct})} in our setting \revised{(a
  similar phenomenon is described in~\cite{Kobayashi02})}.  To this
aim, consider the process
\[
  \ProcessP \eqdef
  \receive\ChannelA\varX
  \new\ChannelC
  \ttparens{
    \bang\receive\ChannelC\varY
    \send\ChannelC\varY
    \parop
    \send\ChannelC\ChannelB
  }
\]
and \revised{the type environment $\EnvP\Use \eqdef \ChannelA :
  \tchan\tint\omega0, \ChannelB : \tchan\tint0\Use$ for an arbitrary
  $\Use$.} We can derive \revised{
\[
\begin{prooftree}
  \vdots
  \pt{
    \pt{
      \pt{
        \pt{
          \vdots
          \qquad
          \pt{
            \vdots
            \qquad
            \vdots
            \justifies
            \wtp{
              \ChannelC : \tchan{\tchan\tint0\Use}\omega\omega,
              \varY : \tchan\tint0\Use
            }{
              \send\ChannelC\varY
            }
            \using\refrule{t-out}
          }
          \justifies
          \wtp{
            \ChannelC : \tchan{\tchan\tint0\Use}\omega\omega
          }{
            \receive\ChannelC\varY
            \send\ChannelC\varY
          }
          \using\refrule{t-in}
        }
        \justifies
        \wtp{
          \ChannelC : \tchan{\tchan\tint0\Use}\omega\omega
        }{
          \bang\receive\ChannelC\varY
          \send\ChannelC\varY
        }
        \using\refrule{t-rep}
      }
      \quad
      \pt{
        \vdots
        \qquad
        \vdots
        \justifies
        \wtp{
          \EnvP\Use,
          \ChannelC : \tchan{\tchan\tint0\Use}01
        }{
          \send\ChannelC\ChannelB
        }
        \using\refrule{t-out}
      }
      \justifies
      \wtp{
        \EnvP\Use,
        \ChannelC : \tchan{\tchan\tint0\Use}\omega\omega
      }{
        \bang\receive\ChannelC\varY
        \send\ChannelC\varY
        \parop
        \send\ChannelC\ChannelB
      }
      \using\refrule{t-par}
    }
    \justifies
    \wtp{\EnvP\Use}{
      \new\ChannelC\cdots
    }
    \using\refrule{t-new}
  }
  \justifies
  \wtp{\EnvP\Use}\ProcessP
  \using\refrule{t-in}
\end{prooftree}
\]
where we have elided a few obvious typing derivations for expressions.
In particular, we can find a derivation where $\ChannelB$ has an
unlimited type ($\Use = 0$) and another one where $\ChannelB$ has a
linear type ($\Use=1$).}
This is possible because channel $\ChannelC$, which is restricted
within $\Process$, can be given different types -- respectively,
$\tchan{\tchan\tint00}\omega\omega$ and
$\tchan{\tchan\tint01}\omega\omega$ -- in the two derivations.
We can now obtain
\revised{
\[
\begin{prooftree}
  \pt{
    \pt{
      \vdots
      \justifies
      \wtp{
        \EnvP0
      }{
        \ProcessP
      }
    }
    \justifies
    \wtp{
      \EnvP0
    }{
      \bang\ProcessP
    }
    \using\refrule{t-rep}
  }
  \qquad
  \pt{
    \vdots
    \justifies
    \wtp{
      \EnvP1
    }{
      \ProcessP
    }
  }
  \justifies
  \wtp{
    \EnvP1
  }{
    \bang\ProcessP \parop \ProcessP
  }  
  \using\refrule{r-par}
\end{prooftree}
\]
because $\unlimited{\EnvP0}$ and $\EnvP0 + \EnvP1 = \EnvP1$.
}
If we allowed the structural congruence rule $\bang\ProcessP \parop
\ProcessP \sle \bang\ProcessP$, \revised{then
  $\wtp{\EnvP1}{\bang\ProcessP}$ would \emph{not} be derivable because
  $\EnvP1$ is linear, hence typing would not be preserved by
  structural pre-congruence.}
This problem is avoided in \cite{KobayashiPierceTurner99,
  IgarashiKobayashi00} by limiting replication to input prefixes,
omitting any structural congruence rule for replications, \revised{and
  adding a dedicated synchronization rule for them}.
In \cite{KobayashiPierceTurner99} it is stated that ``the full
pi-calculus replication operator poses no problems for the linear type
system'', but this holds because there the calculus is typed, so
multiple typing derivations for the same process $\ProcessP$ above
would assign the same type to $\ChannelC$ and, in turn, the same type
to~$\ChannelB$.
\eoe
\end{remark}

\section{Constraint Generation}
\label{sec:generator}

\begin{table}
\framebox[\columnwidth]{
\begin{math}
\displaystyle
\begin{array}{@{}c@{\qquad}c@{}}
\begin{array}[t]{@{}rcl@{\quad}l@{}}
  \TypeExprT, \TypeExprS & ::= & & \textbf{Type expression} \\
  &   & \tvar & \text{(type variable)} \\
  & | & \tint & \text{(integer)} \\
  & | & \tchan\TypeExpr\UseExprU\UseExprV & \text{(channel)} \\
  & | & \TypeExprT \times \TypeExprS & \text{(product)} \\
  & | & \TypeExprT \tsum \TypeExprS & \text{(disjoint sum)} \\
\end{array}
&
\begin{array}[t]{@{}rcl@{\quad}l@{}}
  \UseExprU, \UseExprV & ::= & & \textbf{Use expression} \\
  &   & \uvar & \text{(use variable)} \\
  & | & \Use & \text{(use constant)} \\
  & | & \UseExprU + \UseExprV & \text{(use combination)} \\
\end{array}
\end{array}
\end{math}
}
\caption{\label{tab:expressions} Syntax of use and type expressions.}
\end{table}

We formalize the problem of type reconstruction as follows: given a
process $\Process$, find a type environment $\Env$ such that
$\wtp\Env\Process$, provided there is one.  \revised{In general, in
  the derivation for $\wtp\Env\Process$ we also want to identify as
  many linear channels as possible. We will address this latter aspect
  in Section~\ref{sec:solver}.}

\subsection{Syntax-directed generation algorithm}
The type rules shown in Table~\ref{tab:type_system} rely on a fair
amount of guessing that concerns the structure of types in the type
environment, how they are split/combined using $+$, and the uses
occurring in them. So, these rules cannot be easily interpreted as a
type reconstruction algorithm. The way we follow to define one is
conventional: first, we give an alternative set of (almost)
syntax-directed rules that generate \emph{constraints} on types; then,
we search for a solution of such constraints. The main technical
challenge is that \revised{we cannot base our type reconstruction
  algorithm on conventional unification because we have to deal with
  constraints expressing not only the \emph{equality} between types
  and uses, but also the \emph{combination} of types and uses. In
  addition, we work with possibly infinite types.}

To get started, we introduce use and type \emph{expressions}, which
share the same structure as uses/types but they differ from them in
two fundamental ways:
\begin{enumerate}
\item We allow use/type variables to stand for unknown uses/types.

\item We can express symbolically the combination of use expressions.
\end{enumerate}

We therefore introduce a countable set of \emphdef{use variables} $\uvar$,
$\dots$ as well as a countable set of \emphdef{type variables} $\tvarA$,
$\tvarB$, $\dots$; the syntax of use expressions $\UseExprU$, $\UseExprV$,
$\dots$ and of type expressions $\TypeExprT$, $\TypeExprS$, $\dots$ is given
in Table~\ref{tab:expressions}.
Observe that every use is also a use expression and every finite type is also
a type expression.
\revised{We say that $\TypeExpr$ is \emphdef{proper} if it is
  different from a type variable.}

\emphdef{Constraints} $\Constraint$, $\dots$ are defined by the
grammar below:
\[
\begin{array}{rcl@{\qquad}l}
  \Constraint & ::= & & \textbf{Constraint} \\
  &   & \TypeExprT \ceq \TypeExprS & \text{(type equality)}
  \\
  & | & \TypeExprT \ceq \TypeExprS_1 + \TypeExprS_2 & \text{(type
    combination)}
  \\
  & | & \TypeExprT \con\mall \TypeExprS & \text{(type coherence)}
  \\
  & | & \UseExprU \ceq \UseExprV & \text{(use equality)}
\end{array}
\]
\revised{Constraints express relations between types/uses that must be
  satisfied in order for a given process to be well typed. In
  particular, we need to express equality constraints between types
  ($\TypeExprT \ceq \TypeExprS$) and uses ($\UseExprU \ceq
  \UseExprV$), coherence constraints ($\TypeExprT \con\mall
  \TypeExprS$), and combination constraints between types ($\TypeExprT
  \ceq \TypeExprS_1 + \TypeExprS_2$).}
%
We will write $\unlimited\TypeExpr$ as an abbreviation for the
constraint $\TypeExpr \ceq \TypeExpr + \TypeExpr$. \finalrevision{This
  notation is motivated by Definition~\ref{def:un}, according to
  which} a type is unlimited if and only if it is equal to its own
combination.
We let $\Constraints$, $\dots$ range over finite constraint sets. The
\revised{set of \emphdef{expressions}} of a constraint set
$\Constraints$, written \revised{$\expr(\Constraints)$}, is the
(finite) set of use and type expressions occurring in the constraints
in $\Constraints$.

\begin{table}
\framebox[\textwidth]{
\begin{math}
\displaystyle
\begin{array}{@{}c@{}}
  \inferrule[\defrule{c-env 1}]{
    \dom(\EnvX_1) \cap \dom(\EnvX_2) = \emptyset
  }{
    \combinenv{\EnvX_1}{\EnvX_2}{\EnvX_1,\EnvX_2}{\emptyset}
  }
  \qquad
  \inferrule[\defrule{c-env 2}]{
    \combinenv{\EnvX_1}{\EnvX_2}{\EnvX}{\Constraints}
    \\
    \text{$\tvar$ fresh}
  }{
    \combinenv{(\EnvX_1, \Name : \TypeExprT)}{(\EnvX_2, \Name : \TypeExprS)}{
      \EnvX, \Name : \tvar
    }{
      \Constraints \cup \set{ \tvar \con\meq \TypeExprT + \TypeExprS }
    }
  }
  \\\\
  \inferrule[\defrule{m-env 1}]{}{
    \mergenv{\EmptyEnv}{\EmptyEnv}{\EmptyEnv}{\emptyset}
  }
  \qquad
  \inferrule[\defrule{m-env 2}]{
    \mergenv{\EnvX_1}{\EnvX_2}{\EnvX}{\Constraints}
  }{
    \mergenv{
      (\EnvX_1, \Name : \TypeExprT)
    }{
      (\EnvX_2, \Name : \TypeExprS)
    }{
      \EnvX, \Name : \TypeExprT
    }{
      \Constraints \cup \set{ \TypeExprT \con= \TypeExprS }
    }
  }
\end{array}
\end{math}
}
\caption{\label{tab:env} Combining and merging operators for type environments.}
\end{table}

The type reconstruction algorithm \emph{generates} type environments
for the expressions and processes being analyzed. Unlike the
environments in Section~\ref{sec:types}, these environments associate
names with type expressions. For this reason we will let $\EnvX$,
$\dots$ range over the environments generated by the reconstruction
algorithm, although we will refer to them as type environments.

The algorithm also uses two auxiliary operators $\combineop$ and
$\mergeop$ defined in Table~\ref{tab:env}.
The relation $\combinenv{\EnvX_1}{\EnvX_2}{\EnvX}{\Constraints}$ combines two
type environments $\EnvX_1$ and $\EnvX_2$ into $\EnvX$ when the names in
$\dom(\EnvX_1) \cup \dom(\EnvX_2)$ are used \emph{both} as specified in
$\EnvX_1$ \emph{and also} as specified in $\EnvX_2$ and, in doing so,
generates a set of constraints $\Constraints$.
So $\combineop$ is analogous to $\combine$ in \eqref{eq:combine}.
When $\EnvX_1$ and $\EnvX_2$ have disjoint domains, $\EnvX$ is just
the union of $\EnvX_1$ and $\EnvX_2$ and no constraints are
generated.
Any name $\Name$ that occurs in $\dom(\EnvX_1) \cap \dom(\EnvX_2)$ is used
according to the combination of $\EnvX_1(\Name)$ and $\EnvX_2(\Name)$. In
general, $\EnvX_1(\Name)$ and $\EnvX_2(\Name)$ are type expressions with free
type variables, hence this combination cannot be ``computed'' or ``checked''
right away. Instead, it is recorded as the constraint $\tvar \con\meq
\EnvX_1(\Name) + \EnvX_2(\Name)$ where $\tvar$ is a fresh type variable.

The relation $\mergenv{\EnvX_1}{\EnvX_2}{\EnvX}{\Constraints}$ merges two type
environments $\EnvX_1$ and $\EnvX_2$ into $\EnvX$ when the names in
$\dom(\EnvX_1) \cup \dom(\EnvX_2)$ are used \emph{either} as specified in
$\EnvX_1$ \emph{or} as specified in $\EnvX_2$ and, in doing so, generates a
constraint set $\Constraints$.
This merging is necessary when typing the alternative branches of a
$\mkkeyword{case}$: recall that rule \refrule{t-case} in
Table~\ref{tab:type_system} requires the \emph{same} type environment $\Env$
for typing the two branches of a $\mkkeyword{case}$. Consequently, $\EnvX_1
\mergeop \EnvX_2$ is defined only when $\EnvX_1$ and $\EnvX_2$ have the same
domain, and produces a set of constraints $\Constraints$ saying that the
corresponding types of the names in $\EnvX_1$ and $\EnvX_2$ must be equal.

\begin{table}[t]
\framebox[\textwidth]{
\begin{math}
\displaystyle
\begin{array}{@{}c@{}}
  \multicolumn{1}{@{}l@{}}{\textbf{Expressions}} \\
\inferrule[\defrule{i-int}]{}{
  \rte{n}{\tint}\EmptyEnv\emptyset
}
\qquad
\inferrule[\defrule{i-name}]{%
}{
  \rte{\Name}{\tvar}{\Name:\tvar}\emptyset
}
\qquad
\inferrule[\defrule{i-inl}]{
  \rte{\Expression}{\TypeExpr}{\EnvX}{\Constraints}
}{
  \rte{\Left[\Expression]}{\TypeExprT \tsum \tvar}{\EnvX}{\Constraints}
}
\qquad
\inferrule[\defrule{i-inr}]{
  \rte{\Expression}{\TypeExpr}{\EnvX}{\Constraints}
}{
  \rte{\Right[\Expression]}{\tvar \tsum \TypeExpr}{\EnvX}{\Constraints}
}
\\\\
\inferrule[\defrule{i-pair}]{
  \rte{\Expression_i}{\TypeExpr_i}{\EnvX_i}{\Constraints_i}
  ~{}^{(i=1,2)}
  \\
  \combinenv{\EnvX_1}{\EnvX_2}{\EnvX}{\Constraints_3}
}{
  \rte{\Pair{\Expression_1}{\Expression_2}}{
    \TypeExpr_1 \times \TypeExpr_2
  }{
    \EnvX
  }{
    \Constraints_1 \cup \Constraints_2 \cup \Constraints_3
  }
}
\\\\
\inferrule[\defrule{i-fst}]{
  \rte\Expression\TypeExpr\EnvX\Constraints
}{
  \rte{
    \First[\Expression]
  }{
    \tvarA
  }{
    \EnvX
  }{
    \Constraints \cup \set{
      \TypeExpr \con\meq \tvarA \times \tvarB,
      \unlimited\tvarB
    }
  }
}
\qquad
\inferrule[\defrule{i-snd}]{
  \rte\Expression\TypeExpr\EnvX\Constraints
}{
  \rte{
    \Second[\Expression]
  }{
    \tvarB
  }{
    \EnvX
  }{
    \Constraints \cup \set{
      \TypeExpr \con\meq \tvarA \times \tvarB,
      \unlimited\tvarA
    }
  }
}
\\\\
  \multicolumn{1}{@{}l@{}}{\textbf{Processes}} \\\\
\inferrule[\defrule{i-idle}]{}{
  \rtp{\idle}\EmptyEnv\emptyset
}
\qquad
\inferrule[\defrule{i-in}]{
  \rte\Expression\TypeExprT{\EnvX_1}{\Constraints_1}
  \\
  \rtp\Process{\EnvX_2,\var:\TypeExprS}{\Constraints_2}
  \\
  \combinenv{\EnvX_1}{\EnvX_2}{\EnvX}{\Constraints_3}
}{
  \rtp{\receive\Expression\var\Process}{
    \EnvX 
  }{
    \Constraints_1 \cup \Constraints_2 \cup \Constraints_3 \cup \set{
      \TypeExprT \con= \tchan\TypeExprS{1+\uvar_1}{2\uvar_2}
    }
  }
}
\\\\
\inferrule[\defrule{i-out}]{
  \rte\ExpressionE\TypeExprT{\EnvX_1}{\Constraints_1}
  \\
  \rte\ExpressionF\TypeExprS{\EnvX_2}{\Constraints_2}
  \\
  \combinenv{\EnvX_1}{\EnvX_2}{\EnvX}{\Constraints_3}
}{
  \rtp{\send\ExpressionE\ExpressionF}{
    \EnvX
  }{
    \Constraints_1 \cup \Constraints_2 \cup \Constraints_3 \cup \set{
      \TypeExprT \con= \tchan\TypeExprS{2\uvar_1}{1+\uvar_2}
    }
  }
}
\qquad
\inferrule[\defrule{i-rep}]{
  \rtp{\Process}{\EnvX}{\Constraints}
  \\
  \combinenv{\EnvX}{\EnvX}{\EnvX'}{\Constraints'}
}{
  \rtp{\bang\Process}{\EnvX'}{\Constraints \cup \Constraints'}
}
\\\\
\inferrule[\defrule{i-par}]{
  \rtp{\Process_i}{\EnvX_i}{\Constraints_i}~{}^{(i=1,2)}
  \\
  \combinenv{\EnvX_1}{\EnvX_2}{\EnvX}{\Constraints_3}
}{
  \rtp{\Process_1 \parop \Process_2}{\EnvX}{
    \Constraints_1 \cup \Constraints_2 \cup \Constraints_3
  }
}
\qquad
\inferrule[\defrule{i-new}]{
  \rtp{\Process}{\EnvX, \Channel : \TypeExpr}{\Constraints}
}{
  \rtp{\new\Channel\Process}{\EnvX}{
    \Constraints
    \cup
    \set{ \TypeExpr \con= \tchan\tvar{\uvar}{\uvar} }
  }
}
\\\\
\inferrule[\defrule{i-case}]{
  \rte\Expression\TypeExpr{\EnvX_1}{\Constraints_1}
  \\
  \rtp{\Process_i}{\EnvX_i, \var_i : \TypeExpr_i}{\Constraints_i}{}^{(i=\Left,\Right)}
  \\
  \mergenv{\EnvX_{\Left}}{\EnvX_{\Right}}{\EnvX_2}{\Constraints_2}
  \\
  \combinenv{\EnvX_1}{\EnvX_2}{\EnvX}{\Constraints_3}
}{
  \rtp{
    \CaseShort\Expression{i}{\var}{\Process}
  }{
    \EnvX
  }{
    \Constraints_1 \cup \Constraints_2 \cup \Constraints_3 \cup
    \Constraints_{\Left} \cup \Constraints_{\Right}
    \cup \set{ \TypeExpr \con= \TypeExpr_{\Left} \tsum \TypeExpr_{\Right} }
  }
}
\\\\
\inferrule[\defrule{i-weak}]{
  \rtp{\Process}{\EnvX}{\Constraints}
}{
  \rtp{\Process}{\EnvX, \Name : \tvar}{
    \Constraints \cup \set{ \unlimited\tvar }
  }
}
\end{array}
\end{math}
}
\caption{\label{tab:generation} Constraint generation for expressions
  and processes.}
\end{table}

The rules of the type reconstruction algorithm are presented in
Table~\ref{tab:generation} and derive judgments
$\rte{\Expression}{\TypeExpr}{\EnvX}{\Constraints}$ for expressions and
$\rtp{\Process}{\EnvX}{\Constraints}$ for processes.
In both cases, $\EnvX$ is the generated environment that contains associations
for all the free names in $\Expression$ and $\Process$, while $\Constraints$
is the set of constraints that must hold in order for $\Expression$ or
$\Process$ to be well typed in $\EnvX$.
In a judgment $\rte{\Expression}{\TypeExpr}{\EnvX}{\Constraints}$, the type
expression $\TypeExpr$ denotes the type of the expression $\Expression$.

There is a close correspondence between the type system
(Table~\ref{tab:type_system}) and the reconstruction algorithm
(Table~\ref{tab:generation}). In a nutshell, unknown uses and types become
fresh use and type variables (all use/type variables introduced by the rules
are assumed to be fresh), every application of $+$ in
Table~\ref{tab:type_system} becomes an application of $\combineop$ in
Table~\ref{tab:generation}, and every assumption on the form of types becomes
a constraint. Constraints accumulate from the premises to the conclusion of
each rule of the reconstruction algorithm, which we now review briefly.

Rule~\refrule{i-int} deals with integer constants. Their type is
obviously $\tint$, they contain no free names and therefore they
generate the empty environment and the empty set of constraints.
Rule~\refrule{i-name} deals with the free occurrence of a name $\Name$. A
fresh type variable standing for the type of \emph{this occurrence} of $\Name$
is created and used in the resulting type environment $\Name : \tvar$. Again,
no constraints are generated. In general, different occurrences of the same
name may have different types which are eventually combined with $\tvar$ later
on in the reconstruction process.
In rules~\refrule{i-inl} and~\refrule{i-inr} the type of the summand
that was guessed in \refrule{t-inl} and \refrule{t-inr} becomes a
fresh type variable.
Rule \refrule{t-pair} creates a product type from the type of the
components of the pairs, combines the corresponding environments and
joins all the constraints generated in the process.
Rules~\refrule{i-fst} and~\refrule{i-snd} deal with pair
projections. The type $\TypeExpr$ of the projected expression must be
a product of the form $\tvarA \times \tvarB$. Since the first
projection discards the second component of a pair, $\tvarB$ must be
unlimited in \refrule{i-fst}. Symmetrically for \refrule{i-snd}.

Continuing on with the rules for processes, let us consider
\refrule{i-in} and \refrule{i-out}. The main difference between these
rules and the corresponding ones \refrule{t-in} and \refrule{t-out} is
that the use information of the channel on which the communication
occurs is unknown, hence it is represented using fresh use
variables. The $1 + \uvar_i$ part accounts for the fact that the
channel is being used at least once, for an input or an output. The
$2\uvar_j$ part accounts for the fact that the use information
concerning the capability (either input or output) that is not
exercised must be unlimited (note that we extend the notation $2\Use$
to use expressions).
Rule \refrule{i-rep} deals with a replicated process
$\bang\Process$. In the type system, $\bang\Process$ is well typed in
an unlimited environment. Here, we are building up the type
environment for $\bang\Process$ and we do so by combining the
environment $\EnvX$ generated by $\Process$ with itself. The rationale
is that $\EnvX \combineop \EnvX$ yields an unlimited type environment
that grants at least all the capabilities granted by $\EnvX$.
By now most of the main ingredients of the constraint generation
algorithm have been revealed, and the remaining rules contain no
further novelties but the expected use of the merging operator
$\mergeop$ in \refrule{i-case}.
There is, however, a rule \refrule{i-weak} that has no correspondence
in Table~\ref{tab:type_system}. This rule is necessary because
\refrule{i-in}, \refrule{i-new}, and \refrule{i-case}, which
correspond to the binding constructs of the calculus, \emph{assume}
that the names they bind do occur in the premises on these rules. But
since type environments are generated by the algorithm as it works
through an expression or a process, this may not be the case if a
bound name is never used and therefore never occurs in that expression
or process. Furthermore, the $\mergeop$ operator is defined only on
type environments having the same domain. This may not be the case if
a name occurs in only one branch of a pattern matching, and not in the
other one. With rule \refrule{i-weak} we can introduce missing names
in type environments wherever \revised{this is necessary}. Naturally,
an unused name has an unknown type $\tvar$ that must be unlimited,
whence the constraint $\unlimited\tvar$ \revised{(see
  Example~\ref{ex:simple} for an instance where \refrule{i-weak} is
  necessary)}.
Strictly speaking, with \refrule{i-weak} this set of rules is not
syntax directed, which in principle is a problem if we want to
consider this as an algorithm. In practice, the places where
\refrule{i-weak} may be necessary are easy to spot (in the premises of
all the aforementioned rules for the binding constructs). What we gain
with \refrule{i-weak} is a simpler presentation of the rules for
constraint generation.

\subsection{Correctness and completeness}
If the constraint set generated from $\Process$ is satisfiable, then it
corresponds to a typing for $\Process$. To formalize this property, we must
first define what ``satisfiability'' means for a constraint set.

A \emphdef{\assignment} $\Solution$ is a finite map from type
variables to types and from use variables to uses. We write
$\dom(\Solution)$ for the set of type and use variables for which
there is an association in $\Solution$.
\finalrevision{The \emphdef{application} of a substitution $\Solution$
  to a use/type expression $\UseExpr$/$\TypeExpr$, respectively
  denoted by $\Solution\UseExpr$ and $\Solution\TypeExpr$, replaces
  use variables $\uvar$ and type variables $\tvar$ in
  $\UseExpr$/$\TypeExpr$ with the corresponding uses
  $\Solution(\uvar)$ and types $\Solution(\tvar)$ and computes use
  combinations whenever possible:}
\[
\Solution\UseExpr
\eqdef
\begin{cases}
  \Solution(\uvar) & \text{if $\UseExpr = \uvar \in \dom(\Solution)$}
  \\
  \Solution\UseExpr_1 + \Solution\UseExpr_2 & \text{if $\UseExpr =
    \UseExpr_1 + \UseExpr_2$}
  \\
  \UseExpr & \text{otherwise}
\end{cases}
\qquad
\Solution\TypeExpr
\eqdef
\begin{cases}
  \Solution(\tvar) & \text{if $\TypeExpr = \tvar \in \dom(\Solution)$}
  \\
  \tchan{\Solution\TypeExprS}{\Solution\UseExprU}{\Solution\UseExprV}
  & \text{if $\TypeExpr = \tchan\TypeExprS\UseExprU\UseExprV$} \\
  \Solution\TypeExpr_1 \odot \Solution\TypeExpr_2 & \text{if
    $\TypeExpr = \TypeExpr_1 \odot \TypeExpr_2$}
  \\
  \TypeExpr & \text{otherwise}
\end{cases}
\]
\finalrevision{We will make sure that the application of a
  substitution $\Solution$ to a type expression $\TypeExpr$ is always
  well defined: either $\dom(\Solution)$ contains no type variables,
  in which case $\Solution\TypeExpr$ is a type expression, or
  $\dom(\Solution)$ includes all use/type variables occurring in
  $\TypeExpr$, in which case we say that $\Solution$ \emphdef{covers}
  $\TypeExpr$ and $\Solution\TypeExpr$ is a type.
}
We extend application pointwise to type environments, namely
$\Solution\EnvX \eqdef \set{ \Name : \Solution\EnvX(\Name) \mid \Name
  \in \dom(\EnvX) }$, \finalrevision{and we say that $\Solution$
  \emphdef{covers} $\EnvX$ if it covers all the type expressions in
  the range of $\EnvX$.}

\revised{
\begin{definition}[solution, satisfiability, equivalence]
\label{def:solution}
A \assignment $\Solution$ is a \emphdef{solution} for a constraint set
$\Constraints$ if it covers all the $\TypeExpr\in\expr(\Constraints)$
and the following conditions hold:
\begin{itemize}
\item $\TypeExprT \ceq \TypeExprS \in \Constraints$ implies
  $\Solution\TypeExprT \meq \Solution\TypeExprS$, and

\item $\TypeExprT \ceq \TypeExprS_1 + \TypeExprS_2 \in \Constraints$
  implies $\Solution\TypeExprT \meq \Solution\TypeExprS_1 +
  \Solution\TypeExprS_2$, and

\item $\TypeExprT \con\mall \TypeExprS \in \Constraints$ implies
  $\Solution\TypeExprT \mall \Solution\TypeExprS$, and

\item $\UseExprU \ceq \UseExprV \in \Constraints$ implies
  $\Solution\UseExprU \meq \Solution\UseExprV$.
\end{itemize}

We say that $\Constraints$ is \emphdef{satisfiable} if it has a
solution and \emphdef{unsatisfiable} otherwise.
We say that $\Constraints_1$ and $\Constraints_2$ are
\emphdef{equivalent} if they have the same solutions.
\end{definition}
}



We can now state the correctness result for the type reconstruction
algorithm:

\begin{theorem}
\label{thm:correctness_processes}
If $\rtp\Process\EnvX\Constraints$ and $\Solution$ is a solution for
$\Constraints$ that covers $\EnvX$, then $\wtp{\Solution\EnvX}\Process$.
\end{theorem}

Note that Theorem~\ref{thm:correctness_processes} not only requires
$\Solution$ to be a solution for $\Constraints$, but also that
$\Solution$ must include suitable \assignments for all use and type
variables occurring in $\EnvX$. Indeed, it may happen that $\EnvX$
contains use/type variables not involved in any constraint in
$\Constraints$, therefore a solution for $\Constraints$ does not
necessarily cover $\EnvX$.

The reconstruction algorithm is also complete, in the sense that
\emph{each} type environment $\Env$ such that $\wtp{\Env}{\Process}$
can be obtained by applying a solution for $\Constraints$ to $\EnvX$.

\begin{theorem}
\label{thm:completeness_processes}
\revised{If $\wtp{\Env}{\ProcessP}$, then there exist $\EnvX$,
  $\Constraints$, and $\Solution$ such that
  $\rtp\ProcessP{\EnvX}{\Constraints}$ and $\Solution$ is a solution
  for $\Constraints$ that covers $\EnvX$ and $\Env = \Solution\EnvX$.}
\end{theorem}

\revised{
\begin{example}
\label{ex:simple}
Below we illustrate the reconstruction algorithm at work on the process
\[
\new\Channel\ttparens{
  \send\Channel{3}
  \parop
  \receive\Channel\var\idle
}
\]
which will be instrumental also in the following section:
\[
\begin{prooftree}
  \pt{
    \pt{
      \pt{
        \justifies
        \rte{
          \Channel
        }{
          \tvar_1
        }{
          \Channel : \tvar_1
        }{
          \emptyset
        }
      }
      ~~
      \pt{
        \justifies
        \rte{
          3
        }{
          \tint
        }{
          \emptyset
        }{
          \emptyset
        }
      }
      \justifies
      \rtp{
        \send\Channel{3}
      }{
        \Channel : \tvar_1
      }{
        \set{ \tvar_1 \ceq \tchan\tint{2\uvar_1}{1+\uvar_2} }
      }
      \using\refrule{i-out}
    }
    ~~
    \pt{
      \pt{
        \justifies
        \rte{
          \Channel
        }{
          \tvar_2
        }{
          \Channel : \tvar_2
        }{
          \emptyset
        }
      }
      ~~
      \pt{
        \pt{
          \justifies
          \rtp{
            \idle
          }{
            \emptyset
          }{
            \emptyset
          }
          \using\refrule{i-idle}
        }
        \justifies
        \rtp{
          \idle
        }{
          \var : \tvarC
        }{
          \set{ \unlimited\tvarC }
        }
        \using\refrule{i-weak}
      }
      \justifies
      \rtp{
        \receive\Channel\var\idle
      }{
        \Channel : \tvar_2
      }{
        \set{ \tvar_2 \ceq \tchan\tvarC{1+\uvar_3}{2\uvar_4}, \unlimited\tvarC }
      }
      \using\refrule{i-in}
    }
    \justifies
    \rtp{
      \send\Channel{3}
      \parop
      \receive\Channel\var\idle
    }{
      \Channel : \tvar
    }{
      \set{
        \tvar \ceq \tvar_1 + \tvar_2,
        \tvar_1 \ceq \tchan\tint{2\uvar_1}{1+\uvar_2},
        \tvar_2 \ceq \tchan\tvarC{1+\uvar_3}{2\uvar_4},
        \unlimited\tvarC
      }
    }
    \using\refrule{i-par}
  }
  \justifies
  \rtp{
    \new\Channel\ttparens{
      \send\Channel{3}
      \parop
      \receive\Channel\var\idle
    }
  }{
    \emptyset
  }{
    \set{
      \tvar \ceq \tchan\tvarD{\uvar_5}{\uvar_5},
      \tvar \ceq \tvar_1 + \tvar_2,
      \dots
    }
  }
  \using\refrule{i-new}
\end{prooftree}
\]
The synthesized environment is empty, since the process has no free
names, and the resulting constraint set is
\[
\set{
  \tvar \ceq \tchan\tvarD{\uvar_5}{\uvar_5},
  \tvar \ceq \tvar_1 + \tvar_2,
  \tvar_1 \ceq \tchan\tint{2\uvar_1}{1+\uvar_2},
  \tvar_2 \ceq \tchan\tvarC{1+\uvar_3}{2\uvar_4},
  \unlimited\tvarC
}
\]
Observe that $\Channel$ is used twice and each occurrence is assigned
a distinct type variable $\tvar_i$. Eventually, the reconstruction
algorithm finds out that the \emph{same} channel $\Channel$ is used
simultaneously in different parts of the process, so it records the
fact that the overall type $\tvar$ of $\Channel$ must be the
combination of $\tvar_1$ and $\tvar_2$ in the constraint $\tvar \ceq
\tvar_1 + \tvar_2$.

A solution for the obtained constraint set is the substitution
\[
  \set{
    \tvar \mapsto \tchan\tint11,
    \tvar_1 \mapsto \tchan\tint01,
    \tvar_2 \mapsto \tchan\tint10,
    \tvarC \mapsto \tint,
    \tvarD \mapsto \tint,
    \uvar_{1..4} \mapsto 0,
    \uvar_5 \mapsto 1
  }
\]
confirming that $\Channel$ is a linear channel. This is not the only
solution of the constraint set: another one can be obtained by setting
all the use variables to $\omega$, although in this case $\Channel$ is
not recognized as a linear channel.

Note also that the application of \refrule{i-in} is possible only if
the name $\var$ of the received message occurs in the environment
synthesized for the continuation process $\idle$. Since the
continuation process contains no occurrence of $\var$, this name can
only be introduced using \refrule{i-weak}.
In general, \refrule{i-weak} is necessary to prove the completeness of
the reconstruction algorithm as stated in
Theorem~\ref{thm:completeness_processes}. For example, $\wtp{\var :
  \tint}\idle$ is derivable according to the rules in
Table~\ref{tab:type_system}, but as we have seen in the above
derivation the reconstruction algorithm without~\refrule{i-weak} would
synthesize for $\idle$ an empty environment, not containing an
association for $\var$.
\eoe
\end{example}
}


\begin{example}
\label{ex:simple_generation}
We compute the constraint set of a simple process that accesses the same
composite structure containing linear values. The process in
Example~\ref{ex:evenodd} is too large to be discussed in full, so we
consider the following, simpler process
\[
  \receive{\First[\varX]}\varY
  \send{\Second[\varX]}{\ttparens{\varY \mathbin\ttplus 1}}
\]
which uses a pair $\varX$ of channels and sends on the second channel in the
pair the successor of the number received from the first channel (we assume
that the language and the type reconstruction algorithm have been extended in
the obvious way to support operations on numbers such as addition).
We derive
\[
\begin{prooftree}
  \[
    \justifies
    \rte{
      \varX
    }{
      \tvarA_1
    }{
      \varX : \tvarA_1
    }{
      \emptyset
    }
    \using \refrule{i-name}
  \]
  \justifies
  \rte{\First[\varX]}{
    \tvarB_1
  }{
    \varX : \tvarA_1
  }{
    \set{ \tvarA_1 \ceq \tvarB_1 \times \tvarB_2, \unlimited{\tvarB_2} }
  }
  \using \refrule{i-fst}
\end{prooftree}
\]
for the first projection of $\varX$ and
\[
\begin{prooftree}
  \[
    \justifies
    \rte{
      \varX
    }{
      \tvarA_2
    }{
      \varX : \tvarA_2
    }{
      \emptyset
    }
    \using \refrule{i-name}
  \]
  \justifies
  \rte{\Second[\varX]}{
    \tvarC_2
  }{
    \varX : \tvarA_2
  }{
    \set{ \tvarA_2 \ceq \tvarC_1 \times \tvarC_2, \unlimited{\tvarC_1} }
  }
  \using \refrule{i-snd}
\end{prooftree}
\]
for the second projection of $\varX$. For the output operation we derive
\[
\begin{prooftree}
  \vdots
  \qquad
  \[
    \[
      \justifies
      \rte{
        \varY
      }{
        \tvarD
      }{
        \varY : \tvarD
      }{
        \emptyset
      }
      \using\refrule{i-name}
    \]
    \qquad
    \[
      \justifies
      \rte{
        1
      }{
        \tint
      }{
        \emptyset
      }
      \using\refrule{i-int}
    \]
    \justifies
    \rte{
      \varY \mathbin\ttplus 1
    }{
      \tint
    }{
      \varY : \tvarD
    }{
      \set{ \tvarD \ceq \tint }
    }
  \]
  \justifies
  \rtp{
    \send{\Second[\varX]}{\ttparens{\varY \mathbin\ttplus 1}}
  }{
    \varX : \tvarA_2,
    \varY : \tvarD
  }{
    \set{
      \tvarA_2 \ceq \tvarC_1 \times \tvarC_2,
      \unlimited{\tvarC_1},
      \tvarC_2 \ceq \tchan\tint{2\uvar_3}{1+\uvar_4},
      \tvarD \ceq \tint
    }
  }
  \using\refrule{i-out}
\end{prooftree}
\]
so for the whole process we obtain
\[
\begin{prooftree}
  \vdots
  \justifies
  \rtp{
    \receive{\First[\varX]}\varY
    \send{\Second[\varX]}{\ttparens{\varY \mathbin\ttplus 1}}
  }{
    \varX : \tvarA
  }{
    \set{
      \begin{lines}
        \tvarA \ceq \tvarA_1 + \tvarA_2,
        \tvarA_1 \ceq \tvarB_1 \times \tvarB_2,
        \tvarA_2 \ceq \tvarC_1 \times \tvarC_2, \\
        \tvarB_1 \ceq \tchan\tvarD{1+\uvar_1}{2\uvar_2},
        \tvarC_2 \ceq \tchan\tint{2\uvar_3}{1+\uvar_4}, \\
        \unlimited{\tvarB_2},
        \unlimited{\tvarC_1},
        \tvarD \ceq \tint
    }
      \end{lines}
  }
  \using\refrule{i-in}
\end{prooftree}
\]
\revised{Like in Example~\ref{ex:simple}, here too the variable
  $\varX$ is used multiple times and each occurrence is assigned a
  distinct type variable $\tvarA_i$, but this time such type variables
  must be assigned with a pair type in order for the constraint set to
  be solved.}
\eoe
\end{example}


\newcommand{\decompose}{\mathsf{dec}}
\newcommand{\saturate}{\mathsf{sat}}
\newcommand{\failure}{\textbf{fail}}
\newcommand{\recursion}{\mathsf{rec}}
\newcommand{\dtv}{\mathsf{dtv}}
\newcommand{\entails}{\vDash}
\newcommand{\deft}{\mathsf{def}}
\newcommand{\supt}[1][\Constraints,\Solution]{\mathsf{sup}_{#1}}
\newcommand{\zerot}[1][\Constraints,\Solution]{\mathsf{zero}_{#1}}
\newcommand{\relt}[1]{\mathsf{rel}_{#1}}
\newcommand{\closuret}[2][\Constraints]{\mathsf{cls}_{#2,#1}}
\newcommand{\propert}{\mathsf{proper}}
\newcommand{\rect}{\mathsf{rec}}

\newcommand{\saturated}[1]{\overline{#1}}

\newcommand{\ConstraintsIn}{\Constraints}
\newcommand{\ConstraintsSat}{\saturated\Constraints}

\newcommand{\SolutionU}{\Solution_{\mathit{use}}}
\newcommand{\SolutionT}{\Solution_{\mathit{type}}}

\newcommand{\composable}{$\Constraints$-composable}
\newcommand{\compatible}{$\Constraints$-compatible}

\section{Constraint Solving}
\label{sec:solver}

In this section we describe an algorithm that determines whether a
given constraint set $\Constraints$ is satisfiable and, if this is the
case, computes a solution for $\Constraints$. \revised{Among all
  possible solutions for $\Constraints$, we strive to find one that
  allows us to identify as many linear channels as possible. To this
  aim, it is convenient to recall the notion of solution preciseness
  from~\cite{IgarashiKobayashi00}.}

\revised{
\begin{definition}[solution preciseness]
\label{def:precision}
Let $\prel$ be the total order on uses such that $0 \leq 1 \leq
\omega$. Given two solutions $\Solution_1$ and $\Solution_2$ for a
constraint set $\Constraints$, we say that $\Solution_1$ is
\emphdef{more precise} than $\Solution_2$ if $\Solution_1(\uvar) \leq
\Solution_2(\uvar)$ for every $\uvar\in\expr(\Constraints)$.
\end{definition}
}

\revised{%
  Roughly speaking, the preciseness of a solution is measured in terms
  of the numbers of unused and linear channels it identifies, which
  are related to the number of use variables assigned to $0$ and $1$.
  We will use Definition~\ref{def:precision} as a guideline for
  developing our algorithm, although the algorithm may be unable to
  find \emph{the} most precise solution.
  There are two reasons for this.
  First, there can be solutions with minimal use assignments that are
  incomparable according to Definition~\ref{def:precision}. This is
  related to the fact that the type system presented in
  Section~\ref{sec:types} lacks the principal typing property.
  Second, to ensure termination when constraints concern infinite
  types, our algorithm makes some simplifying assumptions that may --
  in principle -- imply a loss of precision of the resulting solution
  (see Example~\ref{ex:approximation}).
  Despite this, experience with the implementation suggests that the
  algorithm is indeed capable of identifying as many unused and linear
  channels as possible in practical situations, even when infinite
  types are involved.
  Before embarking in the technical description of the algorithm, we
  survey the key issues that we have to address and how they are
  addressed.}

\subsection{Overview}
\label{sec:overview}
We begin by considering again the simple process below
\begin{equation}
\label{eq:discuss1}
\new\Channel\ttparens{\send\Channel3 \parop \receive\Channel\var\idle}
\end{equation}
\revised{for which we have shown the reconstruction algorithm at work
  in Example~\ref{ex:simple}. The process} contains three occurrences
of the channel $\Channel$, two of them in subject position for
input/output operations and one binding occurrence in the
$\mkkeyword{new}$ construct. We have seen that the constraint
generation algorithm associates the two rightmost occurrences of
$\Channel$ with two type variables $\tvar_1$ and $\tvar_2$ that must
respectively satisfy the constraints
\begin{eqnarray}
  \tvar_1 & \con= & \tchan\tint{2\uvar_1}{1+\uvar_2}
  \label{eq:tvar1}
  \\
  \tvar_2 & \con= & \tchan\tvarC{1+\uvar_3}{2\uvar_4}
  \label{eq:tvar2}
\end{eqnarray}
whereas the leftmost occurrence of $\Channel$ has a type $\tvar$
which must satisfy the constraints
\begin{eqnarray}
  \tvar & \con= & \tvar_1 + \tvar_2
  \label{eq:tvar_comp}
  \\
  \tvar & \con= & \tchan\tvarD{\uvar_5}{\uvar_5}
  \label{eq:tvar}
\end{eqnarray}

%

\revised{%
  Even if none of these constraints concerns use variables directly,
  use variables are subject to implicit constraints that should be
  taken into account for finding a precise solution. To expose such
  implicit constraints, observe that} in this first example we are in
the fortunate situation where \revised{the type variables $\tvar$,
  $\tvar_1$, and $\tvar_2$ occur on the left-hand side of a constraint
  of the form} $\tvarB \con= \TypeExpr$ where $\TypeExpr$ is different
from a type variable.  \revised{In this case we say that $\tvarB$ is
  \emph{defined} and we call $\TypeExpr$ its \emph{definition}. If we
  substitute each type variable in \eqref{eq:tvar_comp} with its
  definition we obtain}
\begin{equation*}
  \tchan\tvarD{\uvar_5}{\uvar_5} \con= \tchan\tint{2\uvar_1}{1+\uvar_2}
  + \tchan\tvarC{1+\uvar_3}{2\uvar_4}
\end{equation*}
that reveals the relationships between the use variables. Knowing how
type combination \revised{operates} (Definition~\ref{def:tand}), we
can derive two constraints concerning use variables
\begin{eqnarray*}
  \uvar_5 & \con= & 2\uvar_1 + 1+\uvar_3
  \\
  \uvar_5 & \con= & 1+\uvar_2 + 2\uvar_4
\end{eqnarray*}
for which it is easy to figure out a solution that includes the
substitutions $\set{ \uvar_{1..4} \mapsto 0, \uvar_5 \mapsto 1 }$ (see
Example~\ref{ex:simple}).  No substitution \revised{can be more
  precise} than this one hence such solution, which identifies
$\Channel$ as a linear channel, is in fact optimal.

Let us now consider the following variation of \eqref{eq:discuss1}
\begin{equation*}
\send\Channel3 \parop \receive\Channel\var\idle
\end{equation*}
where we have removed the restriction.
In this case the generated constraints are the same \eqref{eq:tvar1},
\eqref{eq:tvar2}, and \eqref{eq:tvar_comp} as above, except that there
is no constraint \eqref{eq:tvar} that provides a definition for
$\tvar$. In a sense, $\tvar$ \emph{is} defined because we know that it
must be the combination of $\tvar_1$ and $\tvar_2$ for which we do
have definitions. However, \revised{in order to come up with a general
  strategy for solving constraint sets,} it is convenient to
\emph{complete} the constraint set with a defining equation for
$\tvar$: we know that $\tvar$ must be a channel type with messages of
type $\tint$, because that is the shape of the definition for
$\tvar_1$, but we do not know precisely the overall uses of
$\tvar$. Therefore, we generate a new constraint defining the
\emph{structure} of the type $\tvar$, but with fresh use variables
$\uvar_5$ and $\uvar_6$ in place of the unknown uses:
\begin{equation*}
\tvar \con= \tchan\tint{\uvar_5}{\uvar_6}
\end{equation*}

We can now proceed as before, by substituting all type variables in
\eqref{eq:tvar_comp} with their definition and deriving the use
constraints below:
\begin{eqnarray*}
  \uvar_5 & \con= & 2\uvar_1 + 1+\uvar_3
  \\
  \uvar_6 & \con= & 1+\uvar_2 + 2\uvar_4
\end{eqnarray*}

Note that, unlike in \eqref{eq:discuss1}, we do not know whether
$\uvar_5$ and $\uvar_6$ are required to be equal or not. Here we are
typing an open process which, in principle, may be composed in
parallel with other uses of the same channel $\Channel$.
Nonetheless, we can easily find a solution analogous to the previous one
but with the use assignments $\set{ \uvar_{1..5} \mapsto 0,
  \uvar_{5,6} \mapsto 1 }$.

The idea of completing constraints with missing definitions is a
fundamental ingredient of our constraint solving technique. In the
previous example, completion was somehow superfluous \revised{because
  we could have obtained a definition for $\tvar$ by combining the
  definitions of $\tvar_1$ and $\tvar_2$, which were
  available. However, completion allowed us to patch the constraint
  set so that it could be handled as in the previous case of
  process~\eqref{eq:discuss1}. In fact,}
it is easy to find processes for which completion becomes
essential. Consider for example
\begin{equation}
\label{eq:discuss3}
\new\Channel\ttparens{\send\Channel3 \parop \send\ChannelB\Channel}
\end{equation}
where the bound channel $\Channel$ is used once for an output and then
extruded through a free channel $\ChannelB$. \revised{For this
  process, the reconstruction algorithm infers the type environment
  $\ChannelB : \tvarB$ and the constraints below: }
\begin{eqnarray*}
  \tvar_1 & \con= & \tchan\tint{2\uvar_1}{1+\uvar_2}
  \\
  \tvarB & \con= & \tchan{\tvar_2}{2\uvar_3}{1+\uvar_4}
  \\
  \tvar & \con= & \tvar_1 + \tvar_2
  \\
  \tvar & \con= & \tchan\tvarD{\uvar_5}{\uvar_5}
\end{eqnarray*}
\revised{where the three occurrences of $\Channel$ are associated from
  left to right with the type variables $\tvar$, $\tvar_1$, and
  $\tvar_2$ (Section~\ref{sec:extra_solver} gives the derivation for
  \eqref{eq:discuss3}).}
Note that there is no constraint that defines $\tvar_2$. In fact, there is
just no constraint with $\tvar_2$ on the left hand side at all. The only hint
that we have concerning $\tvar_2$ is that it must yield $\tvar$ when combined
with $\tvar_1$. Therefore, according to the definition of type combination, we
can once more deduce that $\tvar_2$ shares the same structure as $\tvar$ and
$\tvar_1$ and we can complete the set of constraints with
\begin{equation*}
  \tvar_2 \con= \tchan\tint{\uvar_6}{\uvar_7}
\end{equation*}
where $\uvar_6$ and $\uvar_7$ are fresh use variables.

After performing the usual substitutions, we can finally derive the
use constraints
\begin{eqnarray*}
  \uvar_5 & \con= & 2\uvar_1 + \uvar_6
  \\
  \uvar_5 & \con= & 1 + \uvar_2 + \uvar_7
\end{eqnarray*}
for which we find a solution including the assignments $\set{
  \uvar_{1..4,7} \mapsto 0, \uvar_{5,6} \mapsto 1 }$.
The interesting fact about this solution is the \assignment $\uvar_6
\mapsto 1$, meaning that the constraint solver has inferred an input
operation for the rightmost occurrence of $\Channel$ in
\eqref{eq:discuss3}, even though there is no explicit evidence of this
operation in the process itself. The input operation is deduced ``by
subtraction'', seeing that $\Channel$ is used once in
\eqref{eq:discuss3} for an output operation and knowing that a
restricted (linear) channel like $\Channel$ must also be used for a
matching input operation.

Note also that this is not the only possible solution for the use
constraints. If, for example, it turns out that the extruded
occurrence of $\Channel$ is never used (or is used twice) for an
input, it is possible to obtain various solutions that include the
assignments $\set{ \uvar_{5,6} \mapsto \omega}$. However, the solution
we have found above is the \revised{most precise according to
  Definition~\ref{def:precision}}.

It is not always possible to find the most precise solution. This can
be seen in the following variation of \eqref{eq:discuss3}
\begin{equation}
\label{eq:discuss4}
\new\Channel\ttparens{\send\Channel3 \parop
  \send\ChannelB\Channel \parop \send\ChannelC\Channel}
\end{equation}
where $\Channel$ is extruded twice, on $\ChannelB$ and on $\ChannelC$
\revised{(Section~\ref{sec:extra_solver} gives the derivation)}. Here,
as in \eqref{eq:discuss3}, an input use for $\Channel$ is deduced ``by
subtraction'', but there is an ambiguity as to whether such input
capability is transmitted through $\ChannelB$ or through
$\ChannelC$. Hence, there exist two incomparable solutions for the
constraint set generated for~\eqref{eq:discuss3}.
The lack of an optimal solution in general (hence of a principal
typing) is a consequence of the condition imposing equal uses for
restricted channels (see \refrule{t-new} and \refrule{i-new}). Without
this condition, it would be possible to find the most precise solution
for the constraints generated by \eqref{eq:discuss4} noticing that
$\Channel$ is \emph{never} explicitly used for an input operation, and
therefore its input use could be 0. We think that this approach
hinders the applicability of the reconstruction algorithm in practice,
where separate compilation and type reconstruction of large programs
are real concerns. We will elaborate more on this in
Example~\ref{ex:filter}.
For the time being, let us analyze one last example showing a feature
that we do \emph{not} handle in our type system, namely
polymorphism. The process
\begin{equation}
\label{eq:discuss5}
\receive\ChannelA\var
\send\ChannelB\var
\end{equation}
models a forwarder that receives a message $\var$ from $\ChannelA$ and
sends it on $\ChannelB$. \revised{For this process the constraint
  generation algorithm yields the environment $\ChannelA : \tvarA,
  \ChannelB : \tvarB$ and the constraints}
\begin{eqnarray*}
  \tvarA & \con= & \tchan\tvarC{1+\uvar_1}{2\uvar_2} \\
  \tvarB & \con= & \tchan\tvarC{2\uvar_3}{1+\uvar_4}
\end{eqnarray*}
\revised{(Section~\ref{sec:extra_solver} gives the complete derivation)}.
In particular, there is no constraint concerning the type variable $\tvarC$
and for good reasons: since the message $\var$ is only passed around in
\eqref{eq:discuss5} but never actually used, the channels $\ChannelA$ and
$\ChannelB$ should be considered \emph{polymorphic}.  Note that in this case
we know nothing about the structure of $\tvarC$ hence completion of the
constraint set is not applicable. In this work we do not deal with
polymorphism and will refrain from solving sets of constraints where there is
no (structural) information for unconstrained type variables. Just observe
that handling polymorphism is not simply a matter of allowing (universally
quantified) type variables in types. For example, a type variable involved in
a constraint $\tvar \con= \tvar + \tvar$ does not have any structural
information and therefore is polymorphic, but can only be instantiated with
unlimited types.
The implementation has a defaulting mechanism that forces
unconstrained type variables to a base type.

We now formalize the ideas presented so far into an algorithm, for
which we have already identified the key phases: the ability to
recognize types that ``share the same structure'', which we call
\emph{structurally coherent} (Definition~\ref{def:mall}); the
\emph{completion} of a set of constraints with ``missing definitions''
so that each type variable has a proper definition; the derivation and
solution of \emph{use constraints}. Let us proceed in order.

\subsection{Verification}
\begin{table}
\framebox[\textwidth]{
\begin{math}
\displaystyle
\begin{array}{@{}c@{}}
  \inferrule[\defrule{c-axiom}]{}{
    \ded{\Constraints\cup\set\Constraint}\Constraint
  }
  \qquad
  \inferrule[\defrule{c-refl}]{
    \TypeExpr \in \expr(\Constraints)
  }{
    \ded\Constraints{\TypeExpr \con\rrel \TypeExpr}
  }
  \qquad
  \inferrule[\defrule{c-symm}]{
    \ded\Constraints{\TypeExprT \con\rrel \TypeExprS}
  }{
    \ded\Constraints{\TypeExprS \con\rrel \TypeExprT}
  }
  \qquad
  \inferrule[\defrule{c-trans}]{
    \ded\Constraints{\TypeExprT \con\rrel \TypeExprT'}
    \\
    \ded\Constraints{\TypeExprT' \con\rrel \TypeExprS}
  }{
    \ded\Constraints{\TypeExprT \con\rrel \TypeExprS}
  }
  \\\\
  \inferrule[\defrule{c-coh 1}]{
    \ded\Constraints{\TypeExprT \con\meq \TypeExprS}
  }{
    \ded\Constraints{\TypeExprT \con\mall \TypeExprS}
  }
  \qquad
  \inferrule[\defrule{c-coh 2}]{
    \ded\Constraints{\TypeExprT \con\meq \TypeExprS_1 + \TypeExprS_2}
  }{
    \ded\Constraints{\TypeExprT \con\mall \TypeExprS_i}
  }
  \quad
  i\in\set{1,2}
  \qquad
  \inferrule[\defrule{c-cong 1}]{
    \ded\Constraints{
      \tchan\TypeExprT{\UseExprU_1}{\UseExprU_2}
      \con\mall
      \tchan\TypeExprS{\UseExprV_1}{\UseExprV_2}
    }
  }{
    \ded\Constraints{
      \TypeExprT \con\meq \TypeExprS
    }
  }
  \\\\
  \inferrule[\defrule{c-cong 2}]{
    \ded\Constraints{
      \TypeExprT_1 \odot \TypeExprT_2
      \con\rrel
      \TypeExprS_1 \odot \TypeExprS_2
    }
  }{
    \ded\Constraints{\TypeExprT_i \con\rrel \TypeExprS_i}
  }
  \quad
  i\in\set{1,2}
  \qquad
  \inferrule[\defrule{c-cong 3}]{
    \ded\Constraints{
      \TypeExprT_1 \odot \TypeExprT_2
      \con=
      \TypeExprS_1 \odot \TypeExprS_2
      +
      \TypeExprS_3 \odot \TypeExprS_4
    }
  }{
    \ded\Constraints{\TypeExprT_i \con= \TypeExprS_i + \TypeExprS_{i+2}}
  }
  \quad
  i\in\set{1,2}
  \\\\
  \inferrule[\defrule{c-subst}]{
    \ded\Constraints{\TypeExprT_1 \con= \TypeExprT_2 + \TypeExprT_3}
    \\
    \ded\Constraints{\TypeExprT_i \con= \TypeExprS_i}~{}^{(1\le i\le 3)}
  }{
    \ded\Constraints{\TypeExprS_1 \con= \TypeExprS_2 + \TypeExprS_3}
  }
  \\\\
\revised{
  \displaystyle
  \inferrule[\defrule{c-use 1}]{
    \ded\Constraints{
      \tchan\TypeExprT{\UseExprU_1}{\UseExprU_2}
      \con\meq
      \tchan\TypeExprS{\UseExprV_1}{\UseExprV_2}
    }
  }{
    \ded\Constraints{
      \UseExprU_i
      \con\meq
      \UseExprV_i
    }
  }
  ~
  i\in\set{1,2}
  \qquad
  \inferrule[\defrule{c-use 2}]{
    \ded\Constraints{
      \tchan\TypeExprT{\UseExprU_1}{\UseExprU_2}
      \con\meq
      \tchan{\TypeExprS_1}{\UseExprV_1}{\UseExprV_2}
      +
      \tchan{\TypeExprS_2}{\UseExprV_3}{\UseExprV_4}
    }
  }{
    \ded\Constraints{
      \UseExprU_i
      \con\meq
      \UseExprV_i + \UseExprV_{i+2}
    }
  }
  ~
  i\in\set{1,2}
}
\end{array}
\end{math}
}
\caption{\label{tab:deduction} Constraint deduction system.}
\end{table}

In \finalrevision{Section~\ref{sec:overview}} we have seen that some
constraints can be \emph{derived} from the ones produced during the
constraint generation phase (Section~\ref{sec:generator}).  We now
define a deduction system that, starting from a given constraint set
$\Constraints$, \finalrevision{computes} all the ``derivable facts''
about the types in $\expr(\Constraints)$. Such deduction system is
presented as a set of inference rules in Table~\ref{tab:deduction},
where $\rrel$ ranges over the \revised{symbols} $\meq$ and
$\mall$. Each rule derives a judgment of the form
$\ded\Constraints\Constraint$ \revised{meaning} that the constraint
$\Constraint$ is derivable from those in $\Constraints$
\revised{(Proposition~\ref{prop:ded} below formalizes this property)}.
Rule~\refrule{c-axiom} simply takes each constraint in $\Constraints$
as an axiom. Rules ~\refrule{c-refl}, \refrule{c-symm}, and
\refrule{c-trans} state the obvious reflexivity, symmetry, and
transitivity of $\meq$ and $\mall$.
Rules~\refrule{c-coh 1} and \refrule{c-coh 2} deduce coherence
relations: equality implies coherence, for ${\meq} \subseteq {\mall}$,
and each component of a combination is coherent to the combination
itself (and therefore, by transitivity, to the other component).
Rules~\refrule{c-cong 1} through \refrule{c-cong 3} state congruence
properties of $\meq$ and $\mall$ which follow directly from
Definition~\ref{def:tand}: when two channel types are coherent, their
message types must be equal; corresponding components of
$\rrel$-related composite types are $\rrel$-related.
Rule~\refrule{c-subst} allows the substitution of equal types in
combinations.
Finally, \refrule{c-use 1} and~\refrule{c-use 2} allow us to deduce
\emphdef{use constraints} of the form \revised{$\UseExprU \con=
  \UseExprV$} involving use variables. Both rules are self-explanatory
and follow directly from Definition~\ref{def:tand}.

We state two important properties of this deduction system:

\begin{proposition}
\label{prop:ded}
Let $\ded\Constraints\Constraint$. The following properties hold:
\begin{enumerate}
\item $\Constraints$ and $\Constraints \cup \set{ \Constraint }$
  \revised{are equivalent (see Definition~\ref{def:solution})}.

\item $\expr(\Constraints) = \expr(\Constraints \cup \set{ \Constraint
  })$.
\end{enumerate}
\end{proposition}
\begin{proof}
  A simple induction on the derivation of
  $\ded\Constraints\Constraint$.
\end{proof}

The first property confirms that all the derivable relations are already
encoded in the original constraint set, in a possibly implicit form. The
deduction system makes them explicit.
The second property assures us that no new type expressions are introduced by
the deduction system. Since the inference rules in Section~\ref{sec:generator}
always generate finite constraint sets, this implies that the set of all
derivable constraints is also finite and can be computed in finite time. This
is important because the presence or absence of particular constraints
determines the (un)satisfiability of a constraint set:

\begin{proposition}
\label{prop:unsatisfiable}
If $\ded\Constraints{\TypeExprT \con\mall \TypeExprS}$ where
$\TypeExprT$ and $\TypeExprS$ are proper type expressions with
different topmost constructors, then $\Constraints$ \revised{has no
  solution}.
\end{proposition}
\begin{proof}
  \revised{Suppose by contradiction that $\Solution$ is a solution for
    $\Constraints$. By Proposition~\ref{prop:ded}(1) we have that
    $\Solution$ is also a solution for $\Constraints \cup \set{
      \TypeExprT \con\mall \TypeExprS }$. This is absurd, for if
    $\TypeExprT$ and $\TypeExprS$ have different topmost constructors,
    then so do $\Solution\TypeExprT$ and $\Solution\TypeExprS$, hence
    $\Solution\TypeExprT \not\mall \Solution\TypeExprS$.}
\end{proof}

\revised{The converse of Proposition~\ref{prop:unsatisfiable} is not
  true in general.  For example, the constraint set $\set{
    \tchan\tint01 \con= \tchan\tint10 }$ has no solution because of
  the implicit constraints on corresponding uses and yet it satisfies
  the premises of Proposition~\ref{prop:unsatisfiable}.
  However, when $\Constraints$ is a constraint set generated by the
  inference rules in Section~\ref{sec:generator}, the converse of
  Proposition~\ref{prop:unsatisfiable} holds. This means that we can
  use structural coherence as a necessary and sufficient condition for
  establishing the satisfiability of constraint sets generated by the
  reconstruction algorithm.
}


Before proving this fact we introduce some useful notation.
\revised{For ${\rrel} \in \set{ {\meq}, {\mall} }$ let
\[
  {\can\rrel}
  \eqdef
  \set{
    (\TypeExprT, \TypeExprS)
    \mid
    \ded\Constraints{\TypeExprT \con\rrel \TypeExprS}
  }
\]}%
and observe that $\can\rrel$ is an equivalence relation on
$\expr(\Constraints)$ by construction, because of the rules
\refrule{c-refl}, \refrule{c-symm}, and \refrule{c-trans}. Therefore,
it partitions the type expressions in $\Constraints$ into
$\rrel$-equivalence classes. 
\revised{Now, we need some way to \emph{choose}, from each
  $\rrel$-equivalence class, \emph{one} representative element of the
  class. To this aim, we fix a total order $\orel$ between type
  expressions such that $\TypeExpr \orel \tvar$ for every proper
  $\TypeExpr$ and every $\tvar$ and we define:\footnote{In a Haskell
    or OCaml implementation such total order could be, for instance,
    the one automatically defined for the algebraic data type that
    represents type expressions and where the value constructor
    representing type variables is the last one in the data type
    definition.}  }

\begin{definition}[canonical representative]
\label{def:cr}
Let $\crep\Constraints\TypeExprT\rrel$ be \revised{the $\orel$-least
  type expression $\TypeExprS$ such that $\TypeExprT \can\rrel
  \TypeExprS$. We say that $\crep\Constraints\TypeExprT\rrel$ is} the
\emph{canonical representative} of $\TypeExprT$ with respect to the
relation $\can\rrel$.
\end{definition}

\revised{%
  Note that, depending on $\orel$, we may have different definitions
  of $\crep\Constraints\TypeExprT\rrel$.
  The exact choice of the canonical representative does not affect the
  ability of the algorithm to compute a solution for a constraint set
  (Theorem~\ref{thm:algorithm}) although -- in principle -- it may
  affect the precision of the solution
  (Example~\ref{ex:approximation}).
  Note also that, because of the assumption we have made on the total
  order $\orel$, $\crep\Constraints\TypeExprT\rrel$ is proper whenever
  $\TypeExprT$ is proper or when $\TypeExprT$ is some type variable
  $\tvar$ such that there is a ``definition'' for $\tvar$ in
  $\Constraints$.
  In fact, it is now time to define precisely the notion of
  \emph{defined} and \emph{undefined} type variables:%
}%

\begin{definition}[defined and undefined type variables]
\label{def:defined}
Let
\begin{eqnarray*}
\defined\Constraints\rrel & \eqdef & \set{ \tvar \in
    \expr(\Constraints) \mid \text{$\crep\Constraints\tvar\rrel$ is
      proper} }
\\
\undefined\Constraints\rrel & \eqdef & \set{ \tvar \in
  \expr(\Constraints) \setminus \defined\Constraints\rrel }
\end{eqnarray*}
We say that $\tvar$ is $\rrel$-defined or $\rrel$-undefined in
$\Constraints$ according to $\tvar \in \defined\Constraints\rrel$ or
$\tvar \in \undefined\Constraints\rrel$.
\end{definition}

We can now prove that the coherence check is also a sufficient condition for
satisfiability.

\begin{proposition}
\label{prop:satisfiable}
Let $\rtp\Process\EnvX\Constraints$.  If $\ded\Constraints{\TypeExprT
  \con\mall \TypeExprS}$ where $\TypeExprT$ and $\TypeExprS$ are
proper type expressions implies that $\TypeExprT$ and $\TypeExprS$
have the same topmost constructor, then $\Constraints$ \revised{has a
  solution}.
\end{proposition}
\begin{proof}
  We only sketch the proof, since we will prove a more general result
  later on (see Theorem~\ref{thm:algorithm}). Consider the use
  \assignment $\SolutionU \eqdef \set{ \uvar \mapsto \omega \mid \uvar
    \in \expr(\Constraints) }$ mapping all use variables in
  $\Constraints$ to $\omega$, let $\Upsigma$ be the system of
  equations \revised{$\set{ \tvar_i = \TypeExpr_i \mid 1\leq i\leq n
    }$ defined by}
\[
  \Upsigma \eqdef \set{
    \tvar = \SolutionU \crep\Constraints\tvar\mall
    \mid
    \tvar \in \defined\Constraints\mall
  } \cup \set{
    \tvar = \tint
    \mid
    \tvar \in \undefined\Constraints\mall
  }
\]
\revised{and observe that every $\TypeExpr_i$ is a proper type
  expression.  From Theorem~\ref{thm:courcelle} we know that
  $\Upsigma$ has a unique solution $\SolutionT = \set{ \tvar_i \mapsto
    \Type_i \mid 1\leq i\leq n}$ such that $\Type_i =
  \SolutionT\TypeExpr_i$ for every $1\leq i\leq n$.}
It only remains to show that $\SolutionU \cup \SolutionT$ is a
solution for $\Constraints$. This follows from two facts: (1)
\revised{from the hypothesis $\rtp\Process\EnvX\Constraints$ we know
  that} all channel types in $\Constraints$ have one use variable in
each of their use slots, hence the \assignment $\SolutionU$ forces all
uses to $\omega$; (2) from the hypothesis and the rules
\rulename{c-cong *} we know that all proper type expressions in the
same $(\mall)$-equivalence class have the same topmost constructor.
\end{proof}

\revised{In Proposition~\ref{prop:satisfiable},} for finding a
\assignment for all the type variables in $\Constraints$, we default
each type variable in $\undefined\Constraints\mall$ to $\tint$. This
\assignment is necessary in order to satisfy the constraints
$\unlimited\tvar$, namely those of the form $\tvar \ceq \tvar +
\tvar$, when $\tvar \in \undefined\Constraints\mall$.  These $\tvar$'s
are the ``polymorphic type variables'' that we have already discussed
earlier. Since we leave polymorphism for future work, in the rest of
this section we make the assumption that $\undefined\Constraints\mall
= \emptyset$, namely that all type variables are $(\mall)$-defined.

\begin{example}
\label{ex:simple_verification}
Below is a summary of the constraint set $\Constraints$ generated in
Example~\ref{ex:simple_generation}:
\[
\begin{array}{@{}r@{~}c@{~}l@{}}
  \tvarA & \ceq & \tvarA_1 + \tvarA_2 \\
  \tvarD & \ceq & \tint
\end{array}
\qquad
\begin{array}{@{}r@{~}c@{~}l@{}}
  \tvarA_1 & \ceq & \tvarB_1 \times \tvarB_2 \\
  \tvarA_2 & \ceq & \tvarC_1 \times \tvarC_2 \\
\end{array}
\qquad
\begin{array}{@{}r@{~}c@{~}l@{}}
  \tvarB_1 & \ceq & \tchan\tvarD{1+\uvar_1}{2\uvar_2} \\
  \tvarB_2 & \ceq & \tvarB_2 + \tvarB_2 \\
\end{array}
\qquad
\begin{array}{@{}r@{~}c@{~}l@{}}
  \tvarC_1 & \ceq & \tvarC_1 + \tvarC_1 \\
  \tvarC_2 & \ceq & \tchan\tint{2\uvar_3}{1+\uvar_4} \\
\end{array}
\]
Note that $\set{\tvarA, \tvarB_2, \tvarC_1} \subseteq
\defined\Constraints\mall \setminus \defined\Constraints\meq$. In particular,
they all have a proper canonical representative, which we may assume to be the
following ones:
\[
\begin{array}{r@{~}c@{~}l}
  \crep\Constraints\tvarA\mall = \crep\Constraints{\tvarA_1}\mall =
  \crep\Constraints{\tvarA_2}\mall & = & \tvarB_1 \times \tvarB_2
  \\
  \crep\Constraints{\tvarB_1}\mall = \crep\Constraints{\tvarC_1}\mall
  & = & \tchan\tvarD{1+\uvar_1}{2\uvar_2}
  \\
  \crep\Constraints{\tvarB_2}\mall = \crep\Constraints{\tvarC_2}\mall
  & = & \tchan\tint{2\uvar_3}{1+\uvar_4}
  \\
  \crep\Constraints\tvarD\mall & = & \tint
\end{array}
\]
It is immediate to verify that the condition of
Proposition~\ref{prop:satisfiable} holds, hence we conclude that
$\Constraints$ is satisfiable.
\revised{Indeed, a solution for $\Constraints$ is
\[
  \set{
    \tvarA, \tvarA_{1,2} \mapsto \tchan\tint\omega\omega \times \tchan\tint\omega\omega,
    \tvarB_{1,2}, \tvarC_{1,2} \mapsto \tchan\tint\omega\omega,
    \tvarD \mapsto \tint,
    \uvar_{1..4} \mapsto \omega
  }
\]
even though we will find a more precise solution in
Example~\ref{ex:simple_synthesis}.}
\eoe
\end{example}

\subsection{Constraint set completion}
If the satisfiability of the constraint set is established
(Proposition~\ref{prop:satisfiable}), the subsequent step is its completion in
such a way that every type variable $\tvar$ has a definition in the form of a
constraint $\tvar \ceq \TypeExpr$ where $\TypeExpr$ is proper. Recall that
this step is instrumental for discovering all the (implicit) use constraints.

\revised{In Example~\ref{ex:simple_verification} we} have seen that
some type variables may be $(\mall)$-defined but
$(\meq)$-undefined. The $\mall$ relation provides information about
the structure of the type that should be assigned to the type
variable, but says nothing about the uses in them. Hence, the main
task of completion is the creation of fresh use variables for those
channel types of which only the structure is known. In the process,
fresh type variables need to be created as well, \revised{and we
  should make sure that all such type variables are $(\meq)$-defined
  to guarantee that completion eventually terminates.}
We will be able to do this, possibly at the cost of some precision of
the resulting solution.

We begin the formalization of completion by introducing an injective function
$\mktvar$ that, given a pair of type variables $\tvarA$ and $\tvarB$, creates
a new type variable $\mktvar(\tvarA, \tvarB)$. We assume that $\mktvar(\tvarA,
\tvarB)$ is different from any type variable generated by the algorithm in
Section~\ref{sec:generator} so that the type variables obtained through
$\mktvar$ are effectively fresh.
Then we define an \emphdef{instantiation} function $\instance{}{}$ that, given
a type variable $\tvar$ and a type expression $\TypeExpr$, produces a new type
expression that is structurally coherent to $\TypeExpr$, but where all use
expressions and type variables have been respectively replaced by fresh use
and type variables. The first argument $\tvar$ of $\instance{}{}$ records the
fact that such instantiation is necessary for completing $\tvar$. Formally:
\begin{equation}
\label{eq:instance}
  \instance\tvarA\TypeExpr
  \eqdef
  \begin{cases}
    \mktvar(\tvarA,\tvarB) & \text{if $\TypeExpr = \tvarB$}
    \\
    \tint & \text{if $\TypeExpr = \tint$}
    \\
    \tchan\TypeExprS{\uvar_1}{\uvar_2}
    & \text{if $\TypeExpr = \tchan\TypeExprS\UseExprU\UseExprV$, $\uvar_i$ fresh}
    \\
    \instance\tvarA{\TypeExpr_1} \odot \instance\tvarA{\TypeExpr_2}
    & \text{if $\TypeExpr = \TypeExpr_1 \odot \TypeExpr_2$}
  \end{cases}
\end{equation}

All the equations but the first one are easily explained:
the instance of $\tint$ cannot be anything but $\tint$ itself;
the instance of a channel type $\tchan\TypeExprS\UseExprU\UseExprV$ is the
type expression $\tchan\TypeExprS{\uvar_1}{\uvar_2}$ where we generate two
fresh use variables corresponding to $\UseExprU$ and $\UseExprV$;
the instance of a composite type $\TypeExprT \odot \TypeExprS$ is the
composition of the instances of $\TypeExprT$ and $\TypeExprS$.
For example, we have
\[
  \instance\tvar{
    \tvarB \times
    \tchan{
      \tchan\tint{\UseExprU_1}{\UseExprU_2}
    }{\UseExprV_1}{\UseExprV_2}
  }
  =
  \mktvar(\tvarA,\tvarB) \times
  \tchan{
    \tchan\tint{\UseExprU_1}{\UseExprU_2}
  }{\uvar_1}{\uvar_2}
\]
where $\uvar_1$ and $\uvar_2$ are fresh.
Note that, while instantiating a channel type
$\tchan\TypeExprS\UseExprU\UseExprV$, there is no need to instantiate
$\TypeExprS$ because $\tchan{\TypeT}{\Use_1}{\Use_2} \mall
\tchan{\TypeS}{\Use_3}{\Use_4}$ implies $\TypeT = \TypeS$ so $\TypeExprS$ is
\emph{exactly} the message type we must use in the instance of
$\tchan\TypeExprS\UseExprU\UseExprV$.

Concerning the first equation in \eqref{eq:instance}, in principle we
want $\instance\tvarA\tvarB$ to be the same as
$\instance\tvarA{\crep\Constraints\tvarB\mall}$, but doing so directly
would lead to an ill-founded definition for $\instance{}{}$, since
\revised{nothing prevents $\tvarB$ from occurring in
  $\crep\Constraints\tvarB\mall$ (types can be infinite)}. We
therefore instantiate $\tvarB$ to a new type variable $\mktvar(\tvarA,
\tvarB)$ which will in turn be defined by a new constraint
$\mktvar(\tvarA, \tvarB) \con\meq
\instance\tvarA{\crep\Constraints\tvarB\mall}$.

There are a couple of subtleties concerning the definition of $\instance{}{}$.
The first one is that, strictly speaking, $\instance{}{}$ is a relation rather
than a function because the fresh use variables in \eqref{eq:instance} are not
uniquely determined. In practice, $\instance{}{}$ can be turned into a proper
function by devising a deterministic mechanism that picks fresh use variables
in a way similar to the $\mktvar$ function that we have defined above. The
formal details are tedious but well understood, so we consider the definition
of $\instance{}{}$ above satisfactory as is.
The second subtlety is way more serious and has to do with the
instantiation of type variables (first equation in
\eqref{eq:instance}) which hides a potential approximation due to this
completion phase. To illustrate the issue, suppose that
\begin{equation}
\label{eq:arec}
  \tvar \con\mall \tchan\tint\UseExprU\UseExprV \times \tvar
\end{equation}
is the only constraint concerning $\tvar$ in some constraint set
$\Constraints$ so that we need to provide a $(\meq)$-definition for
$\tvar$. According to \eqref{eq:instance} we have
\[
  \instance\tvar{\tchan\tint\UseExprU\UseExprV \times \tvar}
  =
  \tchan\tint{\uvar_1}{\uvar_2} \times \mktvar(\tvar,\tvar)
\]
so by adding the constraints
\begin{equation}
\label{eq:adef}
  \tvar \con= \mktvar(\tvar,\tvar)
  \text{\qquad and\qquad}
  \mktvar(\tvar,\tvar) \con= \tchan\tint{\uvar_1}{\uvar_2} \times \mktvar(\tvar,\tvar)
\end{equation}
we complete the definition for $\tvar$. There is a fundamental
difference between the constraint \eqref{eq:arec} and those in
\eqref{eq:adef} in that the former admits far more solutions than
those admitted by \eqref{eq:adef}. For example, \revised{the
  constraint \eqref{eq:arec} can be satisfied by a solution that
  contains the assignment $\tvar \mapsto \Type$ where $\Type =
  \tchan\tint10 \times \tchan\tint01 \times \Type$, but the
  constraint~\eqref{eq:adef} cannot.}
The problem of a constraint like \eqref{eq:arec} is that, when we only
have structural information about a type variable, we have no clue
about the uses in its definition, if they follow a pattern, and what
the pattern is. In principle, in order to account for all the
possibilities, we should generate fresh use variables in place of any
use slot in the possibly infinite type.  In practice, however, we want
completion to eventually terminate, and the definition of
$\instance{}{}$ given by \eqref{eq:instance} is one easy way to ensure
this: what we are saying there is that each type variable $\tvarB$
that contributes to the definition of a $(\meq)$-undefined type
variable $\tvarA$ is instantiated only once.  This trivially
guarantees completion termination, for there is only a finite number
of distinct variables to be instantiated. The price we pay with this
definition of $\instance{}{}$ is a potential loss of precision in the
solution of use constraints. We say ``potential'' because we have been
unable to identify a concrete example that exhibits such loss of
precision.
Part of the difficulty of this exercise is due to the fact that the effects of
the approximation on the solution of use constraints may depend on the
particular choice of canonical representatives, which is an implementation
detail of the constraint solver \revised{(see Definition~\ref{def:cr})}.  In part, the effects of the approximation
are limited to peculiar situations:
\begin{enumerate}
\item There is only a fraction of \revised{constraint sets} where the
  same type variable occurring in several different positions must be
  instantiated, namely \revised{those having as solution types with
    infinite branches containing only finitely many channel type
    constructors. The constraint~\eqref{eq:arec} is one such
    example. In all the other cases,} the given definition of
  $\instance{}{}$ does not involve any approximation.

\item A significant fraction of the type variables for which only structural
  information is known are those generated by the rules \refrule{i-fst},
  \refrule{i-snd}, and \refrule{i-weak}. These type variables stand for
  \emph{unlimited} types, namely for types whose uses are either 0 or
  $\omega$. In fact, in most cases \emph{all} the uses in these unlimited
  types are 0. Therefore, the fact that only a handful of fresh use variables
  is created, instead of infinitely many, does not cause any approximation at
  all, since the use variables in these type expressions would all be
  instantiated to 0 anyway.
\end{enumerate}

\noindent We define the completion of a constraint set $\Constraints$ as the least
superset of $\Constraints$ where all the $(\meq)$-undefined type variables in
$\Constraints$ have been properly instantiated:

\begin{definition}[completion]
\label{def:completion}
The \emph{completion} of $\Constraints$, written
$\saturated\Constraints$, is the least set such that:
\begin{enumerate}
\item $\Constraints \subseteq \saturated\Constraints$;

\item $\tvar \in \undefined{\Constraints}\meq$ implies $\tvar \con=
  \mktvar(\tvar,\tvar) \in \saturated\Constraints$;

\item $\mktvar(\tvarA,\tvarB) \in \expr(\saturated\Constraints)$ implies
  $\mktvar(\tvar,\tvarB) \con=
  \instance\tvarA{\crep\Constraints\tvarB\mall} \in \saturated\Constraints$.
\end{enumerate}
\end{definition}

The completion $\saturated\Constraints$ of a finite constraint set
$\Constraints$ can always be computed in finite time as the number of
necessary instantiations is bound by the square of the cardinality of
$\undefined\Constraints\meq$.
Because of the approximation of instances for undefined variables,
$\Constraints$ and $\saturated\Constraints$ are not equivalent in general (see
Example~\ref{ex:approximation} below). However, the introduction of instances
does not affect the satisfiability of the set of constraints.

\begin{proposition}
\label{prop:completion}
The following properties hold:
\begin{enumerate}
\item If $\Constraints$ is satisfiable, then $\saturated\Constraints$
  is satisfiable.

\item If $\Solution$ is a solution for $\saturated\Constraints$, then
  $\Solution$ is also a solution for $\Constraints$.
\end{enumerate}
\end{proposition}
\begin{proof}
  Each $(\mall)$-equivalence class in $\saturated\Constraints$ contains
  exactly one $(\mall)$-equivalence class in $\Constraints$, for each new type
  expression that has been introduced in $\saturated\Constraints$ is
  structurally coherent to an existing type expression in $\Constraints$.
  Then item~(1) is a consequence of
  Proposition~\ref{prop:satisfiable}, while item~(2) follows from the
  fact that $\Constraints \subseteq \saturated\Constraints$.
\end{proof}

\begin{example}
\label{ex:simple_completion}
Considering the constraint set $\Constraints$ in
Example~\ref{ex:simple_verification}, we have three type variables
requiring instantiation, namely $\tvarA$, $\tvarB_2$, and
$\tvarC_1$. According to Definition~\ref{def:completion}, and using
the same canonical representatives mentioned in
Example~\ref{ex:simple_verification}, we augment the constraint set
with the constraints
\[
\begin{array}[t]{@{}r@{~}c@{~}l@{\qquad}r@{~}c@{~}l@{~}c@{~}l@{}}
  \tvarA & \ceq & \mktvar(\tvarA, \tvarA)
  &
  \mktvar(\tvarA, \tvarA) & \ceq & 
  \instance{\tvarA}{\tvarB_1 \times \tvarB_2} & = &
  \mktvar(\tvarA, \tvarB_1) \times
  \mktvar(\tvarA, \tvarB_2)
  \\
  & & &
  \mktvar(\tvarA, \tvarB_1) & \ceq &
  \instance{\tvarA}{ \tchan\tvarD{1+\uvar_1}{2\uvar_2} } & = &
  \tchan\tvarD{\uvar_5}{\uvar_6}
  \\
  & & &
  \mktvar(\tvarA, \tvarB_2) & \ceq &
  \instance{\tvarA}{ \tchan\tint{2\uvar_3}{1+\uvar_4} } & = &
  \tchan\tint{\uvar_7}{\uvar_8}
  \\
  \tvarB_2 & \ceq & \mktvar(\tvarB_2, \tvarB_2)
  &
  \mktvar(\tvarB_2, \tvarB_2) & \ceq &
  \instance{\tvarB_2}{ \tchan\tint{2\uvar_3}{1+\uvar_4} } & = &
  \tchan\tint{\uvar_9}{\uvar_{10}}
  \\
  \tvarC_1 & \ceq & \mktvar(\tvarC_1, \tvarC_1)
  &
  \mktvar(\tvarC_1, \tvarC_1) & \ceq &
  \instance{\tvarC_1}{ \tchan\tvarD{1+\uvar_1}{2\uvar_2} } & = &
  \tchan\tvarD{\uvar_{11}}{\uvar_{12}}
\end{array}
\]
where the $\uvar_i$ with $i\geq 5$ are all fresh.

Observe that the canonical $(\mall)$-representative of $\tvarB_2$ is
instantiated twice, once for defining $\tvarA$ and once for defining
$\tvarB_2$ itself. We will see in Example~\ref{ex:simple_synthesis} that this
double instantiation is key for inferring that $\Second[\varX]$ in
Example~\ref{ex:simple_generation} is used linearly.
\eoe
\end{example}

\begin{example}
\label{ex:approximation}
In this example we show the potential effects of instantiation on the
\revised{ability of the type reconstruction algorithm to identify
  linear channels}. To this aim, consider the following constraint set
\begin{eqnarray*}
  \tvarA & \con\mall & \tchan\tint\UseExprU\UseExprV \times \tvarA
  \\
  \tvarB & \con\meq & \tchan\tint{0}{1+\uvar_1} \times
                      \tchan\tint{0}{2\uvar_2} \times \tvarB
  \\
  \tvarC & \con\meq & \tchan\tint00 \times \tchan\tint00 \times \tvarC
  \\
  \tvarA & \con\meq & \tvarB + \tvarC
\end{eqnarray*}
where, to limit the number of use variables without defeating the
purpose of the example, we write the constant use $0$ in a few use
slots.
Observe that this constraint set admits the solution $\set{ \tvarA
  \mapsto \TypeT, \tvarB \mapsto \TypeT, \tvarC \mapsto \TypeS,
  \uvar_{1,2} \mapsto 0 }$ where $\TypeT$ and $\TypeS$ are the types
that satisfy the equalities $\TypeT = \tchan\tint01 \times
\tchan\tint00 \times \TypeT$ and $\TypeS = \tchan\tint00 \times
\TypeS$.
Yet, if we instantiate $\tvarA$ following the procedure outlined
above we obtain the constraints
\[
  \tvar \con\meq \mktvar(\tvarA, \tvarA)
  \text{\qquad and\qquad}
  \mktvar(\tvarA,\tvarA) \con\meq
  \tchan\tint{\uvar_3}{\uvar_4}
  \times \mktvar(\tvarA, \tvarA)
\]
and now the two constraints below follow by the congruence
rule~\rulename{c-cong *}:
\begin{eqnarray*}
  \tchan\tint{\uvar_3}{\uvar_4}
  & \con\meq & \tchan\tint{0}{1+\uvar_1} + \tchan\tint00
  \\
  \tchan\tint{\uvar_3}{\uvar_4}
  & \con\meq & \tchan\tint{0}{2\uvar_2} + \tchan\tint00
\end{eqnarray*}
This implies that the use variable $\uvar_4$ must simultaneously satisfy the
constraints
\[
\uvar_4 \ceq 1+\uvar_1
\text{\qquad and\qquad}
\uvar_4 \ceq 2\uvar_2
\]
which is only possible if we assign $\uvar_1$ and $\uvar_2$ to a use
other than $0$ and $\uvar_4$ to $\omega$. In other words, after
completion the only feasible solutions for the constraint set above
have the form $\set{ \tvarA \mapsto \TypeT', \tvarB \mapsto \TypeT',
  \tvarC \mapsto \TypeS, \uvar_{1,2} \mapsto \Use, \uvar_3 \mapsto 0,
  \uvar_4 \mapsto \omega }$ for $1 \leq \Use$ where $\TypeT' =
\tchan\tint0\omega \times \TypeT'$, which are less precise than the one
that we could figure out before the instantiation: \revised{$\TypeT$
  denotes an infinite tuple of channels in which those in odd-indexed
  positions are used for performing exactly one output operation;
  $\TypeT'$ denotes an infinite tuple of channels, each being used for
  an unspecified number of output operations.}
\eoe
\end{example}

\subsection{Solution synthesis}
\label{sec:solution_synthesis}
In this phase, \assignments are found for all the use and type
variables that occur in a (completed) constraint set. We have already
seen that it is always possible to consider a trivial use \assignment
that assigns each use variable to $\omega$. In this phase, however, we
have all the information for finding \revised{a use \assignment that,
  albeit not necessarily optimal because of the approximation during
  the completion phase, is minimal according to the $\prel$ precision
  order on uses of Definition~\ref{def:precision}.}

The first step for computing a use \assignment is to collect the whole
set of constraints concerning use expressions. This is done by
repeatedly applying the rules~\refrule{c-use 1} and~\refrule{c-use 2}
shown in Table~\ref{tab:deduction}. Note that the set of derivable use
constraints is finite and can be computed in finite time because
$\Constraints$ is finite. Also, we are sure to derive \emph{all}
possible use constraints if we apply these two rules to a completed
constraint set.

Once use constraints have been determined, \revised{any particular
  \assignment} for use variables can be found by means of an
exhaustive search over all the possible \assignments: the number of
such \assignments is finite because the number of use variables is
finite and so is the domain $\set{0, 1, \omega}$ on which they
range. Clearly this brute force approach is not practical in general
and in Section~\ref{sec:implementation} we will discuss two techniques
that reduce the search space for use \assignments. \revised{The main
  result of this section is independent of the particular use
  \assignment $\SolutionU$ that has been identified.}

\begin{theorem}[correctness of the constraint solving algorithm]
\label{thm:algorithm}
  Let $\rtp\Process\EnvX\Constraints$. If
\begin{enumerate}
\item $\ded\Constraints{\TypeExprT \con\mall \TypeExprS}$ where
  $\TypeExprT$ and $\TypeExprS$ are proper type expressions implies
  that $\TypeExprT$ and $\TypeExprS$ have the same topmost
  constructor, and

\item $\SolutionU$ is a solution of the use constraints of
  $\saturated\Constraints$, and

\item $\SolutionT$ is the solution of the system $\Upsigma \eqdef
  \set{ \tvar = \SolutionU \crep{\saturated\Constraints}\tvar\meq \mid
    \tvar \in \expr(\saturated\Constraints) }$,
\end{enumerate}
then $\SolutionU \cup \SolutionT$ is a solution for $\Constraints$.
\end{theorem}
\begin{proof}
  Let $\Solution \eqdef \SolutionU \cup \SolutionT$.
  We have to prove the implications of Definition~\ref{def:solution}
  for $\saturated\Constraints$.
  We focus on constraints of the form $\TypeExprT \ceq \TypeExprS_1 +
  \TypeExprS_2$, the other constraints being simpler and/or handled in
  a similar way.

  Let ${\rrel} \eqdef \set{ ((\Solution\TypeExprS_1,
    \Solution\TypeExprS_2), \Solution\TypeExprT) \mid
    \ded{\saturated\Constraints}{\TypeExprT \ceq \TypeExprS_1 +
      \TypeExprS_2} }$.  It is enough to show that $\rrel$ satisfies
  the conditions of Definition~\ref{def:tand}, since type combination
  is the \emph{largest} relation that satisfies those same conditions.
  Suppose $((\TypeS_1, \TypeS_2), \TypeT) \in {\rrel}$. Then there
  exist $\TypeExprT$, $\TypeExprS_1$, and $\TypeExprS_2$ such that
  $\ded{\saturated\Constraints}{\TypeExprT \ceq \TypeExprS_1 +
    \TypeExprS_2}$ and $\TypeT = \Solution\TypeExprT$ and $\TypeS_i =
  \Solution\TypeExprS_i$ for $i=1,2$.
  \finalrevision{Without loss of generality, we may also assume that
    $\TypeExprT$, $\TypeExprS_1$, and $\TypeExprS_2$ are proper type
    expressions.
    Indeed, suppose that this is not the case and, for instance,}
  $\TypeExprT = \tvar$. Then, from \refrule{c-subst} we have that
  $\ded{\saturated\Constraints}{\crep{\saturated\Constraints}\tvar\meq
    \ceq \TypeExprS_1 + \TypeExprS_2}$ and, since $\Solution$ is a
  solution of $\Upsigma$, we know that $\Solution(\tvar) =
  \Solution\crep{\saturated\Constraints}\tvar\meq$. Therefore, the
  same pair $((\TypeS_1, \TypeS_2), \TypeT) \in {\rrel}$ can also be
  obtained from the triple
  $(\crep{\saturated\Constraints}\tvar\meq,\TypeExprS_1,\TypeExprS_2)$
  whose first component is proper.
  The same argument applies for $\TypeExprS_1$ and $\TypeExprS_2$.

  Now we reason by cases on the structure of $\TypeExprT$,
  $\TypeExprS_1$, and $\TypeExprS_2$, knowing that all these type
  expressions have the same topmost constructor from hypothesis~(1)
  and \refrule{c-coh 2}:
\begin{itemize}
\item If $\TypeExprT = \TypeExprS_1 = \TypeExprS_2 = \tint$, then
  condition~(1) of Definition~\ref{def:tand} is satisfied.

\item If $\TypeExprT = \tchan{\TypeExprT'}{\UseExprU_1}{\UseExprU_2}$
  and $\TypeExprS_i =
  \tchan{\TypeExprS_i'}{\UseExprV_{2i-1}}{\UseExprV_{2i}}$ for
  $i=1,2$, then from \refrule{c-coh 2} and \refrule{c-cong 1} we
  deduce $\ded{\saturated\Constraints}{\TypeExprT' \con\meq
    \TypeExprS_i'}$ and from \refrule{c-use 2} we deduce
  $\ded{\saturated\Constraints}{\UseExprU_i \con\meq \UseExprV_i +
    \UseExprV_{i+2}}$ for $i=1,2$.
  Since $\Solution$ is a solution for the equality constraints in
  $\saturated\Constraints$, we deduce $\Solution\TypeExprT' =
  \Solution\TypeExprS_1 = \Solution\TypeExprS_2$.
  Since $\Solution$ is a solution for the use constraints in
  $\saturated\Constraints$, we conclude $\Solution\UseExprU_i \meq
  \Solution\UseExprV_i + \Solution\UseExprV_{i+2}$ for $i=1,2$.
  Hence, condition~(2) of Definition~\ref{def:tand} is satisfied.

\item If $\TypeExprT = \TypeExprT_1 \odot \TypeExprT_2$ and
  $\TypeExprS_i = \TypeExprS_{i1} \odot \TypeExprS_{i2}$, then from
  \refrule{c-cong 3} we deduce
  $\ded{\saturated\Constraints}{\TypeExprT_i \con\meq \TypeExprS_{i1}
    + \TypeExprS_{i2}}$ for $i=1,2$.
  We conclude $((\Solution\TypeExprS_{i1}, \Solution\TypeExprS_{i2}),
  \Solution\TypeExprT_i) \in {\rrel}$ by definition of $\rrel$, hence
  condition~(3) of Definition~\ref{def:tand} is satisfied.
  \qedhere
\end{itemize}
\end{proof}

\finalrevision{%
\noindent   Note that the statement of Theorem~\ref{thm:algorithm} embeds the
  constraint solving algorithm, which includes a verification phase
  (item~(1)), a constraint completion phase along with an
  (unspecified, but effective) computation of a solution for the use
  constraints (item~(2)), and the computation of a solution for the
  original constraint set in the form of a finite system of equations
  (item~(3)). The conclusion of the theorem states that the algorithm
  is correct.}

\begin{example}
\label{ex:simple_synthesis}
There are three combination constraints in the set
$\saturated\Constraints$ obtained in
Example~\ref{ex:simple_completion}, namely $\tvarA \ceq \tvarA_1 +
\tvarA_2$, $\tvarB_2 \ceq \tvarB_2 + \tvarB_2$, and $\tvarC_1 \ceq
\tvarC_1 + \tvarC_1$.
By performing suitable substitutions with \refrule{c-subst} we obtain
\[
\begin{prooftree}
  \[
    \[
      \justifies
      \ded{\saturated\Constraints}{
        \tvarA \ceq \tvarA_1 + \tvarA_2
      }
      \using\refrule{c-axiom}
    \]
    \Justifies
    \ded{\saturated\Constraints}{
      \mktvar(\tvarA, \tvarB_1) \times \mktvar(\tvarA, \tvarB_2)
      \ceq
      \tvarB_1 \times \tvarB_2
      +
      \tvarC_1 \times \tvarC_2
    }
    \using\text{\refrule{c-subst} (multiple applications)}
  \]
  \justifies
  \ded{\saturated\Constraints}{
    \mktvar(\tvarA, \tvarB_i)
    \ceq
    \tvarB_i + \tvarC_i
  }
  \using\refrule{c-cong 3}
\end{prooftree}
\]
from which we can further derive
\[
\begin{prooftree}
  \[
    \vdots
    \justifies
    \ded{\saturated\Constraints}{
      \mktvar(\tvarA, \tvarB_1)
      \ceq
      \tvarB_1 + \tvarC_1
    }
  \]
  \Justifies
  \ded{\saturated\Constraints}{
    \tchan\tvarD{\uvar_5}{\uvar_6}
    \ceq
    \tchan\tvarD{1+\uvar_1}{2\uvar_2}
    +
    \tchan\tvarD{\uvar_{11}}{\uvar_{12}}
  }
  \using\text{\refrule{c-subst} (multiple applications)}
\end{prooftree}
\]
as well as
\[
\begin{prooftree}
  \[
    \vdots
    \justifies
    \ded{\saturated\Constraints}{
      \mktvar(\tvarA, \tvarB_2)
      \ceq
      \tvarB_2 + \tvarC_2
    }
  \]
  \Justifies
  \ded{\saturated\Constraints}{
    \tchan\tvarD{\uvar_7}{\uvar_8}
    \ceq
    \tchan\tvarD{\uvar_9}{\uvar_{10}}
    +
    \tchan\tvarD{2\uvar_3}{1+\uvar_4}
  }
  \using\text{\refrule{c-subst} (multiple applications)}
\end{prooftree}
\]
Analogous derivations can be found starting from $\tvarB_2 \ceq
\tvarB_2 + \tvarB_2$ and $\tvarC_1 \ceq \tvarC_1 + \tvarC_1$. At this
point, using \refrule{c-use 2}, we derive the following set of use
constraints:
\[
\begin{array}[t]{rcl}
  \uvar_5 & \ceq & 1+\uvar_1+\uvar_{11} \\
  \uvar_6 & \ceq & 2\uvar_2+\uvar_{12} \\
  \uvar_7 & \ceq & \uvar_9+2\uvar_3 \\
  \uvar_8 & \ceq & \uvar_{10}+1+\uvar_4 \\
\end{array}
\qquad
\begin{array}[t]{rcl}
  \uvar_{11} & \ceq & 2\uvar_{11} \\
  \uvar_{12} & \ceq & 2\uvar_{12} \\
  \uvar_9 & \ceq & 2\uvar_9 \\
  \uvar_{10} & \ceq & 2\uvar_{10} \\
\end{array}
\]
for which we find the most precise solution $\set{
  \uvar_{1..4,6,7,9..12} \mapsto 0, \uvar_{5,8} \mapsto 1}$.

From this set of use constraints we can also appreciate the increased accuracy
deriving from distinguishing the instance $\mktvar(\tvarA,\tvarB_2)$ of the
type variable $\tvarB_2$ used for defining $\tvarA$ from the instance
$\mktvar(\tvarB_2,\tvarB_2)$ of the same type variable $\tvarB_2$ for defining
$\tvarB_2$ itself. Had we chosen to generate a unique instance of $\tvarB_2$,
which is equivalent to saying that $\uvar_8$ and $\uvar_{10}$ are the same use
variable, we would be required to satisfy the use constraint
\[
  \uvar_{10}+1+\uvar_4 \ceq 2\uvar_{10}
\]
which is only possible if we take $\uvar_8 = \uvar_{10} = \omega$. But this
assignment fails to recognize that $\Second[\varX]$ is used linearly in the
process of Example~\ref{ex:simple_generation}.
\eoe
\end{example}

\section{Implementation}
\label{sec:implementation}

In this section we cover a few practical aspects concerning the implementation
of the type reconstruction algorithm. 

\subsection{Derived constraints}
The verification phase of the solver algorithm requires finding all the
constraints of the form $\TypeExprT \con\mall \TypeExprS$ that are derivable
from a given constraint set $\Constraints$. Doing so allows the algorithm to
determine whether $\Constraints$ is satisfiable or not
(Proposition~\ref{prop:satisfiable}). In principle, then, one should compute
the whole set of constraints derivable from $\Constraints$.
The particular nature of the $\mall$ relation enables a more efficient
way of handling this phase. The key observation is that there is no
need to ever perform substitutions (with the rule \refrule{c-subst})
in order to find all the $\con\mall$ constraints. This is because
\refrule{c-coh 2} allows one to relate the type expressions in a
combination, since they must all be structurally coherent and $\mall$
is insensitive to the actual content of the use slots in channel
types.  This means that all $\con\mall$ constraints can be computed
efficiently using conventional unification techniques (ignoring the
content of use slots).
In fact, the implementation uses unification also for the constraints of the
form $\TypeExprT \ceq \TypeExprS$. Once all the $\ceq$ constraints have been
found and the constraint set has been completed, substitutions in constraints
expressing combinations can be performed efficiently by mapping each type
variable to its canonical representative.

\subsection{Use constraints resolution}
In Section~\ref{sec:solver} we have refrained from providing any
detail about how use constraints are solved and argued that a
particular use \assignment can always be found given that both the set
of constraints and the domain of use variables are finite. While this
argument suffices for establishing the decidability of this crucial
phase of the reconstruction algorithm, a na\"ive solver based on an
exhaustive search of all the use \assignments would be unusable, since
the number of use variables is typically large, even in small
processes. Incidentally, note that completion contributes
significantly to this number, since it generates fresh use variables
for all the instantiated channel types.

There are two simple yet effective strategies that can be used for
speeding up the search of a particular use \assignment{}
\revised{(both have been implemented in the prototype)}. The first
strategy is based on the observation that, although the set of use
variables can be large, it can often be partitioned into many
independent subsets.
Finding partitions is easy: two variables $\uvar_1$ and $\uvar_2$ are
\emphdef{related} in $\Constraints$ if $\ded\Constraints{\UseExprU
  \ceq \UseExprV}$ and $\uvar_1$, $\uvar_2$ occur in $\UseExprU \ceq
\UseExprV$ (regardless of where $\uvar_1$ and $\uvar_2$ occur
exactly). The dependencies between variables induce a partitioning of
the use constraints such that the use variables occurring in the
constraints of a partition are all related among them, and are not
related with any other use variable occurring in a use constraint
outside the partition. Once the partitioning of use constraints has
been determined, each partition can be solved independently of the
others.

The second strategy is based on the observation that many use
constraints have the form $\uvar \ceq \UseExpr$ where $\uvar$ does not
occur in $\UseExpr$. In this case, the value of $\uvar$ is in fact
determined by $\UseExpr$. So, $\UseExpr$ can be substituted in place
of all the occurrences of $\uvar$ in a given set of use constraints
and, once a \assignment is found for the use variables in the set of
use constraints with the substitution, the \assignment for $\uvar$ can
be determined by simply evaluating $\UseExpr$ under such \assignment.

\Luca{Il metodo iterativo di Igarashi e Kobayashi \`e rapido, ma funzionerebbe
  anche nel nostro contesto?}

\subsection{Pair splitting versus pair projection}
\label{sec:splitting}
It is usually the case that linearly typed languages provide a
dedicated construct for splitting pairs \revised{(a notable exception
  is~\cite{Kobayashi06})}. The language introduced
in~\cite[Chapter~1]{Pierce04}, for example, has an expression form
\[
  \mkkeyword{split}~\ExpressionE~\mkkeyword{as}~\PairX\varX\varY~\mkkeyword{in}~\ExpressionF
\]
that evaluates $\ExpressionE$ to a pair, binds the first and second
component of the pair respectively to the variables $\varX$ and
$\varY$, and then evaluates $\ExpressionF$. At the same time, no pair
projection primitives are usually provided. This is because in most
linear type systems linear values ``contaminate'' with linearity the
composite data structures in which they occur: for example, a pair
containing linear values is itself a linear value and can only be used
once, whereas for extracting \emph{both} components of a pair using
the projections one would have to project the pair \emph{twice}, once
using $\First$ and one more time using $\Second$. For this reason, the
$\mkkeyword{split}$ construct becomes the only way to use linear pairs
without violating linearity, as it grants access to both components of
a pair but accessing the pair only once.

The process language we used in an early version of this
article~\cite{Padovani14A} provided a $\mkkeyword{split}$ construct
for splitting pairs and did not have the projections $\First$ and
$\Second$. In fact, the ability to use $\First$ and $\Second$ without
violating linearity constraints in our type system was pointed out by a
reviewer of~\cite{Padovani14A} and in this article we have decided to promote
projections as the sole mechanism for accessing pair components.
Notwithstanding this, there is a practical point in favor of
$\mkkeyword{split}$ when considering an actual implementation of the type
system.
Indeed, the pair projection rules~\refrule{i-fst} and~\refrule{i-snd} are
among the few that generate constraints of the form $\unlimited\tvar$ for
some type variable $\tvar$. In the case of \refrule{i-fst} and
\refrule{i-snd}, the unlimited type variable stands for the component of the
pair that is discarded by the projection. For instance, we can derive
\[
\begin{prooftree}
  \[
    \rte{
      \var
    }{
      \tvarB_1
    }{
      \var : \tvarB_1
    }{
      \emptyset
    }
    \justifies
    \rte{
      \First[\var]
    }{
      \tvarA_1
    }{
      \var : \tvarB_1
    }{
      \set{ \tvarB_1 \ceq \tvarA_1 \times \tvarC_1, \unlimited{\tvarC_1} }
    }
  \]
  \quad
  \[
    \rte{
      \var
    }{
      \tvarB_2
    }{
      \var : \tvarB_2
    }{
      \emptyset
    }
    \justifies
    \rte{
      \Second[\var]
    }{
      \tvarA_2
    }{
      \var : \tvarB_2
    }{
      \set{ \tvarB_2 \ceq \tvarC_2 \times \tvarA_2, \unlimited{\tvarC_2} }
    }
  \]
  \justifies
  \rte{
    \Pair{\First[\var]}{\Second[\var]}
  }{
    \tvarA_1 \times \tvarA_2
  }{
    \var : \tvar
  }{
    \set{
      \tvar \ceq \tvarB_1 + \tvarB_2,
      \tvarB_1 \ceq \tvarA_1 \times \tvarC_1,
      \tvarB_2 \ceq \tvarC_2 \times \tvarA_2,
      \unlimited{\tvarC_1},
      \unlimited{\tvarC_2}
    }
  }
\end{prooftree}
\]
and we observe that $\tvarC_1$ and $\tvarC_2$ are examples of those type
variables for which only structural information is known, but no definition is
present in the constraint set. Compare this with a hypothetical derivation
concerning a splitting construct (for expressions)
\[
\begin{prooftree}
  \rte{
    \var
  }{
    \tvar
  }{
    \var : \tvar
  }{
    \emptyset
  }
  \quad
  \[
    \rte{\var_1}{\tvar_1}{\var_1 : \tvar_1}{\emptyset}
    \qquad
    \rte{\var_2}{\tvar_2}{\var_2 : \tvar_2}{\emptyset}
    \justifies
    \rte{
      \Pair{\var_1}{\var_2}
    }{
      \tvar_1 \times \tvar_2
    }{
      \var_1 : \tvar_1,
      \var_2 : \tvar_2
    }{
      \emptyset
    }
  \]
  \justifies
  \rte{
    \Split\var{\var_1}{\var_2}{\Pair{\var_1}{\var_2}}
  }{
    \tvar_1 \times \tvar_2
  }{
    \var : \tvar
  }{
    \set{ \tvar \ceq \tvar_1 \times \tvar_2 }
  }
\end{prooftree}
\]
producing a much smaller constraint set which, in addition, is free from
$\unlimited\cdot$ constraints and includes a definition for $\tvar$. The
constraint set obtained from the second derivation is somewhat easier to solve,
if only because it requires no completion, meaning fewer use variables to
generate and fewer chances of stumbling on the approximated solution of use
constraints (Example~\ref{ex:approximation}).

Incidentally we observe, somehow surprisingly, that the two constraint sets
are not exactly equivalent. In particular, the constraint set obtained from
the first derivation admits a solution containing the \assignments
\[
  \set{
    \tvar \mapsto \tchan\tint\omega0 \times \tchan\tint0\omega,
    \tvarA_1 \mapsto \tchan\tint10,
    \tvarA_2 \mapsto \tchan\tint01
  }
\]
whereas in the second derivation, if we fix $\tvar$ as in the \assignment
above, we can only have
\[
  \set{
    \tvar \mapsto \tchan\tint\omega0 \times \tchan\tint0\omega,
    \tvarA_1 \mapsto \tchan\tint\omega0,
    \tvarA_2 \mapsto \tchan\tint0\omega
  }
\]
meaning that, using projections, it is possible to \emph{extract} from a pair
only the needed capabilities, provided that what remains unused has an
unlimited type. On the contrary, $\mkkeyword{split}$ always extracts the full
set of capabilities from each component of the pair.

In conclusion, in spite of the features of the type system we argue that it is
a good idea to provide both pair projections and pair splitting, and that pair
splitting should be preferred whenever convenient to use.

\section{Examples}
\label{sec:examples}

In this section we discuss three more elaborate examples that highlight the
features of our type reconstruction algorithm. For better clarity, in these
examples we extend the language with triples, boolean values, conditional
branching, arithmetic and relational operators, OCaml-like polymorphic
variants~\cite[Chapter 4]{OCAML}, and a more general form of pattern matching.
All these extensions can be easily accommodated or encoded in the language
presented in Section~\ref{sec:language} and are supported by the prototype
implementation of the reconstruction algorithm.

\begin{example}
\label{ex:trees}
\newcommand{\Leaf}{\mathtt{Leaf}}
\newcommand{\Node}{\mathtt{Node}}
\newcommand{\takec}{\mathtt{take}}
\newcommand{\skipc}{\mathtt{skip}}
\newcommand{\TypeTree}{\Type_{\mathit{tree}}}
The purpose of this example is to show the reconstruction algorithm at work on
a fairly complex traversal of a binary tree. The traversal is realized by the
two processes $\takec$ and $\skipc$ below
\[
\begin{array}{cl}
  &
  \bang\receive\takec\var
  \CaseV\var{
    \quad\Leaf & \idle
    \\
    \quad\Node\ttparens{\ChannelC\ttcomma\varY\ttcomma\varZ} &
    \send\ChannelC3
    \parop
    \send\takec\varY
    \parop
    \send\skipc\varZ
  }
  \\
  \parop &
  \bang\receive\skipc\var
  \CaseV\varX{
    \quad\Leaf & \idle
    \\
    \quad\Node\ttparens{\ttunderscore\ttcomma\varY\ttcomma\varZ}
    &
    \send\skipc\varY
    \parop
    \send\takec\varZ
  }
\end{array}
\]
where, as customary, we identify the name of a process with the replicated
channel on which the process waits for invocations.

\tikzstyle{every picture}+=[
  remember picture,
  thick,
  rounded corners=1pt,
  transition/.style={->,semithick,sloped,auto},
  branch/.style={auto},
  trace/.style={
    line width=9pt,
    line cap=round,
    line join=round,
    red!60
  },
  inclusion/.style={thick,dotted},
]

\begin{figure}
\begin{tikzpicture}[xscale=0.3,yscale=0.3]
  \draw[trace] (15, 8) -- (7, 6) -- (3, 4) -- (1, 2) -- (0, 0);
  \draw[trace] (6, 0) -- (6, 0);
  \draw[trace] (10, 0) -- (10, 0);
  \draw[trace] (13, 2) -- (12, 0);
  \draw[trace] (18, 0) -- (18, 0);
  \draw[trace] (21, 2) -- (20, 0);
  \draw[trace] (27, 4) -- (25, 2) -- (24, 0);
  \draw[trace] (30, 0) -- (30, 0);

  \node at (15, 8) {$\bullet$};

  \foreach \x in {0, 2, ..., 30} {
    \node at (\x, 0) {$\bullet$};
  }

  \foreach \x in {1, 5, ..., 30} {
    \node at (\x, 2) {$\bullet$};

    \foreach \y in {-1, 1} {
      \draw (\x, 2) -- ($(\x + \y, 0)$);
    }
  }

  \foreach \x in {3, 11, ..., 30} {
    \node at (\x, 4) {$\bullet$};

    \foreach \y in {-2, 2} {
      \draw (\x, 4) -- ($(\x + \y, 2)$);
    }

  }

  \foreach \x in {7, 23} {
    \node at (\x, 6) {$\bullet$};

    \draw (15, 8) -- (\x, 6);

    \foreach \y in {-4, 4} {
      \draw (\x, 6) -- ($(\x + \y, 4)$);
    }
  }
\end{tikzpicture}
\caption{\label{fig:take} Regions of a complete binary tree used by
  $\takec$.}
\end{figure}
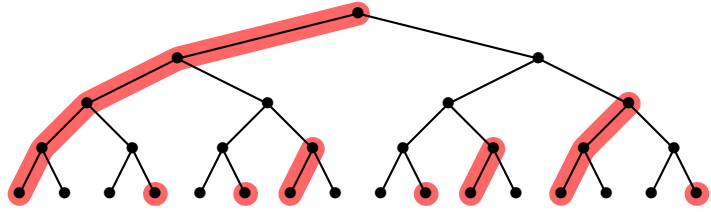

Both $\takec$ and $\skipc$ receive as argument a binary tree $\varX$
and analyze its structure by means of pattern matching. If the tree is
empty, no further operation is performed. When $\takec$ receives a
non-empty tree, it \revised{uses the channel $\ChannelC$} found at the
root of the tree, it recursively visits the left branch $\varY$ and
passes the right branch $\varZ$ to $\skipc$. The process $\skipc$ does
not use the channel found at the root of the tree, but visits the left
branch recursively and passes the right branch to $\takec$.

The types inferred for $\takec$ and $\skipc$ are
\[
  \takec : \tchan\TypeT\omega\omega
  \text{\qquad and\qquad}
  \skipc : \tchan\TypeS\omega\omega
\]
where $\TypeT$ and $\TypeS$ are the types that satisfy the
\revised{equalities}
\begin{eqnarray*}
  \TypeT & = & \Leaf \oplus \Node\ttparens{ \tchan\tint01 \times
    \TypeT \times \TypeS } \\
  \TypeS & = & \Leaf \oplus \Node\ttparens{ \tchan\tint00 \times \TypeS
    \times \TypeT }
\end{eqnarray*}

In words, $\takec$ uses every channel that is found after an even
number of right traversals, whereas $\skipc$ uses every channel that
is found after an odd number of right traversals.
Figure~\ref{fig:take} depicts the regions of a (complete) binary tree
of depth 4 that are used by $\takec$, while the unmarked regions are
those used by $\skipc$. Overall, the invocation
\[
  \send\takec{\mathit{tree}}
  \parop
  \send\skipc{\mathit{tree}}
\]
allows the reconstruction algorithm to infer that \emph{all} the
channels in $\mathit{tree}$ are used, namely that $\mathit{tree}$ has
type $\TypeTree = \Leaf \oplus \Node\ttparens{ \tchan\tint01 \times
  \TypeTree \times \TypeTree} = \TypeT + \TypeS$.
\eoe
\end{example}

\begin{example}
\label{ex:math_server}
\newcommand{\Quit}{\mathtt{Quit}}
\newcommand{\Plus}{\mathtt{Plus}}
\newcommand{\Eq}{\mathtt{Eq}}
\newcommand{\Neg}{\mathtt{Neg}}
\newcommand{\serverc}{\mathtt{server}}
\newcommand{\auxc}{\mathtt{aux}}
\newcommand{\ac}{\mathtt{foo}}
\newcommand{\bc}{\mathtt{bar}}
In this example we show how our type reconstruction algorithm can be
used for inferring \emph{session types}.
\revised{%
  Some familiarity with the related literature and particularly
  with~\cite{DardhaGiachinoSangiorgi12,Dardha14} is assumed.
  Session types~\cite{Honda93,HondaVasconcelosKubo98,GayHole05} are
  protocol specifications describing the sequence of input/output
  operations that are meant to be performed on a (private)
  communication channel. In most presentations, session types $T$,
  $\dots$ include constructs like $\stin\Type\SessionType$ (input a
  message of type $\Type$, then use the channel according to
  $\SessionType$) or $\stout\Type\SessionType$ (output a message of
  type $\Type$, then use the channel according to $\SessionType$) and
  possibly others for describing terminated protocols and protocols
  with branching structure. By considering also recursive session
  types (as done, \eg, in~\cite{GayHole05}), or by taking the regular
  trees over such constructors (as we have done for our type language
  in this paper), it is possible to describe potentially infinite
  protocols. For instance, the infinite regular tree $\SessionType$
  satisfying the equality
\[
  \SessionTypeT = \stout\tint{\stin\tbool\SessionTypeT}
\]
describes the protocol followed by a process that alternates outputs
of integers and inputs of booleans on a session channel, whereas the
infinite regular tree $\SessionTypeS$ satisfying the equality
\[
  \SessionTypeS = \stin\tint{\stout\tbool\SessionTypeS}
\]
describes the protocol followed by a process that alternates inputs of
integers and outputs of booleans. According to the conventional
terminology, $\SessionTypeT$ and $\SessionTypeS$ above are \emph{dual}
of each other: each action described in $\SessionTypeT$ (like the
output of a message of type $\tint$) is matched by a corresponding
co-action in $\SessionTypeS$ (like the input of a message of type
$\tint$). This implies that two processes that respectively follow the
protocols $\SessionTypeT$ and $\SessionTypeS$ when using the
\emph{same} session channel can interact without errors: when one
process sends a message of type $\Type$ on the channel, the other
process is ready to receive a message of the same type from the
channel.
Two such processes are those yielded by the outputs
$\send\ac\ChannelC$ and $\send\bc\ChannelC$ below:
\[
\begin{lines}
  \phantom{{}\parop{}}
  \bang
  \receive\ac\varX
  \send\varX{\mathtt{random}}\ttdot
  \receive\varX{\ttunderscore}
  \send\ac\varX
  \\
  {}\parop
  \bang
  \receive\bc\varY
  \receive\varY{n}
  \send\varY{\ttparens{n \mathbin{\mkkeyword{mod}} 2}}\ttdot
  \send\bc\varY
  \\
  {}\parop
  \new\ChannelC\ttparens{
    \send\ac\ChannelC
    \parop
    \send\bc\ChannelC
  }
\end{lines}
\]
It is easy to trace a correspondence of the actions described by
$\SessionTypeT$ with the operations performed on $\varX$, and of the
actions described by $\SessionTypeS$ with the operations performed on
$\varY$. Given that $\varX$ and $\varY$ are instantiated with the same
channel $\ChannelC$, and given the duality that relates
$\SessionTypeT$ and $\SessionTypeS$, this process exhibits no
communication errors even if the \emph{same} channel $\ChannelC$ is
exchanging messages of \emph{different} types ($\tint$ or $\tbool$).
For this reason, $\ChannelC$ is not a linear channel and the above
process is ill typed \finalrevision{according to our typing
  discipline}. However,} as discussed in~\cite{Kobayashi02b,
DardhaGiachinoSangiorgi12,Dardha14}, binary sessions and binary
session types can be \emph{encoded} in the linear $\pi$-calculus using
a continuation passing style. The key idea of the encoding is that
each communication in a session is performed on a distinct linear
channel, and the exchanged message carries, along with the actual
payload, a continuation channel on which the rest of the conversation
takes place.
\revised{According to this intuition, the process above is encoded in
  the linear $\pi$-calculus as the term:
\[
\begin{lines}
  \phantom{{}\parop{}}
  \bang
  \receive\ac\varX
  \new\ChannelA
  \ttparens{
    \send\varX{\Pair{\mathtt{random}}\ChannelA}
    \parop
    \receive\ChannelA{\ttunderscore\ttcomma\varX'}
    \send\ac{\varX'}
  }
  \\
  {}\parop
  \bang
  \receive\bc\varY
  \receive\varY{n\ttcomma\varY'}
  \new\ChannelB
  \ttparens{
    \send{\varY'}{\Pair{n \mathbin{\mkkeyword{mod}} 2}\ChannelB}
    \parop
    \send\bc\ChannelB
  }
  \\
  {}\parop
  \new\ChannelC\ttparens{
    \send\ac\ChannelC
    \parop
    \send\bc\ChannelC
  }
\end{lines}
\]
where $\ChannelA$, $\ChannelB$, and $\ChannelC$ are all linear
channels (with possibly different types) used for exactly one
communication.
The encoding of processes using (binary) sessions into the linear
$\pi$-calculus induces a corresponding encoding of session types into
linear channel types. In particular, input and output session types
are encoded according to the laws
\begin{equation}
\label{eq:enc}
\begin{array}{rcl}
  \enc{\stin\Type\SessionTypeT} & = & \tchan{ \Type \times \enc\SessionTypeT }10
  \\
  \enc{\stout\Type\SessionTypeT} & = & \tchan{ \Type \times \enc{\co\SessionTypeT} }01
\end{array}
\end{equation}
where we use $\co\SessionTypeT$ to denote the dual protocol of
$\SessionTypeT$. Such encoding is nothing but the coinductive
extension of the one described in~\cite{DardhaGiachinoSangiorgi12} to
infinite protocols.
Note that in $\enc{\stout\Type\SessionTypeT}$, the type of the
continuation channel is the encoding of the \emph{dual} of
$\SessionTypeT$. This is because the transmitted continuation will be
used \emph{by the receiver process} in a complementary fashion with
respect to $\SessionTypeT$, which instead describes the continuation
of the protocol from the viewpoint of the sender.
As an example, the protocols $\SessionTypeT$ and $\SessionTypeS$ above
can be respectively encoded as the types $\TypeT$ and $\TypeS$ that
satisfy the equalities
\[
  \TypeT =
  \tchan{\tint \times \tchan{\tbool \times \TypeS}01}01
  \qquad
  \TypeS =
  \tchan{\tint \times \tchan{\tbool \times \TypeS}01}10
\]
It turns out that these are the types that our type reconstruction
algorithm associates with $\varX$ and $\varY$. This is not a
coincidence, for essentially three reasons: (1) the encoding of a
well-typed process making use of binary sessions is always a
well-typed process in the linear
$\pi$-calculus~\cite{DardhaGiachinoSangiorgi12,Dardha14}, (2) our type
reconstruction algorithm is complete
(Theorem~\ref{thm:completeness_processes}), and (3) it can identify a
channel as linear when it is used for one communication only
(Section~\ref{sec:solution_synthesis}).
The upshot is that, once the types $\TypeT$ and $\TypeS$ have been
reconstructed, the protocols $\SessionTypeT$ and $\SessionTypeS$ can
be obtained by a straightforward procedure that ``decodes'' $\TypeT$
and $\TypeS$ using the inverse of the transformation sketched by the
equations~\eqref{eq:enc}.
There is a technical advantage of such rather indirect way of
performing session type reconstruction.  Duality accounts for a good
share of the complexity of algorithms that reconstruct session types
directly~\cite{Mezzina08}. However, as the authors
of~\cite{DardhaGiachinoSangiorgi12} point out, the notion of duality
that relates $\SessionTypeT$ and $\SessionTypeS$ -- and that
\emph{globally} affects their structure -- boils down to a
\emph{local} swapping of uses in the topmost channel types in $\TypeT$
and $\TypeS$. This is a general property of the encoding that has
important practical implications: the hard work is carried over during
type reconstruction for the linear $\pi$-calculus, where there is no
duality to worry about; once such phase is completed, session types
can be obtained from linear channel types with little effort.

We have equipped the prototype implementation of the type
reconstruction algorithm with a flag that decodes linear channel types
into session types (the decoding procedure accounts for a handful of
lines of code).  In this way, the tool can be used for inferring the
communication protocol of processes encoded in the linear
$\pi$-calculus. Since the type reconstruction algorithm supports
infinite and disjoint sum types, both infinite protocols and protocols
with branches can be inferred. Examples of such processes, like for
instance the server for mathematical operations described
in~\cite{GayHole05}, are illustrated on the home page of the tool and
in its source archive.}
\eoe

\end{example}

\begin{example}
\label{ex:filter}
\newcommand{\filterc}{\mathtt{filter}}
In this example we motivate the requirement expressed in the rules
\refrule{t-new} and \refrule{i-new} imposing that the type of
restricted channels should have the same use in its input/output use
slots. To this aim, consider the process below
\[
  \bang\receive\filterc{\PairX\ChannelA\ChannelB}
  \receive\ChannelA{\PairX{n}\ChannelC}
  \If{
    n \geq 0
    \begin{lines}
  }{
    \new\ChannelD\ttparens{
      \send\ChannelB{\Pair{n}{\ChannelD_1}}
      \parop
      \send\filterc{\Pair\ChannelC{\ChannelD_2}}
    }
    \\
  }{
    \send\filterc{\Pair\ChannelC\ChannelB}
    \end{lines}
  }
\]
which filters numbers received from channel $\ChannelA$ and forwards
the non-negative ones on channel $\ChannelB$. Each number $n$ comes
along with a continuation channel $\ChannelC$ from which the next
number in the stream will be received. Symmetrically, any message sent
on $\ChannelB$ includes a continuation $\ChannelD$ on which the next
non-negative number will be sent. For convenience, we distinguish
$\ChannelD$ bound by $\mkkeyword{new}$ from the two rightmost
occurrences $\ChannelD_1$ and $\ChannelD_2$ of $\ChannelD$.

For this process the reconstruction algorithm infers the type
\begin{equation}
\label{eq:filter_right}
  \filterc :
  \tchan{
    \Type \times \tchan{\tint \times \Type}01
  }\omega\omega
\end{equation}
where $\Type$ is the type that satisfies the equality $\Type =
\tchan{\tint \times \Type}10$ meaning that $\ChannelD_1$ and
$\ChannelD_2$ are respectively assigned the types $\Type$ and
$\tchan{\tint \times \Type}01$ and overall $\ChannelD$ has type $\Type
+ \tchan{\tint \times \Type}01 = \tchan{\tint \times \Type}10 +
\tchan{\tint \times \Type}01 = \tchan{\tint \times \Type}11$.
The reason why $\ChannelD_2$ has type $\tchan{\tint \times \Type}01$, namely
that $\ChannelD_2$ is used for an output operation, is clear, since
$\ChannelD_2$ must have the same type as $\ChannelB$ and $\ChannelB$ is indeed
used for an output operation in the body of $\filterc$. However, in the whole
process there is no explicit evidence that $\ChannelD_1$ will be used for an
input operation, and the input use $1$ in its type $\Type = \tchan{\tint
  \times \Type}10$ is deduced ``by subtraction'', as we have discussed in the
informal overview at the beginning of Section~\ref{sec:solver}.

If we do \emph{not} impose the constraint that restricted (linear)
channel should have the same input/output use, we can find \revised{a
  more precise solution that determines for $\filterc$ the type}
\begin{equation}
\label{eq:filter_wrong}
\filterc :
  \tchan{
    \Type \times \tchan{\tint \times \TypeS}01
  }\omega\omega
\end{equation}
where $\TypeS$ is the type that satisfies the equality $\TypeS =
\tchan{\tint \times \TypeS}00$.
According to \eqref{eq:filter_wrong}, $\ChannelD_1$ is assigned the
type $\TypeS$ saying that no operation will ever be performed on it.
This phenomenon is a consequence of the fact that, when we apply the
type reconstruction algorithm on an isolated process, like $\filterc$
above, which is never invoked, the reconstruction algorithm has only a
partial view of the behavior of the process on the channel it
creates. For extruded channels like $\ChannelD$, in particular, the
algorithm is unable to infer any direct use.
We argue that the typing \eqref{eq:filter_wrong} renders $\filterc$ a
useless process from which it is not possible to receive any message,
unless $\filterc$ is typed along with the rest of the program that
invokes it. But this latter strategy prevents \emph{de facto} the
modular application of the reconstruction algorithm to the separate
constituents of a program.

The typing \eqref{eq:filter_right} is made possible by the completion
phase (Section~\ref{sec:solver}), which is an original feature of our
type reconstruction algorithm. The prototype implementation of the
algorithm provides a flag that disables the constraint on equal uses
in \refrule{i-new} allowing experimentation of the behavior of the
algorithm on examples like this one.
\eoe
\end{example}

\section{Concluding Remarks}
\label{sec:conclusion}

Previous works on the linear $\pi$-calculus either \revised{do not
  treat} composite
types~\cite{KobayashiPierceTurner99,IgarashiKobayashi00} or are based
on an interpretation of linearity that limits data sharing and
parallelism~\cite{Igarashi97,IgarashiKobayashi97}. \finalrevision{Type
  reconstruction for} recursive \revised{or, somewhat equivalently,
  infinite} types has also been neglected, despite the key role played
by these types for describing structured data (lists, trees, etc.) and
structured interactions~\cite{Dardha14}.
In this work we have extended the linear $\pi$-calculus with both
composite and \revised{infinite} types and have adopted a more relaxed
attitude towards linearity that fosters data sharing and parallelism
while maintaining the availability of a type reconstruction algorithm.
The extension is a very natural one, as witnessed by the fact that our
type system uses essentially the same rules of previous works, the
main novelty being a different type combination operator.  This small
change has nonetheless non-trivial consequences on the reconstruction
algorithm, which must reconcile the propagation of constraints across
composite types \revised{and} the impossibility to rely on plain type
unification: different occurrences of the same identifier may be
assigned different types and types \revised{may be infinite}.
Our extension also gives renewed relevance to types like
$\tchan\Type{0}{0}$. In previous works these types were admitted but
essentially useless: channels with such types could only be passed
around in messages without actually ever being used. That is, they
could be erased without affecting processes. In our type system, it is
the existence of these types that enables the sharing of structured
data (see the decomposition of $\TypeList$ into $\TypeEven$ and
$\TypeOdd$ in Section~\ref{sec:introduction}).

Binary sessions~\cite{Honda93,HondaVasconcelosKubo98} can be encoded
into the linear
$\pi$-calculus~\cite{Kobayashi02b,DardhaGiachinoSangiorgi12}. Thus, we
indirectly provide a complete reconstruction algorithm for
\revised{possibly infinite}, \finalrevision{higher-order,} binary
session types. As shown in~\cite{Mezzina08}, direct session type
reconstruction poses two major technical challenges: on the one hand,
the necessity to deal with dual types; on the other hand, the fact
that subtyping must be taken into account for that is the only way to
properly handle selections in conditionals. Interestingly, both
complications disappear when session types are encoded in the linear
$\pi$-calculus: duality simply turns into swapping the input/output
use annotations in channel types~\cite{DardhaGiachinoSangiorgi12},
whereas selections become outputs of variant data types which can be
dealt with using conventional techniques based on
unification~\cite{OCAML}.

To assess the feasibility of the approach, we have implemented the
type reconstruction algorithm in a tool for the static analysis of
$\pi$-calculus processes. Given that even simple processes generate
large constraint sets, the prototype has been invaluable for testing
the algorithm at work on non-trivial examples. The reconstruction
described in this article is only the first step for more advanced
forms of analysis, such as those for reasoning on deadlocks and
locks~\cite{Padovani14B}. We have extended the tool in such a way that
subsequent analyses can be plugged on top of the reconstruction
algorithm for linear channels~\cite{PadovaniChenTosatto15}.

Structural subtyping and polymorphism are two natural developments of
our work. The former has already been considered
in~\cite{IgarashiKobayashi97}, but it is necessary to understand how
it integrates with our notion of type combination and how it affects
constraint generation and resolution. Polymorphism makes sense for
unlimited channels only (there is little point in having polymorphic
linear channels, since they can only be used once
anyway). Nevertheless, support for polymorphism is not entirely
trivial, since some type variables may need to be restricted to
unlimited types. For example, the channel $\mathtt{first}$ in the
process
$\bang\receive{\mathtt{first}}{\PairX\varX\varY}\send\varY{\First[\varX]}$
would have type
$\forall\tvarA.\forall\tvarB.\unlimited\tvarB \Rightarrow
\tchan{(\tvarA \times \tvarB) \times \tchan\tvarA01}\omega0$.

\subsection*{Acknowledgements.}
The author is grateful to the anonymous reviewers whose numerous
questions, detailed comments and suggestions have significantly
contributed to improving both content and presentation of this
article.
The author is also grateful to Naoki Kobayashi for his comments on an
earlier version of the article.

\bibliographystyle{abbrv}
\bibliography{main}

\appendix

\section{Supplement to Section~\ref{sec:types}}
\label{sec:extra_types}

To prove Theorem~\ref{thm:sr} we need a series of standard auxiliary
results, including weakening (Lemma~\ref{lem:weak}) and substitution
(Lemma~\ref{lem:subst}) for both expressions and processes.

\begin{lemma}[weakening]
\label{lem:weak}
The following properties hold:
\begin{enumerate}
\item If $\wte\Env\Expression\Type$ and $\unlimited{\Env'}$ and $\Env
  + \Env'$ is defined, then $\wte{\Env + \Env'}{\Expression}{\Type}$.

\item If $\wtp\Env\Process$ and $\unlimited{\Env'}$ and $\Env \tand
  \Env'$ is defined, then $\wtp{\Env \tand \Env'}\Process$.
\end{enumerate}
\end{lemma}
\begin{proof}
  Both items are proved by a standard induction on the typing
  derivation. In case~(2) we assume, without loss of generality, that
  $\bn(\Process) \cap \dom(\Env) = \EmptyEnv$ (recall that we identify
  processes modulo renaming of bound names).
\end{proof}

\begin{lemma}[substitution]
\label{lem:subst}
Let $\wte{\Env_1}{\Value}{\Type}$. The following properties hold:
\begin{enumerate}
\item If $\wte{\Env_2, \var : \Type}{\Expression}{\TypeS}$ and $\Env_1
  + \Env_2$ is defined, then $\wte{\Env_1 +
    \Env_2}{\Expression\subst\Value\var}{\TypeS}$.

\item If $\wtp{\Env_2, \var : \Type}{\Process}$ and $\Env_1 \tand
  \Env_2$ is defined, then $\wtp{\Env_1 \tand
    \Env_2}{\Process\subst\Value\var}$.
\end{enumerate}
\end{lemma}
\begin{proof}
  The proofs are standard, except for the following property of the type
  system: $\unlimited\Type$ implies $\unlimited{\Env_1}$, which can be easily
  proved by induction on the derivation of $\wte{\Env_1}{\Value}{\Type}$.
\end{proof}

Next is type preservation under structural pre-congruence.

\begin{lemma}
\label{lem:struct}
If $\wtp\Env\ProcessP$ and $\ProcessP \sle \ProcessQ$, then
$\wtp\Env\ProcessQ$.
\end{lemma}
\begin{proof}
  We only show the case in which a replicated process is expanded.
  Assume $\ProcessP = \bang\ProcessP' \sle \bang\ProcessP' \parop
  \ProcessP' = \ProcessQ$.
  From the hypothesis $\wtp\Env\ProcessP$ and \refrule{t-rep} we
  deduce $\wtp\Env{\ProcessP'}$ and $\unlimited\Env$.
  By definition of unlimited environment (see Definition~\ref{def:un})
  we have $\Env = \Env + \Env$.
  We conclude $\wtp\Env\ProcessQ$ with an application of
  \refrule{t-par}.
\end{proof}

\begin{lemma}
\label{lem:env}
If $\Env \lred[\Label] \Env'$ and $\Env \tand \Env''$ is defined, then
$\Env \tand \Env'' \lred[\Label] \Env' \tand \Env''$.
\end{lemma}
\begin{proof}
  Easy consequences of the definition of $\lred[\Label]$ on type
  environments.
\end{proof}

\revised{We conclude with type preservation for expressions and
  subject reduction for processes.}

\begin{lemma}
\label{lem:sr_expr}
Let $\wte\Env\Expression\Type$ and $\Expression \eval \Value$.  Then
$\wte\Env\Value\Type$.
\end{lemma}
\begin{proof}
  By induction on $\Expression \eval \Value$ using the hypothesis that
  $\Expression$ is well typed.
\end{proof}

\begin{reptheorem}{thm:sr}
  Let $\wtp\Env\ProcessP$ and $\ProcessP \lred[\Label]
  \ProcessQ$. Then $\wtp{\Env'}\ProcessQ$ for some $\Env'$ such that
  $\Env \lred[\Label] \Env'$.
\end{reptheorem}
\begin{proof}
  By induction on the derivation of $\ProcessP \lred[\Label]
  \ProcessQ$ and by cases on the last rule applied. We only show a few
  interesting cases; the others are either similar or simpler.

\rproofcase{r-comm}
Then $\ProcessP = \send{\ExpressionE_1}\ExpressionF \parop
\receive{\ExpressionE_2}\var\ProcessR$ and $\ExpressionE_i \eval
\Channel$ for every $i=1,2$ and $\ExpressionF \eval \Value$ and
$\Label = \Channel$ and $\ProcessQ = \ProcessR\subst\Value\var$.
From \refrule{t-par} we deduce $\Env = \Env_1 \tand \Env_2$ where
$\wtp{\Env_1}{\send{\ExpressionE_1}\ExpressionF}$ and
$\wtp{\Env_2}{\receive{\ExpressionE_2}\var\ProcessR}$.
From \refrule{t-out} we deduce $\Env_1 = \Env_{11} + \Env_{12}$ and
$\wte{\Env_{11}}{\ExpressionE_1}{\tchan\TypeT{2\Use_1}{1+\Use_2}}$ and
$\wte{\Env_{12}}{\ExpressionF}{\TypeT}$.
From \refrule{t-in} we deduce $\Env_2 = \Env_{21} + \Env_{22}$ and
$\wte{\Env_{21}}{\ExpressionE_2}{\tchan\TypeS{1+\Use_3}{2\Use_4}}$ and
$\wtp{\Env_{22}, \var : \TypeS}{\ProcessR}$.
From Lemma~\ref{lem:sr_expr} we have
$\wte{\Env_{11}}\Channel{\tchan\TypeT{2\Use_1}{1+\Use_2}}$ and
$\wte{\Env_{12}}\Value\Type$ and
$\wte{\Env_{21}}\Channel{\tchan\TypeS{1+\Use_3}{2\Use_4}}$.
Also, since $\Env_{11} + \Env_{21}$ is defined, it must be the case
that $\TypeT = \TypeS$.
Note that $1 + \Use_2 = 1 + 2\Use_2$ and $1 + \Use_3 = 1 + 2\Use_3$.
Hence, from \refrule{t-name} we deduce that $\Env_{11} = \Env_{11}',
\Channel : \tchan\Type{2\Use_1}{1+\Use_2} = (\Env_{11}', \Channel :
\tchan\Type{2\Use_1}{2\Use_2}) + \Channel : \tchan\Type01$ and
$\Env_{21} = \Env_{21}', \Channel : \tchan\Type{1+\Use_3}{2\Use_4} =
(\Env_{21}', \Channel : \tchan\Type{2\Use_3}{2\Use_4}) + \Channel :
\tchan\Type10$ for some unlimited $\Env_{11}'$ and $\Env_{21}'$.
Let $\Env_{11}'' \eqdef \Env_{11}', \Channel :
\tchan\Type{2\Use_1}{2\Use_2}$ and $\Env_{12}'' \eqdef \Env_{21}',
\Channel : \tchan\Type{2\Use_3}{2\Use_4}$ and observe that
$\Env_{11}''$ and $\Env_{21}''$ are also unlimited.
From Lemma~\ref{lem:subst} we deduce $\wtp{\Env_{12} +
  \Env_{22}}{\ProcessR\subst\Value\var}$.
Take $\Env' = \Env_{11}'' + \Env_{12} + \Env_{21}'' + \Env_{22}$.
From Lemma~\ref{lem:weak} we deduce $\wtp{\Env'}\ProcessQ$ and we
conclude by observing that $\Env \lred[\Channel] \Env'$ thanks to
Lemma~\ref{lem:env}.

\rproofcase{r-case}
Then $\ProcessP = \CaseShort\Expression{i}\var\Process$ and
$\Expression \eval \Inject{k}\Value$ for some $k\in\set{\Left,\Right}$
and $\Label = \tau$ and $\ProcessQ = \Process_k\subst\Value{\var_k}$.
From \refrule{t-case} we deduce that $\Env = \Env_1 + \Env_2$ and
$\wte{\Env_1}{\Expression}{\Type_{\Left} \oplus \Type_{\Right}}$ and
$\wtp{\Env_2, \var : \Type_k}{\Process_k}$.
From Lemma~\ref{lem:sr_expr} and either \refrule{t-inl} or
\refrule{t-inr} we deduce $\wte{\Env_1}{\Value}{\Type_k}$.
We conclude $\wtp{\Env}{\ProcessP_k\subst\Value{\var_k}}$ by
Lemma~\ref{lem:subst}.

\rproofcase{r-par}
Then $\ProcessP = \ProcessP_1 \parop \ProcessP_2$ and $\ProcessP_1
\lred[\Label] \ProcessP_1'$ and $\ProcessQ = \ProcessP_1' \parop
\ProcessP_2$.
From \refrule{t-par} we deduce $\Env = \Env_1 \tand \Env_2$ and
$\wtp{\Env_i}{\Process_i}$.
By induction hypothesis we deduce $\wtp{\Env_1'}{\Process_1'}$ for
some $\Env_1'$ such that $\Env_1 \lred[\Label] \Env_1'$.
By Proposition~\ref{lem:env} we deduce that $\Env \lred[\Label]
\Env_1' \tand \Env_2$.
We conclude $\wtp{\Env'}{\ProcessQ}$ by taking $\Env' = \Env_1' \tand
\Env_2$.
\end{proof}

\section{Supplement to Section~\ref{sec:generator}}
\label{sec:extra_generator}

First of all we prove two technical lemmas that explain the relationship
between the operators $\combineop$ and $\mergeop$ used by the constraint
generation rules (Table~\ref{tab:generation}) and type environment combination
$+$ and equality used in the type rules (Table~\ref{tab:type_system}).

\begin{lemma}
  \label{lem:combine}
  If $\combinenv{\EnvX_1}{\EnvX_2}{\EnvX}{\Constraints}$ and $\Solution$ is a
  solution for $\Constraints$ covering $\EnvX$, then $\Solution\EnvX =
  \Solution\EnvX_1 \combine \Solution\EnvX_2$.
\end{lemma}
\begin{proof}
  By induction on the derivation of
  $\combinenv{\EnvX_1}{\EnvX_2}{\EnvX}{\Constraints}$ and by cases on
  the last rule applied. We have two cases:
  \proofcase{$\dom(\EnvX_1) \cap \dom(\EnvX_2) = \emptyset$}
  Then $\EnvX = \EnvX_1, \EnvX_2$ and we conclude $\Solution\EnvX =
  \Solution\EnvX_1,\Solution\EnvX_2 = \Solution\EnvX_1 \combine
  \Solution\EnvX_2$.

  \proofcase{$\EnvX_1 = \EnvX_1', \Name : \TypeExprT$ and $\EnvX_2 =
    \EnvX_2', \Name : \TypeExprS$}
  Then $\combinenv{\EnvX_1'}{\EnvX_2'}{\EnvX'}{\Constraints'}$ and
  $\EnvX = \EnvX', \Name : \tvar$ and $\Constraints = \Constraints'
  \cup \set{ \tvar \ceq \TypeExprT + \TypeExprS }$ for some $\tvar$.
  Since $\Solution$ is a solution for $\Constraints$, we deduce
  $\Solution(\tvar) = \Solution\TypeExprT + \Solution\TypeExprS$.
  By induction hypothesis we deduce $\Solution\EnvX' =
  \Solution\EnvX_1' \combine \Solution\EnvX_2'$.
  We conclude $\Solution\EnvX = \Solution\EnvX',
  \Name:\Solution(\tvar) = \Solution\EnvX', \Name:\Solution\TypeExprT
  + \Solution\TypeExprS = (\Solution\EnvX_1' \combine
  \Solution\EnvX_2'), \Name:\Solution\TypeExprT + \Solution\TypeExprS
  = \Solution\EnvX_1 \combine \Solution\EnvX_2$.
\end{proof}

\begin{lemma}
\label{lem:merge}
If $\mergenv{\EnvX_1}{\EnvX_2}{\EnvX}{\Constraints}$ and $\Solution$ is a
solution for $\Constraints$ covering $\EnvX$, then $\Solution\EnvX =
\Solution\EnvX_1 = \Solution\EnvX_2$.
\end{lemma}
\begin{proof}
  Straightforward consequence of the definition of
  $\mergenv{\EnvX_1}{\EnvX_2}{\EnvX}{\Constraints}$.
\end{proof}

The correctness of constraint generation is proved by the next two results.

\begin{lemma}
\label{lem:correctness_expressions}
If $\rte\Expression\TypeExpr\EnvX\Constraints$ and $\Solution$ is a solution
for $\Constraints$ covering $\EnvX$, then
$\wte{\Solution\EnvX}\Expression{\Solution\TypeExpr}$.
\end{lemma}
\begin{proof}
  By induction on the derivation of
  $\rte\Expression\TypeExpr\EnvX\Constraints$ and by cases on the last
  rule applied.
  We only show two significant cases.

\rproofcase{i-name}
Then $\Expression = \Name$ and $\TypeExpr = \tvar$ fresh and $\EnvX =
\Name : \tvar$ and $\Constraints = \emptyset$.
We have $\Solution\EnvX = \Name : \Solution(\tvar)$ and
$\Solution\TypeExpr = \Solution(\tvar)$, hence we conclude
$\wte{\Solution\EnvX}\Expression{\Solution\TypeExpr}$.

\rproofcase{i-pair}
Then $\Expression = \Pair{\Expression_1}{\Expression_2}$ and
$\TypeExpr = \TypeExpr_1 \times \TypeExpr_2$ and $\Constraints =
\Constraints_1 \cup \Constraints_2 \cup \Constraints_3$ where
$\combinenv{\EnvX_1}{\EnvX_2}{\EnvX}{\Constraints_3}$ and
$\rte{\Expression_i}{\TypeExpr_i}{\EnvX_i}{\Constraints_i}$ for
$i=1,2$.
We know that $\Solution$ is a solution for $\Constraints_i$ for all
$i=1,2,3$.
By induction hypothesis we deduce
$\wte{\Solution\EnvX_i}{\Expression}{\Solution\TypeExpr_i}$ for
$i=1,2$.
From Lemma~\ref{lem:combine} we obtain $\Solution\EnvX =
\Solution\EnvX_1 \combine \Solution\EnvX_2$.
We conclude with an application of \refrule{t-pair}.
\end{proof}

\begin{reptheorem}{thm:correctness_processes}
  If $\rtp\Process\EnvX\Constraints$ and $\Solution$ is a solution for
  $\Constraints$ that covers $\EnvX$, then $\wtp{\Solution\EnvX}\Process$.
\end{reptheorem}
\begin{proof}
  By induction on the derivation of $\rtp\Process\EnvX\Constraints$ and
  by cases on the last rule applied.

\rproofcase{i-idle}
Then $\Process = \idle$ and $\EnvX = \EmptyEnv$ and $\Constraints =
\emptyset$. We conclude with an application of \refrule{t-idle}.

\rproofcase{i-in}
Then $\ProcessP = \receive\Expression\var\ProcessQ$ and
$\rte\Expression\TypeExprT{\EnvX_1}{\Constraints_1}$ and
$\rtp\ProcessQ{\EnvX_2, \var : \TypeExprS}{\Constraints_2}$ and
$\combinenv{\EnvX_1}{\EnvX_2}{\EnvX}{\Constraints_3}$ and
$\Constraints = \Constraints_1 \cup \Constraints_2 \cup \Constraints_3
\cup \set{ \TypeExprT \ceq \tchan\TypeExprS{1+\uvar_1}{2\uvar_2}}$.
By Lemma~\ref{lem:correctness_expressions} we deduce
$\wte{\Solution\EnvX_1}{\Expression}{\Solution\TypeExprT}$.
By induction hypothesis we deduce $\wtp{\Solution\EnvX_2, \var :
  \Solution\TypeExprS}{\ProcessQ}$.
By Lemma~\ref{lem:combine} we deduce $\Solution\EnvX =
\Solution\EnvX_1 + \Solution\EnvX_2$.
From the hypothesis that $\Solution$ is a solution for $\Constraints$
we know $\Solution\TypeExpr =
\tchan{\Solution\TypeExprS}{1+\Solution(\uvar_1)}{2\Solution(\uvar_2)}$.
We conclude with an application of \refrule{t-in}.

\rproofcase{i-out}
Then $\ProcessP = \send\ExpressionE\ExpressionF$ and
$\rte\ExpressionE\TypeExprT{\EnvX_1}{\Constraints_1}$ and
$\rte\ExpressionF\TypeExprS{\EnvX_2}{\Constraints_2}$ and
$\combinenv{\EnvX_1}{\EnvX_2}{\EnvX}{\Constraints_3}$ and
$\Constraints = \Constraints_1 \cup \Constraints_2 \cup \Constraints_3
\cup \set{ \TypeExprT \ceq \tchan\TypeExprS{2\uvar_1}{1+\uvar_2}}$.
By Lemma~\ref{lem:correctness_expressions} we deduce
$\wte{\Solution\EnvX_1}{\ExpressionE}{\Solution\TypeExprT}$ and
$\wte{\Solution\EnvX_2}{\ExpressionF}{\Solution\TypeExprS}$.
By Lemma~\ref{lem:combine} we deduce $\Solution\EnvX =
\Solution\EnvX_1 + \Solution\EnvX_2$.
From the hypothesis that $\Solution$ is a solution for $\Constraints$
we know $\Solution\TypeExprT =
\tchan{\Solution\TypeExprS}{2\Solution(\uvar_1)}{1+\Solution(\uvar_2)}$.
We conclude with an application of \refrule{t-out}.

\rproofcase{i-par}
Then $\ProcessP = \Process_1 \parop \Process_2$ and
$\rtp{\Process_i}{\EnvX_i}{\Constraints_i}$ for $i=1,2$ and
$\combinenv{\EnvX_1}{\EnvX_2}{\EnvX}{\Constraints_3}$ and $\Constraints =
\Constraints_1 \cup \Constraints_2 \cup \Constraints_3$.
By induction hypothesis we deduce $\wtp{\Solution\EnvX_i}{\Process_i}$
for $i=1,2$.
By Lemma~\ref{lem:combine} we deduce $\Solution\EnvX =
\Solution\EnvX_1 \combine \Solution\EnvX_2$.
We conclude with an application of \refrule{t-par}.

\rproofcase{i-rep}
Then $\ProcessP = \bang\ProcessQ$ and
$\rtp\ProcessQ{\EnvX'}{\Constraints_1}$ and
$\combinenv{\EnvX'}{\EnvX'}{\EnvX}{\Constraints_2}$ and $\Constraints
= \Constraints_1 \cup \Constraints_2$.
By induction hypothesis we deduce $\wtp{\Solution\EnvX'}\ProcessQ$.
By Lemma~\ref{lem:combine} we deduce $\Solution\EnvX = \Solution\EnvX'
+ \Solution\EnvX'$.
By Definition~\ref{def:un} we know that $\unlimited{\Solution\EnvX}$
holds. Furthermore, $\Solution\EnvX' + \Solution\EnvX$ is defined.
By Lemma~\ref{lem:weak} and Definition~\ref{def:tand} we deduce
$\wtp{\Solution\EnvX}\ProcessQ$.
We conclude with an application of \refrule{t-rep}.\Luca{Bisognerebbe
  discutere cosa succede se si riformula \refrule{i-rep} con i vincoli
  unlimited. Rispecchia pi\`u fedelmente il type system, ma produce
  usi meno precisi.}

\rproofcase{i-new}
Then $\ProcessP = \new\Channel\ProcessQ$ and
$\rtp\ProcessQ{\EnvX,\Channel:\TypeExpr}{\Constraints'}$ and
$\Constraints = \Constraints' \cup \set{ \TypeExpr \ceq
  \tchan\tvar\uvar\uvar }$.
By induction hypothesis we deduce $\wtp{\Solution\EnvX,
  \Channel:\Solution\TypeExpr}\ProcessQ$.
Since $\Solution$ is a solution for $\Constraints'$ we know that
$\Solution\TypeExpr =
\tchan{\Solution(\tvar)}{\Solution(\uvar)}{\Solution(\uvar)}$.
We conclude with an application of \refrule{t-new}.

\rproofcase{i-case}
Then $\ProcessP = \CaseShort\Expression{i}{\var}{\Process}$ and
$\rte{\Expression}{\Type}{\EnvX_1}{\Constraints_1}$ and
$\rtp{\Process_i}{\EnvX_i,\var_i:\TypeExpr_i}{\Constraints_i}$ for
$i=\Left,\Right$ and
$\mergenv{\EnvX_{\Left}}{\EnvX_{\Right}}{\EnvX_2}{\Constraints_2}$ and
$\combinenv{\EnvX_1}{\EnvX_2}{\EnvX}{\Constraints_3}$ and
$\Constraints = \Constraints_1 \cup \Constraints_2 \cup \Constraints_3
\cup \Constraints_{\Left} \cup \Constraints_{\Right} \cup \set{
  \TypeExpr \ceq \TypeExpr_{\Left} \tsum \TypeExpr_{\Right}}$.
By Lemma~\ref{lem:correctness_expressions} we deduce
$\wte{\Solution\EnvX_1}{\Expression}{\Solution\TypeExpr}$.
By induction hypothesis we deduce $\wtp{\Solution\EnvX_i}{\Process_i}$
for $i=\Left,\Right$.
By Lemma~\ref{lem:merge} we deduce $\Solution\EnvX_{\Left} =
\Solution\EnvX_{\Right} = \Solution\EnvX_2$.
By Lemma~\ref{lem:combine} we deduce $\Solution\EnvX =
\Solution\EnvX_1 + \Solution\EnvX_2$.
Since $\Solution$ is a solution for $\Constraints$, we have
$\Solution\TypeExpr = \Solution\TypeExpr_{\Left} \tsum
\Solution\TypeExpr_{\Right}$.
We conclude with an application of \refrule{t-case}.

\rproofcase{i-weak}
Then $\EnvX = \EnvX', \Name : \tvar$ and $\Constraints = \Constraints'
\cup \set{ \unlimited\tvar }$ where $\tvar$ is fresh and
$\rtp\ProcessP{\EnvX'}{\Constraints'}$.
By induction hypothesis we deduce $\wtp{\Solution\EnvX'}\Process$.
Since $\Solution$ is a solution for $\Constraints'$ we know that
$\unlimited{\Solution(\tvar)}$ holds.
Since $\Name \not\in \dom(\EnvX')$ we know that $\EnvX'\Solution +
\Name : \Solution(\tvar)$ is defined.
By Lemma~\ref{lem:weak}(2) we conclude $\wtp{\Solution\EnvX', \Name :
  \Solution(\tvar)}{\ProcessP}$.
\end{proof}

The next lemma relates once more $\combineop$ and type environment combination
$+$. It is, in a sense, the inverse of Lemma~\ref{lem:combine}.

\begin{lemma}
  \label{lem:combine_inverse}
  If $\Solution\EnvX_1 + \Solution\EnvX_2$ is defined, then there exist
  $\EnvX$, $\Constraints$, and $\Solution' \supseteq \Solution$ such that
  $\combinenv{\EnvX_1}{\EnvX_2}{\EnvX}{\Constraints}$ and $\Solution'$ is a
  solution for $\Constraints$ that covers $\EnvX$.
\end{lemma}
\begin{proof}
  By induction on the maximum size of $\EnvX_1$ and $\EnvX_2$. We
  distinguish two cases.

  \proofcase{$\dom(\EnvX_1) \cap \dom(\EnvX_2) = \emptyset$}
  We conclude by taking $\EnvX \eqdef \EnvX_1,\EnvX_2$ and $\Constraints
  \eqdef \emptyset$ and $\Solution' \eqdef \Solution$ and observing that
  $\combinenv{\EnvX_1}{\EnvX_2}{\EnvX}{\emptyset}$.

  \proofcase{$\EnvX_1 = \EnvX_1', \Name : \TypeExprT$ and $\EnvX_2 =
    \EnvX_2', \Name : \TypeExprS$}
  Since $\Solution\EnvX_1 + \Solution\EnvX_2$ is defined, we know that
  $\Solution\EnvX_1' + \Solution\EnvX_2'$ is defined as well and
  furthermore that $(\Solution\EnvX_1 + \Solution\EnvX_2)(\Name) =
  \Solution\TypeExprT + \Solution\TypeExprS$.
  By induction hypothesis we deduce that there exist $\EnvX'$,
  $\Constraints'$, and $\Solution'' \supseteq \Solution$ such that
  $\combinenv{\EnvX_1'}{\EnvX_2'}{\EnvX'}{\Constraints'}$ and $\Solution''$ is
  a solution for $\Constraints'$ that covers $\EnvX'$.
  Take $\EnvX \eqdef \EnvX', \Name : \tvar$ where $\tvar$ is fresh,
  $\Constraints \eqdef \Constraints' \cup \set{ \tvar \ceq \TypeExprT
    + \TypeExprS }$ and $\Solution' \eqdef \Solution'' \cup \set{
    \tvar \mapsto \Solution\TypeExprT + \Solution\TypeExprS }$.
  We conclude observing that $\Solution'$ is a solution for $\Constraints$
  that covers $\EnvX$.
\end{proof}

In order to prove the completeness of type reconstruction for
expressions, we extend the reconstruction algorithm with one more
weakening rule for expressions:
\[
\inferrule[\defrule{i-weak expr}]{
  \rte\Expression\TypeExpr\EnvX\Constraints
}{
  \rte\Expression\TypeExpr{\EnvX, \Name:\tvar}{\Constraints \cup \unlimited\tvar}
}
\]
This rule is unnecessary as far as completeness is concerned, because
there is already a weakening rule \refrule{i-weak} for processes that
can be used to subsume it. However, \refrule{i-weak expr} simplifies
both the proofs and the statements of the results that follow.

\begin{lemma}
\label{lem:completeness_expressions}
If $\wte\Env\Expression\Type$, then there exist $\TypeExpr$, $\EnvX$,
$\Constraints$, and $\Solution$ such that
$\rte\Expression\TypeExpr\EnvX\Constraints$ and $\Solution$ is a
solution for $\Constraints$ and $\Env = \Solution\EnvX$ and $\Type =
\Solution\TypeExpr$.
\end{lemma}
\begin{proof}
  By induction on the derivation of $\wte\Env\Expression\Type$ and by
  cases on the last rule applied. We only show two representative
  cases.

\rproofcase{t-name}
Then $\Expression = \Name$ and $\Env = \Env', \Name : \Type$ and
$\unlimited{\Env'}$. Let $\Env' = \set{ \Name_i : \Type_i }_{i\in I}$.
Take $\TypeExpr \eqdef \tvar$ and $\EnvX \eqdef \set{ \Name_i :
  \tvar_i }_{i\in I}, \Name : \tvar$ and $\Constraints \eqdef \set{
  \unlimited{\tvar_i} \mid i \in I }$ and $\Solution \eqdef \set{
  \tvar_i \mapsto \Type_i }_{i\in I} \cup \set{ \tvar \mapsto \Type}$
where $\tvar$ and the $\tvar_i$'s are all fresh type variables.
Observe that $\rte\Expression\TypeExpr\EnvX\Constraints$ by means of
one application of \refrule{i-name} and as many applications of
\refrule{i-weak expr} as the cardinality of $I$.
We conclude observing that $\Solution$ is a solution for $\Constraints$ and
$\Env = \Solution\EnvX$ and $\Type = \Solution\TypeExpr$ by definition of
$\Solution$.

\rproofcase{t-pair}
Then $\Expression = \Pair{\Expression_1}{\Expression_2}$ and $\Env =
\Env_1 + \Env_2$ and $\Type = \Type_1 \times \Type_2$ and
$\wte{\Env_i}{\Expression_i}{\Type_i}$ for $i=1,2$.
By induction hypothesis we deduce that there exist $\TypeExpr_i$,
$\EnvX_i$, $\Constraints_i$, and $\Solution_i$ solution for
$\Constraints_i$ such that
$\rte{\Expression_i}{\TypeExpr_i}{\EnvX_i}{\Constraints_i}$ and
$\Env_i = \Solution_i\EnvX_i$ and $\Type_i = \Solution_i\TypeExpr_i$
for $i=1,2$.
Since the reconstruction algorithm always chooses fresh type
variables, we also know that $\dom(\Solution_1) \cap \dom(\Solution_2)
= \emptyset$.
Take $\Solution' \eqdef \Solution_1 \cup \Solution_2$. We have that
$\Solution'\EnvX_1 + \Solution'\EnvX_2 = \Env_1 + \Env_2$ is
defined. Therefore, by Lemma~\ref{lem:combine_inverse}, we deduce that
there exist $\EnvX$, $\Constraints_3$, and $\Solution \supseteq
\Solution'$ such that
$\combinenv{\EnvX_1}{\EnvX_2}\EnvX{\Constraints_3}$ and $\Solution$
is a solution for $\Constraints$ that covers $\EnvX$.
We conclude with an application of \refrule{i-pair} and taking
$\TypeExpr \eqdef \TypeExpr_1 \times \TypeExpr_2$ and $\Constraints
\eqdef \Constraints_1 \cup \Constraints_2 \cup \Constraints_3$.
\end{proof}

\begin{reptheorem}{thm:completeness_processes}
  \revised{If $\wtp{\Env}{\ProcessP}$, then there exist $\EnvX$,
    $\Constraints$, and $\Solution$ such that
    $\rtp\ProcessP{\EnvX}{\Constraints}$ and $\Solution$ is a solution
    for $\Constraints$ that covers $\EnvX$ and $\Env =
    \Solution\EnvX$.}
\end{reptheorem}
\begin{proof}
  By induction on the derivation of $\wtp{\Env}{\ProcessP}$ and by cases on
  the last rule applied. We only show a few cases, the others being analogous.

\rproofcase{t-idle}
Then $\Process = \idle$ and $\unlimited\Env$. Let $\Env = \set{
  \Name_i : \Type_i }_{i\in I}$.
Take $\EnvX \eqdef \set{ \Name_i : \tvar_i }_{i\in I}$ and
$\Constraints \eqdef \set{ \unlimited{\tvar_i} }_{i\in I}$ and
$\Solution \eqdef \set{ \tvar_i \mapsto \Type_i }_{i\in I}$ where the
$\tvar_i$'s are all fresh type variables.
By repeated applications of \refrule{i-weak} and one application of
\refrule{i-idle} we derive $\rtp\idle\EnvX\Constraints$.
We conclude observing that $\Solution$ is a solution for $\Constraints$ and
$\Env = \Solution\EnvX$.

\rproofcase{t-in}
Then $\Process = \receive\Expression\var\ProcessQ$ and $\Env = \Env_1
+ \Env_2$ and
$\wte{\Env_1}\Expression{\tchan\Type{1+\Use_1}{2\Use_2}}$ and
$\wtp{\Env_2,\var : \Type}{\ProcessQ}$.
By Lemma~\ref{lem:completeness_expressions} we deduce that there exist
$\TypeExpr$, $\EnvX_1$, $\Constraints_1$, and $\Solution_1$ solution
for $\Constraints_1$ such that
$\rte\Expression\TypeExpr{\EnvX_1}{\Constraints_1}$ and $\Env_1 =
\Solution_1\EnvX_1$ and $\tchan\Type{1+\Use_1}{2\Use_2} =
\Solution_1\TypeExpr$.
By induction hypothesis we deduce that there exist $\EnvX_2'$,
$\Constraints_2$, and $\Solution_2$ solution for $\Constraints_2$ such
that $\Env_2, \var : \Type = \Solution_2\EnvX_2'$.
Then it must be the case that $\EnvX_2' = \EnvX_2, \var : \TypeExprS$
for some $\EnvX_2$ and $\TypeExprS$ such that $\Env_2 =
\Solution_2\EnvX_2$ and $\Type = \Solution_2\TypeExprS$.
Since all type variables chosen by the type reconstruction algorithm
are fresh, we know that $\dom(\Solution_1) \cap \dom(\Solution_2) =
\emptyset$. Take $\Solution' \eqdef \Solution_1 \cup \Solution_2 \cup
\set{ \uvar_1 \mapsto \Use_1, \uvar_2 \mapsto \Use_2 }$.
Observe that $\Solution'\EnvX_1 + \Solution'\EnvX_2 = \Env_1 + \Env_2$ which
is defined. By Lemma~\ref{lem:combine_inverse} we deduce that there exist
$\EnvX$, $\Constraints_3$, and $\Solution \supseteq \Solution'$ such that
$\combinenv{\EnvX_1}{\EnvX_2}{\EnvX}{\Constraints_3}$ and $\Solution$ is a
solution for $\Constraints_3$ that covers $\EnvX$.
Take $\Constraints \eqdef \Constraints_1 \cup \Constraints_2 \cup
\Constraints_3 \cup \set{ \TypeExprT \ceq
  \tchan\TypeExprS{1+\uvar_1}{2\uvar_2}}$.
Then $\Solution$ is a solution for $\Constraints$, because
$\Solution\TypeExprT = \tchan\Type{1+\Use_1}{2\Use_2} =
\tchan{\Solution\TypeExprS}{1+\Solution(\uvar_1)}{2\Solution(\uvar_2)}
= \Solution\tchan\TypeExprS{1+\uvar_1}{2\uvar_2}$.
Also, by Lemma~\ref{lem:combine} we have $\Solution\EnvX =
\Solution\EnvX_1 + \Solution\EnvX_2 = \Env_1 + \Env_2 = \Env$.
We conclude $\rtp\ProcessP\EnvX\Constraints$ with an application of
\refrule{i-in}.

\rproofcase{t-par}
Then $\Process = \Process_1 \parop \Process_2$ and $\Env = \Env_1 +
\Env_2$ and $\wtp{\Env_i}{\Process_i}$ for $i=1,2$.
By induction hypothesis we deduce that, for every $i=1,2$, there exist
$\EnvX_i$, $\Constraints_i$, and $\Solution_i$ solution for
$\Constraints_i$ such that $\rtp{\Process_i}{\EnvX_i}{\Constraints_i}$
and $\Env_i = \Solution_i\EnvX_i$.
We also know that $\dom(\Solution_1) \cap \dom(\Solution_2) =
\emptyset$ because type/use variables are always chosen fresh.
Take $\Solution' \eqdef \Solution_1 \cup \Solution_2$.
By Lemma~\ref{lem:combine_inverse} we deduce that there exist $\EnvX$,
$\Constraints_3$, and $\Solution \supseteq \Solution'$ such that
$\combinenv{\EnvX_1}{\EnvX_2}{\EnvX}{\Constraints_3}$ and $\Solution$ is a
solution for $\Constraints_3$ that covers $\EnvX$.
By Lemma~\ref{lem:combine} we also deduce that $\Solution\EnvX =
\Solution\EnvX_1 + \Solution\EnvX_2 = \Env_1 + \Env_2 = \Env$.
We conclude by taking $\Constraints \eqdef \Constraints_1 \cup
\Constraints_2 \cup \Constraints_3$ with an application of
\refrule{i-par}.

\rproofcase{t-rep}
Then $\Process = \bang\ProcessQ$ and $\wtp\Env\ProcessQ$ and
$\unlimited\Env$.
By induction hypothesis we deduce that there exist $\EnvX'$,
$\Constraints'$, and $\Solution'$ solution for $\Constraints'$ such
that $\rtp\ProcessQ{\EnvX'}{\Constraints'}$ and $\Env =
\Solution'\EnvX'$.
Obviously $\Solution'\EnvX' + \Solution'\EnvX'$ is defined, hence by
Lemma~\ref{lem:combine_inverse} we deduce that there exist $\EnvX$,
$\Constraints''$, and $\Solution \supseteq \Solution'$ such that
$\combinenv{\EnvX'}{\EnvX'}{\EnvX}{\Constraints''}$ and $\Solution$ is
a solution for $\Constraints''$.
By Lemma~\ref{lem:combine} we deduce $\Solution\EnvX = \Solution\EnvX'
+ \Solution\EnvX' = \Env + \Env = \Env$, where the last equality
follows from the hypothesis $\unlimited\Env$ and
Definition~\ref{def:un}.
We conclude $\rtp\Process\EnvX\Constraints$ with an application of
\refrule{i-rep} by taking $\Constraints \eqdef \Constraints' \cup
\Constraints''$.
\end{proof}

\section{Supplement to Section~\ref{sec:solver}}
\label{sec:extra_solver}

\revised{Below is the derivation showing the reconstruction algorithm
  at work on the process~\eqref{eq:discuss3}.
\[
\begin{prooftree}
  \pt{
    \pt{
      \rte{
        \Channel
      }{
        \tvar_1
      }{
        \Channel : \tvar_1
      }{
        \emptyset
      }
      \quad
      \rte{
        3
      }{
        \tint
      }{
        \emptyset
      }{
        \emptyset
      }
      \justifies
      \rtp{
        \send\Channel{3}
      }{
        \Channel : \tvar_1
      }{
        \set{ \tvar_1 \ceq \tchan\tint{2\uvar_1}{1+\uvar_2} }
      }
      \using\refrule{i-out}
    }
    \quad
    \pt{
      \rte{
        \ChannelB
      }{
        \tvarB
      }{
        \ChannelB : \tvarB
      }{
        \emptyset
      }
      \quad
      \rte{
        \ChannelA
      }{
        \tvar_2
      }{
        \Channel : \tvar_2
      }{
        \emptyset
      }
      \justifies
      \rtp{
        \send\ChannelB\ChannelA
      }{
        \ChannelA : \tvar_2,
        \ChannelB : \tvarB
      }{
        \set{
          \tvarB \ceq \tchan{\tvar_2}{2\uvar_3}{1+\uvar_4}
        }
      }
      \using\refrule{i-out}
    }
    \justifies
    \rtp{
      \send\Channel{3}
      \parop
      \send\ChannelB\ChannelA
    }{
      \ChannelA : \tvarA,
      \ChannelB : \tvarB
    }{
      \set{
        \tvar \ceq \tvar_1 + \tvar_2,
        \tvar_1 \ceq \tchan\tint{2\uvar_1}{1+\uvar_2},
        \tvarB \ceq \tchan{\tvar_2}{2\uvar_3}{1+\uvar_4}
      }
    }
    \using\refrule{i-par}
  }
  \justifies
  \rtp{
    \new\Channel\ttparens{
      \send\Channel{3}
      \parop
      \send\ChannelB\ChannelA
    }
  }{
    \ChannelB : \tvarB
  }{
    \set{
      \tvar \ceq \tchan\tvarD{\uvar_5}{\uvar_5},
      \tvar \ceq \tvar_1 + \tvar_2,
      \dots
    }
  }
  \using\refrule{i-new}
\end{prooftree}
\]
\noindent
Below is the derivation showing the reconstruction algorithm at work
on the process~\eqref{eq:discuss4}. Only the relevant differences with
respect to the derivation above are shown.
\[
\begin{prooftree}
  \pt{
    \vdots
    \quad
    \pt{
      \pt{
        \vdots
        \qquad
        \vdots
        \justifies
        \rtp{
          \send\ChannelB\ChannelA
        }{
          \ChannelA : \tvarA_2,
          \ChannelB : \tvarB
        }{
          \set{
            \tvarB \ceq \tchan{\tvar_2}{2\uvar_3}{1+\uvar_4}
          }
        }
      }
      \quad
      \pt{
        \vdots
        \qquad
        \vdots
        \justifies
        \rtp{
          \send\ChannelC\ChannelA
        }{
          \ChannelA : \tvarA_3,
          \ChannelC : \tvarC
        }{
          \set{
            \tvarC \ceq \tchan{\tvar_3}{2\uvar_5}{1+\uvar_6}
          }
        }
      }
      \justifies
      \rtp{
        \send\ChannelB\ChannelA
        \parop
        \send\ChannelC\ChannelA
      }{
        \ChannelA : \tvar_{23},
        \ChannelB : \tvarB,
        \ChannelC : \tvarC
      }{
        \set{
          \tvar_{23} \ceq \tvar_2 + \tvar_3,
          \dots
        }
      }
      \using\refrule{i-par}
    }
    \justifies
    \rtp{
      \send\Channel{3}
      \parop
      \send\ChannelB\ChannelA
      \parop
      \send\ChannelC\ChannelA
    }{
      \ChannelA : \tvarA,
      \ChannelB : \tvarB,
      \ChannelC : \tvarC
    }{
      \set{
        \tvar \ceq \tvar_1 + \tvar_{23},
        \dots
      }
    }
    \using\refrule{i-par}
  }
  \justifies
  \rtp{
    \new\Channel\ttparens{
      \send\Channel{3}
      \parop
      \send\ChannelB\ChannelA
      \parop
      \send\ChannelC\ChannelA
    }
  }{
    \ChannelB : \tvarB,
    \ChannelC : \tvarC
  }{
    \set{
      \tvar \ceq \tchan\tvarD{\uvar_5}{\uvar_5},
      \tvar \ceq \tvar_1 + \tvar_{23},
      \dots
    }
  }
  \using\refrule{i-new}
\end{prooftree}
\]
\noindent
Below is the derivation showing the reconstruction algorithm at work
on the process~\eqref{eq:discuss5}.
\[
\begin{prooftree}
  \pt{
    \rte{
      \ChannelB
    }{
      \tvarB
    }{
      \ChannelB : \tvarB
    }{
      \emptyset
    }
    \qquad
    \rte{
      \var
    }{
      \tvarC
    }{
      \var : \tvarC
    }{
      \emptyset
    }
    \justifies
    \rtp{
      \send\ChannelB\var
    }{
      \ChannelB : \tvarB,
      \var : \tvarC
    }{
      \set{
        \tvarB \ceq \tchan\tvarC{2\uvar_1}{1+\uvar_2}
      }
    }
    \using\refrule{i-out}
  }
  \justifies
  \rtp{
    \receive\ChannelA\var
    \send\ChannelB\var
  }{
    \ChannelA : \tvarA,
    \ChannelB : \tvarB
  }{
    \set{
      \tvarA \ceq \tchan\tvarC{1+\uvar_3}{2\uvar_4},
      \tvarB \ceq \tchan\tvarC{2\uvar_1}{1+\uvar_2}
    }
  }
  \using\refrule{i-in}
\end{prooftree}
\]
}

\end{document}